\documentclass[11pt]{llncs}

\usepackage{makeidx}  % allows for indexgeneration
\usepackage{amssymb}  % \mathbb
\usepackage[dvips]{graphicx}
\usepackage{algorithmic,algorithm} % new entry 
\usepackage{amsmath}
\usepackage{comment}
\usepackage{subfigure}
\usepackage{multirow}

\usepackage{epsfig}

\usepackage{float}

\newcommand{\shortqed}{\hfill \mbox{$\blacksquare$} \smallskip}

\newcommand{\Real}{\mathbb{R}}

\makeatletter
  
  \@addtoreset{algorithm}{section}
\makeatother

\usepackage{comment}	% required for `\comment' (yatex added)
\begin{document} 

\title{Plane Formation by Synchronous Mobile Robots in 
the Three Dimensional Euclidean Space}
\author{Yukiko Yamauchi \thanks{Corresponding author. 
Address: 744 Motooka, Nishi-ku, Fukuoka 819-0395, Japan. 
Fax: +81-92-802-3637. Email: \texttt{yamauchi@inf.kyushu-u.ac.jp}}, 
Taichi Uehara, 
Shuji Kijima, 
\and 
Masafumi Yamashita}
\institute{
Kyushu University, Japan. 
}

\maketitle 

\begin{abstract}
Creating a swarm of mobile computing entities frequently called
robots, agents or sensor nodes, 
with self-organization ability is a contemporary challenge 
in distributed computing.
Motivated by this,
we investigate the {\em plane formation problem} that 
requires a swarm of robots moving in the three dimensional Euclidean 
space to land on a common plane. 
The robots are fully synchronous and endowed with visual perception.
But they do not have identifiers,
nor access to the global coordinate system,
% any means of explicit communication with each other, 
% nor memory of past. 
nor any means of explicit communication with each other. 
Though there are plenty of results on the agreement problem 
for robots in the two dimensional plane, for example, 
the point formation problem, the pattern formation problem, and so on, 
this is the first result for robots in the {\em three dimensional space}. 
This paper presents a necessary and sufficient condition
for fully-synchronous robots to solve the 
plane formation problem that does not depend on obliviousness 
i.e., the availability of local memory at robots. 
An implication of the result is somewhat counter-intuitive:
The robots {\em cannot} form a plane from most of the semi-regular 
polyhedra,
while they {\em can} form a plane from every regular polyhedron 
(except a regular icosahedron), 
whose symmetry is usually considered to be higher than any semi-regular
polyhedrdon. 
\end{abstract}

\noindent{\bf Keywords.}
symmetry breaking, 
mobile robots, plane formation, rotation group.

%========================================================================
\section{Introduction}
\label{sec:intro}

Self-organization in a swarm of mobile computing entities 
frequently called robots, agents or sensor nodes, 
has gained much attention 
as sensing and controlling devices are developed and become cheaper. 
It is expected that mobile robot systems perform patrolling, 
sensing, and exploring in a harsh environment such as disaster area, 
deep sea, and space without any human intervention.
Theoretical aspect of such mobile robot systems in the
{\em two dimensional} Euclidean space (2D-space or plane) 
attracts much attention and distributed control of mobile robots
with very weak capabilities has been
investigated~\cite{AP06,AOSY99,BGT10,CFPS12,CP05,CP08,DFPSY16,DFSY15,EP09,FPS12,FPSV14,FPSW05,FPSW08,FYOKY15,IKY14,ISKI12,P05,SDY09,SY99,YK96,YS10,YY13,YY14}. 
The robots are anonymous, oblivious (memory-less),
have neither access to the global coordinate system nor
explicit communication medium. 
For robots moving in the {\em three dimensional} Euclidean space (3D-space),
we first investigate the {\em plane formation problem},
which is a fundamental self-organization problem that requires robots 
to occupy distinct positions on a common plane 
without making any multiplicity,   
mainly motivated by an obvious observation that 
robots on a plane would be easier to control than those moving 
in 3D-space.

In this paper, a robot is anonymous and is represented by a point 
in 3D-space. A robot repeats executing
a ``Look-Compute-Move'' cycle, during which, 
it observes, in a {\em Look phase}, 
the positions of all robots by taking a snapshot, 
which we call a {\em local observation} in this paper,
computes the next position based on a given deterministic algorithm 
in a {\em Compute phase},
and moves to the next position in a {\em Move phase}.
This definition of Look-Compute-Move cycle implies that
it has {\em full vision}, i.e., the vision is unrestricted 
% the algorithm is {\em oblivious},
% i.e., it does not depend on a snapshot of the past,
and the move is atomic, 
i.e., each robot does not stop en route to the next position
and we do not care which route it takes. 
A robot is {\em oblivious} if in a Compute phase, 
it uses only the snapshot just taken in the preceding Look phase, 
i.e., the output of the algorithm depends neither on a snapshot
nor computation of the past cycles.
Otherwise, a robot is {\em non-oblivious}. 
A robot has no access to the global $x$-$y$-$z$ coordinate system 
and all actions are done in terms of its local $x$-$y$-$z$
coordinate system. 
We assume that all local coordinate systems are right-handed.
A {\em configuration} of such robot system is a set of
points observed in the global coordinate system. 
Each robot obtains a configuration translated with its local
coordinate system in a Look phase. 

The robots can see each other,
but do not have direct communication capabilities; 
communication among robots must take place 
solely by moving and observing robots' positions 
with tolerating possible inconsistency among the local coordinate systems.
The robots are {\em anonymous}; 
they have no unique identifiers and are indistinguishable by their looks 
and execute the same algorithm. 
Finally, they are fully synchronous (FSYNC);
they all start the $i$-th Look-Compute-Move cycle simultaneously
and synchronously execute each of its Look, Compute, and Move phases.

The purpose of this paper is to show a necessary and sufficient
condition for the robots to solve the plane formation
problem.\footnote{Because multiplicity is not
allowed, gathering at one point (i.e., point formation) is not
a solution for the plane formation problem.} 
The {\em line formation problem} in 2D-space 
is the counterpart of the plane formation problem in 3D-space
and is {\em unsolvable} from an initial configuration $P$ 
if $P$ is a regular polygon (i.e., the robots occupy the vertices of
a regular polygon),
intuitively because anonymous robots forming a regular polygon 
cannot break symmetry among themselves 
and lines they propose are also symmetric, 
so that they cannot agree on one line from them~\cite{SY99}.
Hence symmetry breaking among robots would play a crucial role
in the plane formation problem. 

The {\em pattern formation problem} requires robots 
to form a target pattern from an initial configuration 
and our plane formation problem is a subproblem of 
the pattern formation problem in 3D-space.
To investigate the pattern formation problem in 2D-space,
which contains the line formation problem as a subproblem,
Suzuki and Yamashita~\cite{SY99} used the concept of {\em symmetricity} 
to measure the degree of symmetry of a configuration 
in 2D-space.\footnote{
The symmetricity was originally introduced in~\cite{YK96} 
for anonymous networks to investigate the solvability 
of some agreement problems.}
Let $P$ be a configuration.
Then its symmetricity $\rho(P)$ is the order of the cyclic group of $P$, 
with the rotation center $o$ being the center of the smallest enclosing
circle of $P$, if $o \not\in P$.
That is, $\rho(P)$ is the number of angles such that
rotating $P$ by $\theta$ 
($\theta \in [0, 2 \pi)$) around $o$ produces $P$ itself,
which intuitively means that the $\rho(P)$ robots forming 
a regular $\rho(P)$-gon in $P$ may not be able to break symmetry among
them.\footnote{We consider a point as a regular $1$-gon with an
arbitrary center and a set of two points as a regular $2$-gon with
the center at the midpoint.}
However, when $o \in P$, the symmetricity $\rho(P)$ is defined to be $1$ 
independently of its rotational symmetry.
This is the crucial difference between the cyclic group 
and the symmetricity that reflects the fact that the robot at $o$ 
{\em can} translate $P$ into another configuration $P'$ with 
symmetricity $1$ by simply leaving $o$. 
The following result has been obtained~\cite{FYOKY15,SY99,YS10}:
A target pattern $F$ is formable from an initial configuration $P$,
if and only if $\rho(P)$ divides $\rho(F)$.

In this paper, based on the results in 2D-space,
we measure the symmetry among robots in 3D-space by rotation groups,
each of which is defined by a set of rotation axes and their
arrangement. 
In 3D-space, such rotation groups with finite order are classified into 
the cyclic groups, the dihedral groups, the tetrahedral group, 
the octahedral group, and the icosahedral group.
We call the cyclic groups and the dihedral groups {\em two-dimensional} (2D)
rotation groups 
in the sense that the plane formation problem is obviously solvable
from a configuration on which only a 2D rotation group acts,
since there is a single rotation axis or a principal rotation axis 
and all robots can agree on a plane perpendicular to the axis and 
containing the center of the smallest enclosing ball of the robots. 
Then the oblivious (thus, non-oblivious)
FSYNC robots can easily solve the plane formation problem
by moving onto the agreed plane.

The other three rotation groups are recognized as the groups formed by
the rotations on the corresponding regular polyhedra 
and they are also called polyhedral groups.
A regular polyhedron consists of congruent regular polygons 
and all its vertices are congruent. 
A regular polyhedron has {\em vertex-transitivity}, 
that is, there are rotations that replace any two vertices 
with keeping the polyhedron unchanged as a whole. 
For example, we can rotate a cube around any axis 
containing two opposite vertices, 
any axis containing the centers of opposite faces, 
and any axis containing the midpoints of opposite edges. 
For each regular polyhedron, 
the rotations applicable to it form a group 
and, in this way, the tetrahedral group, 
the octahedral group, and the icosahedral group 
are defined.\footnote{
There are five regular polyhedra; 
regular tetrahedron, regular cube, regular octahedron, 
regular dodecahedron, and a regular icosahedron. 
A cube and a regular octahedron are dual each other,
and so are a regular dodecahedron and a regular icosahedron.
A tetrahedron is a self-dual. 
Since the same rotations are applicable 
to a regular polyhedron and its dual,
there are three rotation groups.}
We call them {\em three-dimensional} (3D) rotation groups. 

When a 3D rotation group acts on a configuration, 
the robots are not on one plane. 
In addition, 
the vertex-transitivity among the robots may allow some of them 
to have identical local observations. 
This may result in an infinite execution,
where the robots keep symmetric movement in 3D-space
and never agree on a plane. 
A vertex-transitive set of points is obtained by 
specifying a seed point and a set of symmetry operations, 
which consists of rotations around an axis, 
reflections for a mirror plane ({\em bilateral symmetry}), 
reflections for a point ({\em central inversion}), 
and {\em rotation-reflections}~\cite{C97}. 
However, it is sufficient to consider vertex-transitive sets of points
obtained by transformations that preserve the center of 
the smallest enclosing ball of the robots 
and keep Euclidean distance and handedness, 
in other words, direct congruent transformations,
since otherwise, the robots can break the symmetry 
in a vertex-transitive set of points because all local coordinate
systems are righthanded. 
Such symmetry operations consist of rotations around some axes. 
(See \cite{C73} for more detail.) 

We define the {\em rotation group} of a configuration in 3D-space
as the rotation group that acts on the configuration, i.e.,  
a set of points. 
Let $P$ and  $\gamma(P)$ be a set of points in 3D-space 
and its rotation group, respectively. 
Then the robots are partitioned into vertex-transitive subsets 
regarding $\gamma(P)$, 
so that for each subset, 
the robots in it may have the same local observation. 
We call this decomposition {\em $\gamma(P)$-decomposition} of $P$. 
The goal of this paper is to show the following theorem:

\begin{theorem}
\label{theorem:main}
Let $P$ and $\{ P_1, P_2, \ldots, P_m \}$ be an initial configuration
and the $\gamma(P)$-decomposition of $P$, respectively. 
Then irrespective of obliviousness, 
FSYNC robots can form a plane from $P$ 
if and only if (i) $\gamma(P)$ is a 2D rotation group, or 
(ii) $\gamma(P)$ is a 3D rotation group
and there exists a subset $P_i$ such that $|P_i| \not\in \{12, 24, 60\}$. 
\end{theorem}
We can rephrase this theorem as follows: 
FSYNC robots cannot form a plane from an initial configuration $P$ 
if and only if $\gamma(P)$ is a 3D rotation group
and $|P_i| \in \{12, 24, 60\}$ for each $P_i$. 
The impossibility proof is by a construction 
based on the $\gamma(P)$-decomposition of the robots.
Obviously $12, 24$, and $60$ are the cardinalities of the 3D rotation groups 
and when a vertex-transitive set has a cardinality in $\{12, 24, 60\}$, 
the corresponding rotation group allows ``symmetric'' local coordinate
systems of those robots 
that allows identical local observations of those robots.
Thus they move symmetrically regarding the rotation group that results
in an infinite execution 
where the robots' positions keep the 3D rotation group forever.
Local memory at robots does not improve the situation 
since there exists an initial configuration where
the positions and local coordinate systems of robots are symmetric 
and 
the contents of local memory at robots are identical, e.g., empty. 
Hence we have the same impossibility result for 
non-oblivious FSYNC robots. 

Theorem~\ref{theorem:main} implies the following,
which is somewhat counter-intuitive:
The plane formation problem is {\em solvable},
even if the robots form a regular polyhedron
except the regular icosahedron in an initial configuration $P$, 
while it is {\em unsolvable} for the semi-regular polyhedra
except an icosidodecahedron. 

For the possibility proof, 
we present a plane formation algorithm 
for oblivious FSYNC robots, that non-oblivious
FSYNC robots can execute by ignoring the content of
their local memory. 
The proposed algorithm consists of a symmetry breaking algorithm
and a landing algorithm.
When the rotation group $\gamma(P)$ of an initial configuration $P$
is a 3D rotation group, 
the symmetry breaking algorithm translates $P$ 
into another configuration $P'$ 
whose rotation group $\gamma(P')$ is a 2D rotation group. 
From the condition of Theorem~\ref{theorem:main},
the $\gamma(P)$-decomposition of $P$ contains one of the
above five (semi-)regular polyhedra, i.e., 
a regular tetrahedron, a regular octahedron, a cube,
a regular dodecahedron, and an icosidodecahedron. 
The symmetry breaking algorithm breaks the symmetry of these five
polyhedra so that the resulting configuration $P'$ as a whole has a
2D rotation group. 
Then the robots can agree on a plane as described before for the 
2D-rotation groups and the landing algorithm assigns distinct
landing points on the agreed plane. 
The landing algorithm is quite simple but contains some technical
subtleties. 
We describe the entire plane formation algorithm with its 
correctness proofs.

\medskip
\noindent{\bf Related works.~}
Autonomous mobile robot systems in 2D-space
has been extensively investigated and
the main research interest has been the computational power
of robots. 
Many fundamental distributed tasks have been introduced, for example,
{\em gathering}, {\em pattern formation}, {\em partitioning},
and {\em covering}. 
These problems brought us deep insights on 
the limit of computational power of autonomous mobile robot systems and 
revealed necessary assumptions of such systems to complete a given task. 
We survey the state of the art of autonomous mobile robot systems 
in 2D-space 
since there is few research on robots in 3D-space. 
The book by Flocchini et al.~\cite{FPS12} contains almost all results 
on autonomous mobile robot systems up to year 2012.

Asynchrony and movement of robots are considered to be
subject to the {\em adversary}. In other words, we consider
the worst case scenario. 
Besides fully synchronous (FSYNC) robots,
there are two other types of robots, 
{\em semi-synchronous} (SSYNC) and {\em asynchronous} (ASYNC) robots.
The robots are SSYNC if some robots do not start 
the $i$-th Look-Compute-Move cycle for some $i$, 
but all of those who have started the cycle synchronously 
execute their Look, Compute and Move phases~\cite{SY99}. 
The robots are ASYNC if no assumptions are made on the execution of 
Look-Compute-Move cycles~\cite{FPSW08}.
The movement of a robot is {\em non-rigid} if in each Move phase, 
the robot moves at least unknown minimum moving distance $\delta$,
but after moving $\delta$ it may stop on any arbitrary point on
the track to the next position. 
If the length of the track to the next position is smaller than $\delta$,
it stops at the next position.
If a robot reaches its next position 
in any Move phase, its movement is {\em rigid}.
Most existing papers consider non-rigid movement of robots. 
Another important assumption is whether the robots agree on
the clockwise direction, i.e., {\em chirality}. 
Most existing literature assumes non-rigid movement and
chirality. 

One of the most general form of formation tasks
for autonomous mobile robot
systems is the {\em pattern formation problem}
that requires the robots to form a given target pattern. 
The pattern formation problem in 2D-space 
includes the line formation problem as a subproblem and
Yamashita et al. investigated its solvability 
for each of the FSYNC, SSYNC and ASYNC models~\cite{FYOKY15,SY99,YS10},
that are summarized as follows:
(1) For non-oblivious FSYNC robots,
a pattern $F$ is formable from an initial configuration $P$
if and only if $\rho(P)$ divides $\rho(F)$.
(2) Pattern $F$ is formable from $P$ by oblivious ASYNC robots 
if $F$ is formable from $P$ by non-oblivious FSYNC robots,
except for $F$ being a point of multiplicity 2.

This exceptional case is called the rendezvous problem.
Indeed, it is trivial for two FSYNC robots,
but is unsolvable for two oblivious SSYNC (and hence ASYNC)
robots~\cite{SY99}.
On the other hand, 
oblivious SSYNC (and ASYNC) robots can converge to a point. 
Therefore it is a bit surprising to observe that the point formation problem 
for more than two robots is solvable even for ASYNC robots.
The result first appeared in \cite{SY99} for SSYNC robots
and then is extended for ASYNC robots in \cite{CFPS12}.
As a matter of fact, 
except the existence of the rendezvous problem,
the point formation problem for more than two robots (which is also
called as the gathering problem) is the easiest problem
in that it is solvable from any initial configuration $P$,
since $\rho(F) = n$ when $F$ is a point of multiplicity $n$,
and $\rho(P)$ is always a divisor of $n$ by the definition of the symmetricity,
where $n$ is the number of robots.

The other easiest case is a regular $n$-gon
(frequently called the circle formation problem), 
since $\rho(F) = n$.
A circle is formable from any initial configuration,
like the point formation problem for more than two robots.
Recently the circle formation problem for $n$ oblivious ASYNC 
robots ($n \neq 4$) 
is solved without chirality~\cite{FPSV14}.

Das et al. considered formation of a sequence of
patterns by oblivious SSYNC robots with rigid movement~\cite{DFSY15}.
They showed that the symmetricity of each pattern of 
a formable sequence should be identical and a multiple of the
symmetricity of an initial configuration.
Such sequence of patterns is a geometric global memory formed
by oblivious robots. 

To circumvent the symmetricity and enable arbitrary pattern formation, 
Yamauchi and Yamashita proposed a randomized algorithm  that
allows the robots
to probabilistically break the symmetricity of the initial configuration 
and showed that the oblivious ASYNC robots can form any
target pattern with probability $1$~\cite{YY14}.

The notion of {\em compass} was first introduced in~\cite{FPSW05} 
that assumes agreement of the direction and/or the orientation
of $x$-$y$ local coordinate systems. 
Flocchini et al. showed that if the oblivious ASYNC robots without
chirality agree on the 
directions and orientations of $x$ and $y$ axes, 
they can form any arbitrary target pattern~\cite{FPSW08}. 

Flocchini et al. showed that agreement of the directions
and orientation of both axes of local coordinate systems
allows oblivious ASYNC robots
with {\em limited visibility} to solve the point formation
problem~\cite{FPSW05}.
A robot has limited visibility if it can observe other robots
within unknown fixed distance from itself. 
Agreement of the direction and the orientation of two axes
can be replaced by agreement of direction and the orientation of
one axis and chirality.
Souissi et al. investigate the effect of the deviation of one axis
from the global coordinate system at robots with 
chirality on the point formation problem
and first introduced unreliable compasses, called
{\em eventually consistent compass}, that is inaccurate
for an arbitrary long time, i.e., it has an 
arbitrary deviation and the deviation dynamically changes,
but eventually stabilizes to accurate axes~\cite{SDY09}. 
Izumi et al. investigated the maximum static and dynamic deviation
of compass for the point formation problem of 
two oblivious ASYNC robots~\cite{ISKI12}. 

Robustness of autonomous mobile robot systems has been discussed against 
error in sensing, computation, control, and 
several kinds of faults. 
A system is {\em self-stabilizing} if it accomplishes its task
from an arbitrary initial configuration. 
A self-stabilizing system can tolerate any finite number of
transient faults by considering the configuration
after the final fault as an arbitrary initial configuration~\cite{D74}. 
Suzuki and Yamashita pointed out that any oblivious mobile robot system 
is self-stabilizing since it does not depend on previous cycles~\cite{SY99}. 
Cohen and Peleg considered error in sensing, computation, and 
control, and showed acceptable range of them for oblivious ASYNC robots
to converge to a point~\cite{CP08}. 
Two fundamental types of permanent faults in distributed computing are 
{\em crash fault} that stops the faulty entity and
{\em Byzantine fault} that allows arbitrary (malicious) behavior of
faulty entity.
% Generally it is required that non-faulty robots complete a given task
% while there is no way to control faulty robots.  
Cohen and Peleg considered the effect of crash faults at robots on 
the convergence problem for oblivious ASYNC robots~\cite{CP05}. 
Bouzid et al. considered the effect of Byzantine faults at robots 
on the convergence problem in one-dimensional space 
(i.e., line) for SSYNC and ASYNC robots~\cite{BGT10}.
Agmon and Peleg considered both crash faults and Byzantine faults for the
point formation problem~\cite{AP06}. 

Efrima and Peleg considered the {\em partitioning problem}
that requires the robots to form teams of size $k$
that divides $n$~\cite{EP09}.
Without any compass, the partition problem is unsolvable
from a symmetric initial configuration and they considered the
availability of compass and asynchrony among robots.
Izumi et al. proposed an approximation algorithm for the
{\em set cover problem} of SSYNC robots
that requires that for a given set of 
target points, there is at least one robot in a unit distance 
from each target point~\cite{IKY14}. 
In contrast to the pattern formation problem, 
these problems have no (absolute) predefined final positions. 

Computational power of robots with limited visibility and without
any additional assumption has been also discussed. 
Yamauchi and Yamashita showed that oblivious FSYNC (thus SSYNC and
ASYNC) robots
with limited visibility have substantially weaker formation power
than the robots with unlimited visibility~\cite{YY13}. 
Ando et al. proposed a convergence algorithm for oblivious SSYNC
robots with limited visibility~\cite{AOSY99} 
while Flocchini et al. assumed consistent compass for
convergence of oblivious ASYNC robots with limited
visibility~\cite{FPSW05}. 

Peleg et al. first introduced the 
{\em luminous robot model} where each robot is equipped with
externally and/or internally visible lights~\cite{P05}. 
Light is an abstraction of both 
local memory and communication medium. 
Das et al. investigated the class of tasks that the luminous
robots can accomplish~\cite{DFPSY16}. 
They provided simulation algorithms for oblivious robots with
constant number of externally visible bits to simulate 
robots without lights in stronger synchronization model. 

All these papers discuss autonomous mobile robot systems
in 2D-space 
and little is known when the robots are placed in 3D-space.
This paper first investigates autonomous mobile robot systems in
3D-space and give a characterization of the plane formation problem.

\medskip
\noindent{\bf Organization.~~} 
In Section~\ref{sec:prel}, 
we first define the robot model and introduce the rotation group for
points in 3D-space. Then we briefly show our main idea for the
symmetry breaking algorithm. 
We start with some properties imposed on the robots by their 
rotation group in Section~\ref{sec:3Dsym}. 
In Section~\ref{sec:PlF}, 
we prove Theorem~\ref{theorem:main}
by showing the impossibility of symmetry breaking 
and by presenting a plane formation algorithm
for oblivious FSYNC robots for solvable instances. 
Finally, Section~\ref{sec:concl} concludes this paper 
by giving some concluding remarks. 

%========================================================================
\section{Preliminary}
\label{sec:prel}

\subsection{Robot model} 
\label{SSmodel}

Let $R = \{r_1, r_2, \ldots, r_n\}$ be a set of anonymous $n$ robots 
each of which is represented by a point in 3D-space.
Their indices are used just for description. 
Without loss of generality, we assume $n \geq 4$,
since all robots are already on a plane when $n \leq 3$. 
By $Z_0$ we denote the global $x$-$y$-$z$ coordinate system.
Let $p_i(t) \in \Real^3$ be the position of $r_i$ at time $t$ in $Z_0$,
where $\Real$ is the set of real numbers.
A {\em configuration} of $R$ at time $t$ is denoted by 
$P(t) = \{p_1(t), p_2(t), \ldots, p_n(t)\}$. 
We assume that the robots initially occupy distinct positions,
i.e., $p_i(0) \not= p_j(0)$ for all $1 \leq i < j \leq n$. 
In general, $P(t)$ can be a multiset,
but it is always a set throughout this paper 
since the proposed algorithm avoids any multiplicity.\footnote{
It is impossible to break up multiple oblivious FSYNC robots 
(with the same local coordinate system)
on a single position as long as they execute the same algorithm. 
Our algorithm is designed to avoid any multiplicity.
However, we need to take into account any algorithm that may lead
$R$ to a configuration with multiplicity 
when proving the impossibility result by reduction to the absurd.}
The robots have no access to $Z_0$.
Instead, each robot $r_i$ has a local $x$-$y$-$z$ coordinate system $Z_i$,
where the origin is always its current location,
while the direction of each positive axis and the magnitude of 
the unit distance are arbitrary but never change.
We assume that $Z_0$ and all $Z_i$ are right-handed.
Thus $Z_i$ is either a uniform scaling, transformation,
rotation, or their combination of $Z_0$. 
By $Z_i(p)$ we denote the coordinates of a point $p$ in $Z_i$.

Each robot repeat a {\em Look-Compute-Move cycle}. 
We investigate fully synchronous (FSYNC) robots in this paper.
They all start the $t$-th Look-Compute-Move cycle simultaneously 
and synchronously execute each of its Look, Compute, and Move phases.
We specifically assume without loss of generality 
that the $(t+1)$-th Look-Compute-Move cycle starts 
at time $t$ and finishes before time $t+1$.
At time $t$,
each robot $r_i$ simultaneously looks 
and obtains a set\footnote{
Since $Z_i$ changes whenever $r_i$ moves,
notation $Z_i(t)$ is more rigid,
but we omit parameter $t$ to simplify its notation.}
\begin{equation*}
Z_i(P(t)) = \{ Z_i(p_1(t)), Z_i(p_2(t)), \ldots , Z_i(p_n(t)) \}.
\end{equation*}
We call $Z_i(P(t))$ the {\em local observation} of $r_i$ at $t$.
Next, $r_i$ computes its next position using an algorithm $\psi$,
which is common to all robots. 
If $\psi$ uses only $Z_i(P(t))$, we say that $r_i$ is {\em oblivious}. 
Thus $\psi$ is a total function from ${\cal P}_n^3$ to $\Real^3$,
where ${\cal P}_n^3 = (\Real^3)^n$ is the set of all
configurations.\footnote{A configuration generally contains
multiplicities and ${\cal P}_n^3$ contains such configurations.
However we do not assume multiplicity detection ability of robots.
Thus the input to an algorithm is a set of points. As we will show
later, the proposed pattern formation algorithm makes no multiplicity
during any execution thus the input to the algorithm is always a set of 
$n$ points.} 
Otherwise, we say $r_i$ is {\em non-oblivious}, i.e., 
$r_i$ can use past local observations and past outputs of $\psi$. 
We say that a non-oblivious robot is equipped with local memory. 
Finally, $r_i$ moves to $\psi(Z_i(P(t)))$ in $Z_i$ before time $t+1$.
Thus we assume rigid movement. 

An infinite sequence of configurations
${\cal E}: P(0), P(1), \ldots$ is called an {\em execution} 
from an {\em initial configuration} $P(0)$.
Observe that the execution $\cal E$ is uniquely determined,
once initial configuration $P(0)$, 
local coordinate systems at time $0$,
local memory contents (for non-oblivious robots),  
and algorithm $\psi$ are fixed.

We say that an algorithm $\psi$ {\em forms a plane}
from an initial configuration $P(0)$,
if, regardless of the choice of initial local coordinate systems of
robots and their initial memory contents (if any), 
for any execution $P(0), P(1), \ldots$, 
there exists finite $t \geq 0$ such that $P(t)$ satisfies the following
three conditions:
\begin{description}
\item[(a)]
$P(t)$ is contained in a plane, 
\item[(b)]
all robots occupy distinct positions in $P(t)$, and 
\item[(c)] 
for any $t' \geq t$, $P(t') = P(t)$. 
\end{description}
Because of $(b)$, gathering the robots to one point (i.e., point
formation) is not a solution for the
plane formation problem.

\subsection{Rotation groups in 3D-space} 

In 2D-space,
the symmetricity $\rho(P)$ of a set of points $P$
is defined by the order of its cyclic group, 
where the rotation center $o$ is the center of the smallest 
enclosing circle of $P$, if $o \not\in P$. 
Otherwise, $\rho(P) = 1$.
Then $P$ is decomposed into $n/\rho(P)$ regular 
$\rho(P)$-gons with $o$ being the common center, 
where $n = |P|$~\cite{SY99}.
(See Figure~\ref{fig:2dim-sym}.) 
Since the robots in the same regular $\rho(P)$-gon may have the
same local observation, 
no matter which deterministic algorithm they obey, 
we cannot exclude the possibility that they continue to keep 
a regular $\rho(P)$-gon during the execution.
This is the main reason that a target pattern $F$ is not formable
from an initial configuration $P$,
if $\rho(P)$ does not divide $\rho(F)$~\cite{FYOKY15,YK96,YS10}.

%%%%%%%%%%%%%%%%%%%%%%%%%%%%%%%%%%%%%%%%%%%%%%%%%%%%%%%%%%%%%%%%%%%%%%%%%%
\begin{figure}[t]
\centering 
\includegraphics[width=3cm]{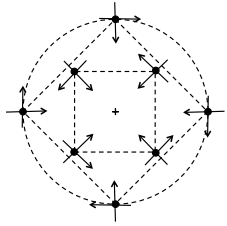}
\caption{A symmetric initial configuration in 2D-space,
whose symmetricity is $4$.
Eight robots and their local coordinate systems are 
symmetric with respect to the center of their smallest enclosing circle. 
There are two groups consisting of $4$ symmetric robots  
and the robots in each group cannot break their symmetry. }
\label{fig:2dim-sym}
\end{figure}
%%%%%%%%%%%%%%%%%%%%%%%%%%%%%%%%%%%%%%%%%%%%%%%%%%%%%%%%%%%%%%%%%%%%%%%%%%

In 3D-space,
we consider the smallest enclosing ball and 
the convex hull of the positions of robots, i.e., 
robots are vertices of a convex polyhedron. 
Typical symmetric polyhedra are regular polyhedra (Platonic solids) 
and semi-regular polyhedra (Archimedean solids). 
A {\em uniform polyhedron} is a polyhedron consisting of 
regular polygons and all its vertices are congruent. 
Any uniform polyhedron is {\em vertex transitive}, i.e., 
for any pair of vertices of the polyhedron, 
there exists a symmetry operation that moves one vertex 
to the other with keeping the polyhedron as a whole. 
Intuitively, it makes sense to expect that all vertices (robots) 
in a uniform polyhedron may have identical local observations  
and might not break the symmetry in the worst case. 
The family of uniform polyhedra consists of $5$ regular polyhedra  
(the regular tetrahedron, the cube, the regular octahedron, 
the regular dodecahedron, and the regular icosahedron), 
13 semi-regular polyhedra, and other non-convex 57 polyhedra.\footnote{
We do not consider Miller's solid as semi-regular polyhedra though
it satisfies the definition because we focus on rotation groups.
Actually the rotation group of Miller's solid is not a polyhedral group 
but $D_4$.} 
We do not care for non-convex uniform polyhedra. 
Contrary to the intuition above,
we will show that when robots form a regular tetrahedron, a regular
octahedron, a cube, a regular dodecahedron, or an icosidodecahedron, 
they can break their symmetry and form a plane. 

In general, symmetry operations on a polyhedron consists of 
rotations around an axis, reflections for a mirror plane 
({\em bilateral symmetry}), 
reflections for a point ({\em central inversion}), 
and {\em rotation-reflections}~\cite{C97}. 
But as briefly argued in Section~\ref{sec:intro}.
since all local coordinate systems are right-handed,
it is sufficient to consider only direct congruent transformations 
and those keeping the center. They are rotations 
around some axes that contains the center. 
We thus concentrate on rotation groups with finite order.

In 3D-space, there are five kinds of rotation groups of finite order
each of which is defined by the set of rotation axes and their
arrangement~\cite{C97}.
We can recognize each of them as the group formed by
rotation operations on some polyhedron. 
Consider a regular pyramid that has a regular $k$-gon as its base 
(Figure~\ref{fig:axis-cyc}). 
The rotation operations for this regular pyramid is rotation 
by $2\pi i/k$ for $1 \leq i \leq k$ 
around an axis containing the apex and the center of the base. 
We call such an axis {\em $k$-fold axis}. 
Let $a^i$ be the rotation by $2\pi i/k$ 
around this $k$-fold axis with $a^k = e$ where $e$ is the identity element. 
Then, $a^1, a^2, \ldots, a^k$ form a group, which is called 
the {\em cyclic group}, denoted by $C_k$. 

A regular prism (except a cube) that has a regular $\ell$-gon as its base 
has two types of rotation axes, 
one is the $\ell$-fold axis containing the centers of its base and top, 
and the others are $2$-fold axes that exchange the base and the 
top (Figure~\ref{fig:axis-dih}). 
We call this single $\ell$-fold axis {\em principal axis} and 
the remaining $\ell$ $2$-fold axes {\em secondary axes}. 
These rotation operations on a regular prism form a group, 
which is called the {\em dihedral group}, denoted by $D_{\ell}$.
The order of $D_{\ell}$ is $2\ell$. 
When $\ell=2$, we can define $D_2$ in the same way, 
but in the group theory we do not distinguish the principal axis
from the secondary one.
Indeed, $D_2$ is isomorphic to the Klein four-group, denoted by $K_4$,
which is an abelian group and is a normal subgroup of the 
alternating group of degree $4$, denoted by $A_4$. 
Later we will show that we can recognize the principal axis of $D_2$ from
the others because we consider rotations on a set of points. 

The rotation axes of a regular polyhedron are classified into three types: 
The axes that contain the centers of opposite faces (type $a$), 
the axes that contain opposite vertices (type $b$), and 
the axes that contain the midpoints of opposite edges (type $c$). 
For each regular polyhedron, 
the rotation operations also form a group 
and the following three groups are called the {\em polyhedral groups}. 

% tetrahedral 
The regular tetrahedron has four $3$-fold type $a$ (and $b$) axes and 
three $2$-fold type $c$ axes (Figure~\ref{fig:axis-tetra}). 
This rotation group is called the {\em tetrahedral group}
denoted by $T$. 
The tetrahedral group is isomorphic to $A_4$ and its order is $12$. 

% octahedral 
The regular octahedron has four $3$-fold type $a$ axes, 
three $4$-fold type $b$ axes, 
and six $2$-fold type $c$ axes (Figure~\ref{fig:axis-octa}). 
This rotation group is called the {\em octahedral group} 
denoted by $O$. 
The octahedral group is isomorphic to the symmetric group of degree $4$ 
denoted by $S_4$ and its order is $24$. \footnote{Consider a cube to which we can 
perform the rotation of $O$.  
Each rotation permutes the diagonal lines of the cube. }

% icosahedral 
The regular icosahedron has ten $3$-fold type $a$ axes, 
six $5$-fold type $b$ axes, 
and fifteen $2$-fold type $c$ axes (Figure~\ref{fig:axis-icosa}). 
This rotation group is called the {\em icosahedral group}, 
denoted by $I$. 
The icosahedral group is isomorphic to the alternating group of 
degree $5$ denoted by $A_5$ and its order is $60$.

For each regular polyhedron, consider the center of each face. 
These centers also form a regular polyhedron, 
which is called the {\em dual} of the original regular polyhedron. 
Any dual polyhedron has the same rotation group as its original polyhedron. 
The regular tetrahedron is self-dual, 
the cube and the regular octahedron are dual each other, 
and so are the regular dodecahedron and the regular icosahedron. 
Hence we have three polyhedral groups. 

%%%%%%%%%%%%%%%%%%%%%%%%%%%%%%%%%%%%%%%%%%%%%%%%%%%%%%%%%%%%%%%%%%%%%%%%%%
\begin{figure}[t]
\centering 
\subfigure[]{\includegraphics[width=2cm]{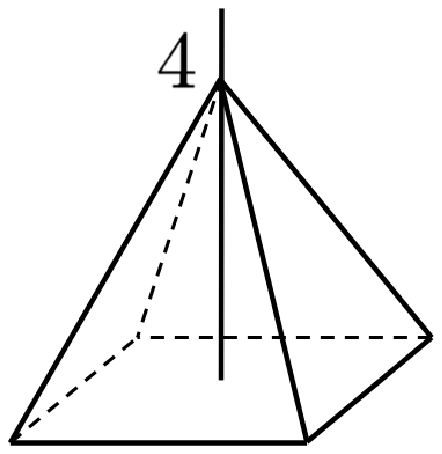}\label{fig:axis-cyc}}
\hspace{2mm}
\subfigure[]{\includegraphics[width=2cm]{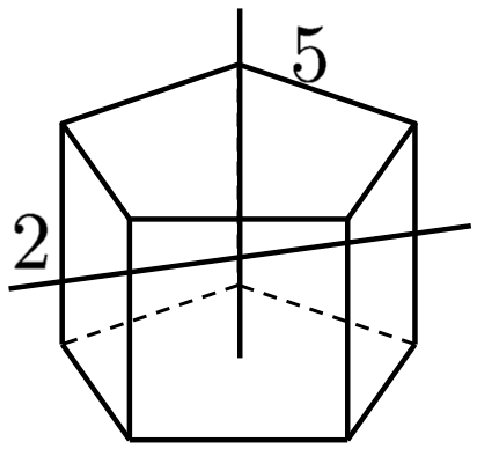}\label{fig:axis-dih}}
\hspace{2mm}
\subfigure[]{\includegraphics[width=2cm]{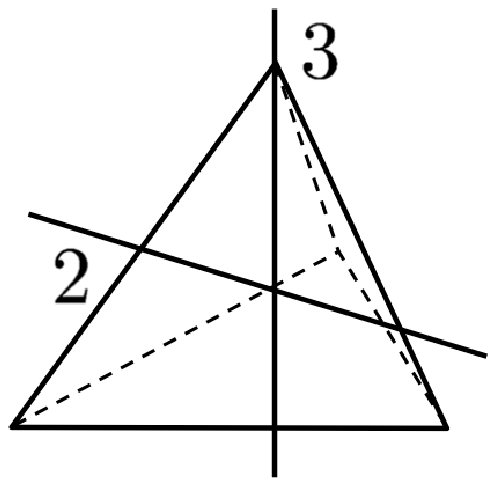}\label{fig:axis-tetra}}
\hspace{2mm}
\subfigure[]{\includegraphics[width=2cm]{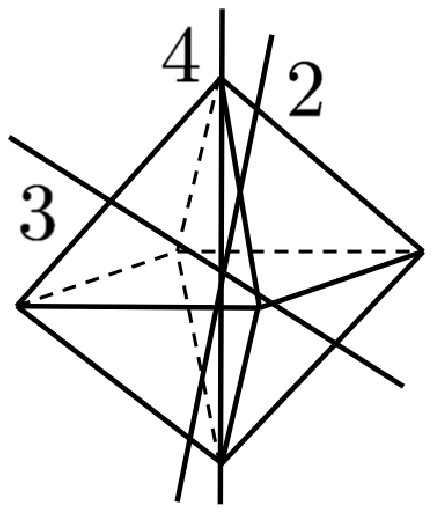}\label{fig:axis-octa}}
\hspace{2mm}
\subfigure[]{\includegraphics[width=2cm]{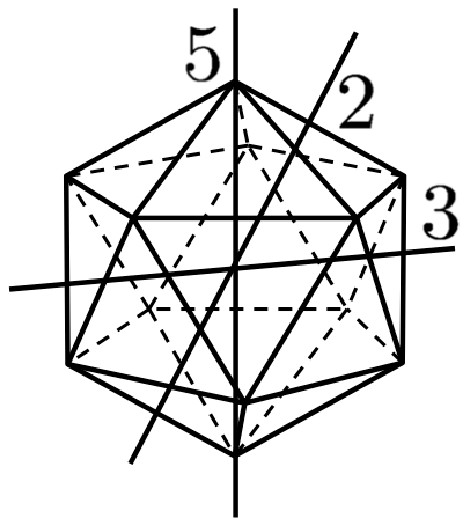}\label{fig:axis-icosa}}
\label{fig:csym}
\caption{Rotation groups: 
(a) the cyclic group $C_4$, (b) the dihedral group $D_5$, 
(c) the tetrahedral group $T$, (d) the octahedral group $O$, and (e) 
the icosahedral group $I$. Figures show only one axis for each type. }
\end{figure}
%%%%%%%%%%%%%%%%%%%%%%%%%%%%%%%%%%%%%%%%%%%%%%%%%%%%%%%%%%%%%%%%%%%%%%%%%%

Table~\ref{table:elements} shows for each of the four rotation groups, 
$T$, $O$, and $I$, 
the number of elements (excluding the identity element) around 
its $k$-fold axes ($k \in \{2,3,4,5\}$).
\begin{table}
\begin{center}
 \caption{Polyhedral groups.
 The number of elements around $k$-fold axes
(excluding the identity element) and their orders. }
\label{table:elements} 
\begin{tabular}{c|c|c|c|c|c}
\hline 
Polyhedral group & $2$-fold axes & $3$-fold axes & $4$-fold axes &
$5$-fold axes & Order \\ 
\hline 
$T$ & 3 & 8 & - & - & 12 \\ 
$O$ & 6 & 8 & 9 & - & 24 \\ 
$I$ & 15 & 20 & - & 24 & 60\\ 
\hline
\end{tabular}
\end{center}
\end{table}

Let ${\mathbb S} = \{C_{k}, D_{\ell}, T, O, I \ | k = 1,2,
\ldots, and \ \ell = 2,3,\ldots\}$ 
be the set of rotation groups,
where $C_1$ is the rotation group with order 1;
its unique element is the identity element (i.e., $1$-fold rotation).
When $G'$ is a subgroup of $G$ ($G, G' \in {\mathbb S}$), we denote it 
by $G' \preceq G$. If $G'$ is a proper subgroup of $G$ (i.e., $G \neq G'$), 
we denote it by $G' \prec G$. For example, we have $D_2 \prec T$, 
$T \prec O, I$, but $O \not\preceq I$. 
If $G \in {\mathbb S}$ has a $k$-fold axis, 
$C_{k'} \preceq G$ if $k'$ divides $k$. 

We now define the rotation group of a set of points in 3D-space. 

\begin{definition}
\label{def:gamma}
The rotation group $\gamma(P)$ of a set of points $P \in {\cal P}_n^3$ 
is the group that acts 
on $P$ and none of its proper supergroup in ${\mathbb S}$ acts on $P$. 
\end{definition}

Clearly, $\gamma(P)$ for any given set of points $P$ is uniquely
determined. 
For example, when $P$ is the set of vertices of a cube, 
$\gamma(P)$ is the octahedral group $O$. 
The major difference between the symmetricity in 2D-space and
the rotation group in 3D-space is that
even when the points of $P$ are on one plane, 
its rotation group is chosen from the dihedral groups and
cyclic groups. In our context, symmetricity in 2D-space assumes
the ``top'' direction against the plane
where the points reside~\cite{FYOKY15,SY99,YS10}, while in
3D-space there is no agreement on the ``top'' direction. 

For any $P \in {\cal P}_n^3$, 
by $B(P)$ and $b(P)$,
we denote the smallest enclosing ball of $P$ and its center,
respectively.
From the definition, all rotation axis of $\gamma(P)$ contains $b(P)$
and $b(P)$ is the intersection of all 
rotation axes of $\gamma(P)$ unless $\gamma(P)=C_1$. 
A point on the sphere of a ball is said to be {\em on} the ball,
and we assume that the {\em interior} or the {\em exterior}
of a ball does not include its sphere. 
When all points are on $B(P)$, 
we say that the set of points is {\em spherical}. 
For a ball $B$, we denote the radius of the ball by $rad(B)$
in the coordinate system to observe $B$. 

We say that a set of points $P$ is {\em transitive}
regarding a rotation group $G$ if it is an orbit of
$G$ through some seed point $s$, i.e.,
$P = Orb(s) = \{ g*s : g \in G \}$ for some $s \in P$.\footnote{
For a transitive set of points $P$,
any $s \in P$ can be a seed point. } 
The vertex-transitivity of uniform polyhedra 
corresponds to transitivity regarding a 3D rotation group. 
In the following,
we use ``vertex-transitivity'' for a polyhedron while 
we use ``transitivity'' for a set of points. 
Note that a transitive set of points is always spherical. 

Given a set of points $P$, 
$\gamma(P)$ determines the arrangement of its rotation axes.  
We thus use $\gamma(P)$ and the arrangement 
of its rotation axes in $P$ interchangeably.
For two groups $G, H \in {\mathbb S}$, 
an {\em embedding} of $G$ to $H$
is an embedding of each rotation
axis of $G$ to one of the rotation axes of $H$ 
so that any $k$-fold axis of $G$ overlaps a $k'$-fold axis of $H$
with keeping the arrangement of $G$ where $k$ divides $k'$. 
For example, we can embed $T$ to $O$ so that each $3$-fold axis of $T$ 
overlaps a $3$-fold axis of $O$, and each $2$-fold 
axis of $T$ overlaps a $4$-fold axis of $O$. 
Note that there may be many embeddings of $G$ to $H$. 
There are three embeddings of $C_4$ to $O$ depending on the choice of 
the $4$-fold axis. 
Observe that we can embed $G$ to $H$ if and only if $G \preceq H$. 
For example, $O$ cannot be embedded to $I$,
since $O$ is not a subgroup of $I$.

In the group theory, we do not distinguish the principal axis 
of $D_2$ from the other two $2$-fold axes. 
Actually, since we consider the rotations on a set of points in 3D-space, 
we can recognize the principal axis of $D_2$. 
Consider a sphenoid consisting of $4$ congruent non-regular triangles 
(Figure~\ref{fig:sphenoid}).
A rotation axes of such a sphenoid contains the midpoints of 
opposite edges and there are three $2$-fold axis 
perpendicular to each other. 
Hence the rotation group of the vertices of such a sphenoid is $D_2$. 
However we can recognize, for example, the vertical $2$-fold axis 
from the others by their lengths (between the midpoints connecting).
The vertex-transitive polyhedra on which only $D_2$ can act are 
rectangles and the family of such sphenoids and we can always
recognize the principal axis. 
Other related polyhedra are lines, squares, and regular tetrahedra, but 
$D_{\infty}$ acts on a line, 
$D_4$ acts on a square, and $T$ acts on a regular tetrahedron. 
Hence their rotation groups are proper supergroup of $D_2$. 
We can show the following property regarding the principal axis of
$D_2$. See Appendix~\ref{app:rotation-groups} for the proof. 
\begin{property}
\label{property:d2-principal}
Let $P \in {\cal P}_n^3$ be a set of points. 
If $D_2$ acts on $P$ and we cannot distinguish the principal axis of 
(an arbitrary embedding of) $D_2$, then $\gamma(P) \succ D_2$. 
\end{property}

Later we will show that the robots can form a plane 
if they can recognize a single rotation axis or a principal axis. 
Based on this, 
we say that the cyclic groups and the dihedral groups 
are {\em two-dimensional} (2D), 
while the polyhedral groups are {\em three-dimensional} (3D)
since polyhedral groups do not act on a set of points on a plane.

%%%%%%%%%%%%%%%%%%%%%%%%%%%%%%%%%%%%%%%%%%%%%%%%%%%%%%%%%%%%%%%%%%%%%%%%%%
\begin{figure}[t]
\centering 
\includegraphics[width=2.3cm]{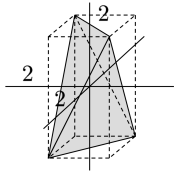}
\caption{A sphenoid consisting of $4$ congruent isosceles triangles.
Its rotation group is $D_2$.
Since the vertices are not placed equidistant positions from the three axes,
we can distinguish one axis as the principal axis from the others.}
\label{fig:sphenoid}
\end{figure}
%%%%%%%%%%%%%%%%%%%%%%%%%%%%%%%%%%%%%%%%%%%%%%%%%%%%%%%%%%%%%%%%%%%%%%%%%%

\subsection{Basic idea} 
\label{subsec:idea}

We first show an stimulating example that shows our idea of the
symmetry breaking algorithm and the impossibility of the plane formation
problem. 
From Theorem~\ref{theorem:main}, oblivious FSYNC robots can
form a plane from an initial configuration where four robots
form a regular tetrahedron 
(i.e., they occupy the vertices of a regular tetrahedron).  
In such an initial configuration, 
their local observation may be identical because 
of the vertex-transitivity of the regular tetrahedron. 
If each robot proposes one plane,
these four planes may be symmetric regarding $T$,
and because $T$ is three dimensional,
these four planes never become identical. 
They must break their rotation group $T$ to form a plane. 
It is an essential challenge of this paper that 
the robots solve this symmetry breaking problem
by a deterministic algorithm. 

We introduce a simple ``go-to-midpoint'' algorithm for the robots to
break the regular tetrahedron. 
This algorithm makes each robot select an arbitrary edge of
the regular tetrahedron which is incident to the vertex it resides
and goes along the edge, 
but stops it $\epsilon$ before the midpoint, 
where $\epsilon$ is $1/100$ of the length of the edge. 
The selection is somehow done in a deterministic way. 
We briefly show that this go-to-midpoint algorithm successfully
breaks the symmetry of the regular tetrahedron 
and the robots can form a plane. 
We could have a better understanding of the execution 
by illustrating the positions of robots
in an embedding of the regular tetrahedron
to a cube. Figure~\ref{fig:tcube1} shows an initial configuration $P$. 
Since at least two edges are selected by the four robots, 
we have the following three cases. 

\noindent{\bf Case A:~} Two edges are selected.
See Figure~\ref{fig:tcube2}. 
The two edges are opposite edges  
and the robots form skew lines of length $2\epsilon$,
since otherwise, two edges cannot cover the four vertices. 
The four robots can agree on the plane perpendicular to 
the line segment containing the midpoints of the skew lines and 
containing its midpoint.

\noindent{\bf Case B:~} Three edges are selected.
See Figures~\ref{fig:tcube3},~\ref{fig:tcube6} and~\ref{fig:tcube7}. 
There is only one pair of robots with distance $2\epsilon$  
and the four robots can agree on the plane 
formed by the midpoint of the two robots with distance $2\epsilon$ 
and the positions of the remaining two robots. 

\noindent{\bf Case C:~} Four edges are selected.
If three of the selected edges form a regular triangle 
(Figure~\ref{fig:tcube4}), 
the distance from the remaining robot to two of the three robots 
is larger than the edge of the regular triangle. 
Hence, the four robots can agree on the plane containing the regular triangle. 
Otherwise, the selected edges form a cycle on the 
original regular tetrahedron (Figure~\ref{fig:tcube5}). 
In this case, the four robots form a set of skew lines 
and can agree on the plane like (A).

In each case, 
the four robots can land on the foot of the perpendicular line 
to the agreed plane starting from its current position. 
They succeed in plane formation since they are FSYNC.

%%%%%%%%%%%%%%%%%%%%%%%%%%%%%%%%%%%%%%%%%%%%%%%%%%%%%%%%%%%%%%%%%%%%%%%%%%
\begin{figure}[t]
\centering 
\subfigure[]{\includegraphics[width=2cm]{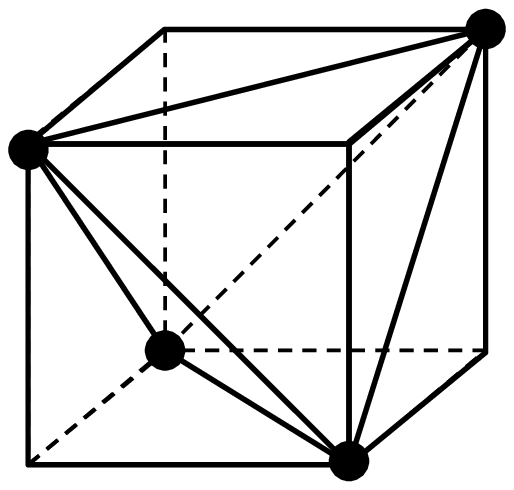}\label{fig:tcube1}}
\hspace{1mm}
\subfigure[]{\includegraphics[width=2cm]{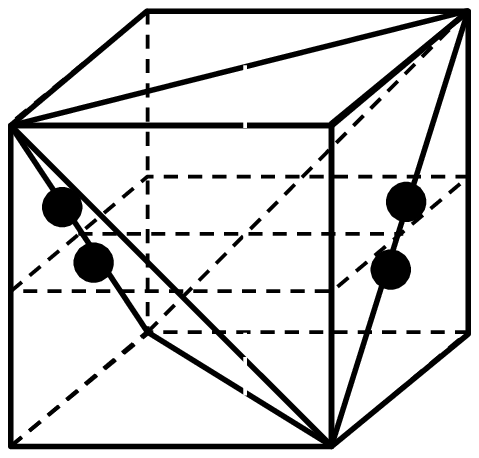}\label{fig:tcube2}}
\hspace{1mm}
\subfigure[]{\includegraphics[width=2cm]{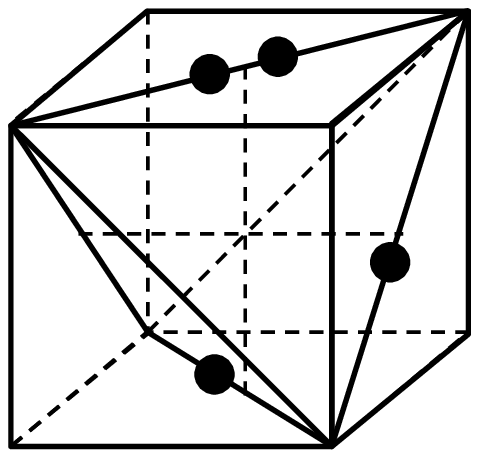}\label{fig:tcube3}}
\hspace{1mm}
\subfigure[]{\includegraphics[width=2cm]{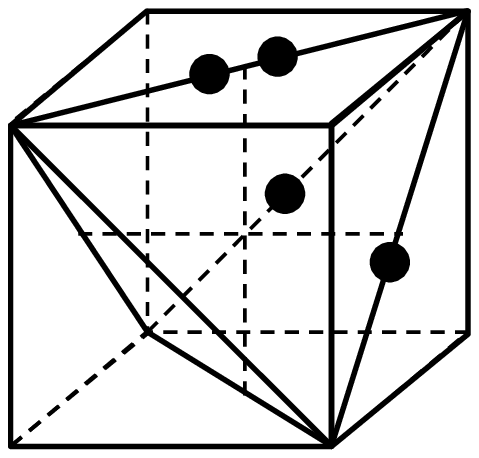}\label{fig:tcube6}}
\hspace{1mm}
\subfigure[]{\includegraphics[width=2cm]{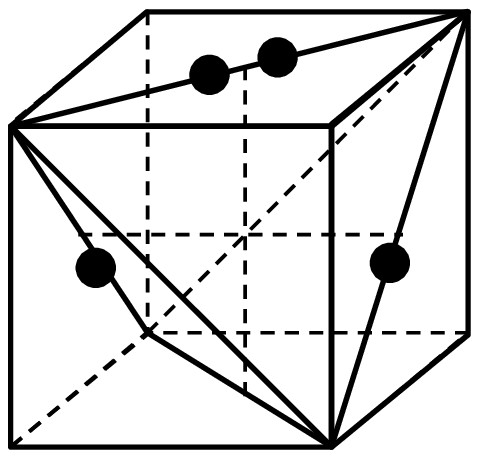}\label{fig:tcube7}}
\hspace{1mm}
\subfigure[]{\includegraphics[width=2cm]{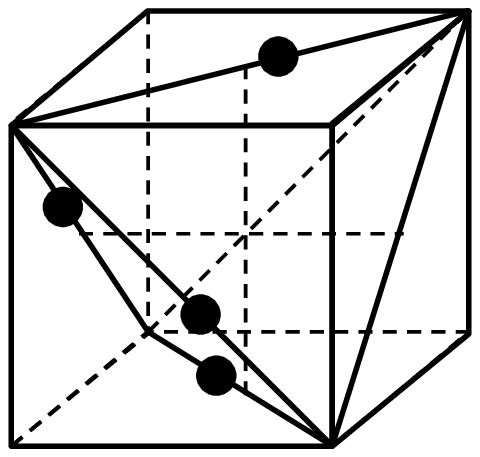}\label{fig:tcube4}}
\hspace{1mm}
\subfigure[]{\includegraphics[width=2cm]{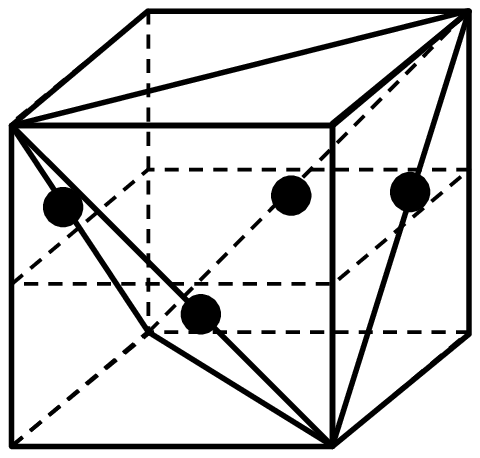}\label{fig:tcube5}}
\label{fig:tcube}
 \caption{Execution of the go-to-midpoint algorithm from an
 initial configuration where the four robots form a regular
 tetrahedron. }
\end{figure}
%%%%%%%%%%%%%%%%%%%%%%%%%%%%%%%%%%%%%%%%%%%%%%%%%%%%%%%%%%%%%%%%%%%%%%%%%%

One might expect that the go-to-midpoint algorithm could be used to
break symmetry of any other regular 
polyhedra because of the Euler's equality: 
For a polyhedron with $V$ vertices, $E$ edges, and $F$ faces, 
we have $V - E + F = 2$. 
If the go-to-midpoint algorithm is executed in such a configuration,
since $F > 2$ and hence $V \not= E$,
there exists at least one edge which is selected by two robots
or is not selected by any robot. 
However, as a matter of fact, the go-to-midpoint algorithm
does not work, for example, when
the robots form a regular icosahedron. 
Figure~\ref{fig:sym-eicosa-dges} shows an example of a
configuration $P'$ obtained 
by the go-to-midpoint algorithm from an initial configuration
where the robots form a regular icosahedron. 
The robots cannot agree on a plane in $P'$
because a 3D rotation group $T$ cats on $P'$ as shown in
Figure~\ref{fig:d2min} and the twelve planes
that the robots propose are not identical. 
Later we will show that the robots following any algorithm cannot 
agree on a plane forever from this configuration 
irrespective
of obliviousness. 

The ``go-to-midpoint'' algorithm shows that
the robots can reduce their rotation group by deterministic movement,
while in some cases this reduction stops at some
subgroup of the rotation group of the initial configuration. 
Our plane formation algorithm proposed in Subsection~\ref{subsec:suff}
translates an initial configuration whose rotation group is
a 3D rotation group to another configuration 
whose rotation group is a 2D rotation group. 
Then robots can agree on a plane that is perpendicular to the
single (or principal) axis and contains the center of
their smallest enclosing ball. 
Then they land on the plane. 

To show a necessary condition, we characterize the initial
configurations from which the robots cannot always form a plane
in terms of the rotation group and the number of robots. 
 
%%%%%%%%%%%%%%%%%%%%%%%%%%%%%%%%%%%%%%%%%%%%%%%%%%%%%%%%%%%%%%%%%%%%%%%%%%
\begin{figure}[t]
\centering 
\subfigure[]{\includegraphics[width=2.5cm]{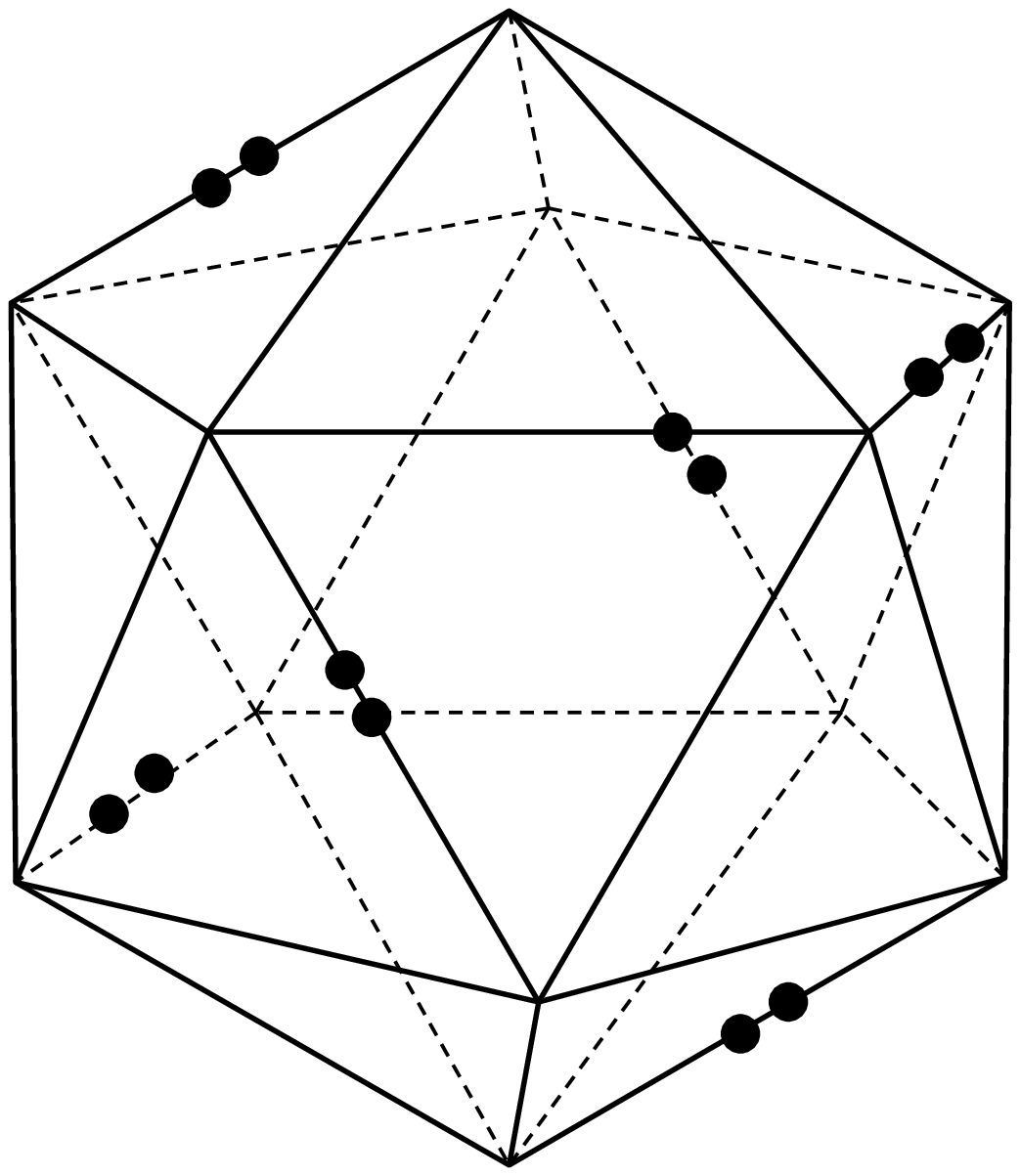}\label{fig:sym-eicosa-dges}}
\hspace{5mm}
\subfigure[]{\includegraphics[width=2.5cm]{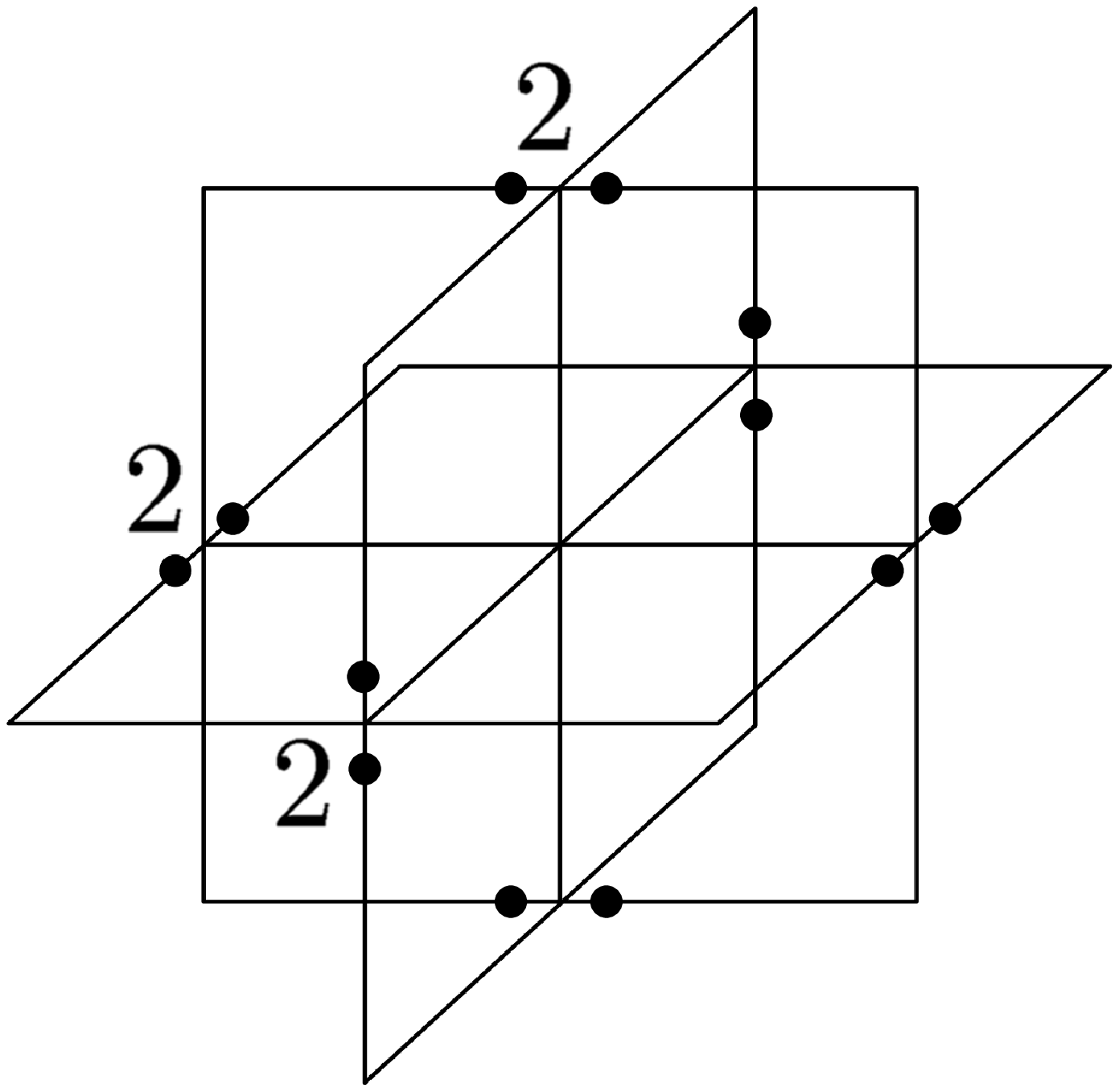}\label{fig:d2min}}
 \caption{An example of a resulting configuration of the
 go-to-midpoint algorithm from an initial configuration
 where the robots form a regular icosahedron. 
 (a) A resulting configuration. 
 (b) The rotation group of the resulting configuration is $T$ and
 its $2$-fold axes are illustrated. (We omit the four $3$-fold axes.) }
\label{fig:sb-icosa} 
\end{figure}
%%%%%%%%%%%%%%%%%%%%%%%%%%%%%%%%%%%%%%%%%%%%%%%%%%%%%%%%%%%%%%%%%%%%%%%%%%

\section{Decomposition of the robots}
\label{sec:3Dsym}

In this section, we will show that the robots can agree on some
global properties by using the rotation group of their positions. 
In a configuration $P$, each robot $r_i$ can obviously calculate 
$\gamma(P)$ from $Z_i(P)$ 
by checking all rotation axes that keep $P$ unchanged.
Then the group action of $\gamma(P)$ 
decomposes $P$ into a family of transitive sets of points
and the robots can agree on the ordering of these elements. 
As we will show in Section~\ref{subsec:nec}, 
each of these elements are a set of indivisible robots
in the worst case that have the same local observation,
move symmetrically, and keep $\gamma(P)$ forever. 
On the other hand, this ordering allows us to control the
robots in some order and plays an important role
when we design a plane formation algorithm for solvable
initial configurations. 
We start with the following theorem. 

\begin{theorem}
\label{theorem:decomposition}
Let $P \in {\mathcal P}_n^3$ be a configuration of robots represented as 
a set of points. 
Then $P$ is decomposed into disjoint sets $\{P_1, P_2, \ldots , P_m\}$ 
so that each $P_i$ is transitive regarding $\gamma(P)$. 
Furthermore, the robots can agree on a total ordering among the 
elements. 
\end{theorem}

Such decomposition of $P$ is unique as a matter of fact 
and we call this decomposition $\{ P_1, P_2, \ldots , P_m \}$
the $\gamma(P)$-{\em decomposition} of $P$.
Let us start with the first part of Theorem~\ref{theorem:decomposition}.

\begin{lemma}
\label{lemma:subsets}
Let $P \in {\mathcal P}_n^3$ be a configuration of robots represented as 
a set of points. 
Then $P$ is decomposed into disjoint sets $\{P_1, P_2, \ldots , P_m\}$ 
so that each $P_i$ is transitive regarding $\gamma(P)$.
\end{lemma}

\begin{proof}
For any point $p \in P$,
let $Orb(p) = \{ g*p \in P : g \in \gamma(P) \}$ be the orbit of 
the group action of $\gamma(P)$ through $p$.
By definition $Orb(p)$ is transitive regarding $\gamma(P)$.
Let $\{ Orb(p) : p \in P\} = \{ P_1, P_2, \ldots , P_m \}$ be its orbit space.
Then $\{ P_1, P_2, \ldots , P_m \}$ is obviously a partition, 
which satisfies the property of the lemma.
Additionally, such decomposition is unique. 
 \shortqed 
\end{proof}

Note that $|P_i| = |P_j|$ ($i \neq j$) may not hold,
while in 2D-space a set of points $P$ is decomposed into 
regular $|\rho(P)|$-gons by $\rho(P)$~\cite{SY99,YS10,FYOKY15}. 
Consider a configuration $P$ consisting of the vertices of 
a regular tetrahedron (4 vertices) and 
the vertices of a truncated tetrahedron (12 vertices)
(Figure~\ref{fig:subsets}). 
Then $\gamma(P) = T$ 
and the sizes of the elements of the $\gamma(P)$-decomposition of $P$
are different. 

%%%%%%%%%%%%%%%%%%%%%%%%%%%%%%%%%%%%%%%%%%%%%%%%%%%%%%%%%%%%%%%%%%%%%%%%%%
\begin{figure}[t]
\centering 
\includegraphics[width=3cm]{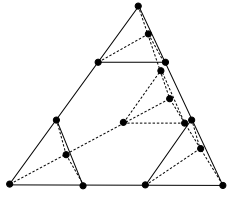}
\caption{A set of points $P$ consisting of 16 points.
Its rotation group is $\gamma(P)$ is $T$ 
and the $\gamma(P)$-decomposition of $P$ consists of two elements: 
a set of points forming a regular tetrahedron (of size 4) and
a set of points forming a truncated tetrahedron (of size 12).}
\label{fig:subsets}
\end{figure}
%%%%%%%%%%%%%%%%%%%%%%%%%%%%%%%%%%%%%%%%%%%%%%%%%%%%%%%%%%%%%%%%%%%%%%%%%%

Let us go on the second part of the theorem.
For the robots to consistently compare two elements $P_i$ and $P_j$
of the $\gamma(P)$-decomposition of $P$,
each robot $r_i$ computes the ``local view''
of each robot $r_j$
which is determined only by configuration $P$ 
independently of its local coordinate system $Z_i$,
although $r_i$ observes $P$ in $Z_i$. 

Local views of robots defined in this section satisfy the following
properties:
\begin{enumerate}
\item 
For each $P_i$ ($i = 1, 2, \ldots , m$), 
all robots in $P_i$ have the same local view.
\item
Any two robots, one in $P_i$ and the other in $P_j$, 
have different local views, for all $i \not= j$.
\end{enumerate}
Then we give an expression of a local view as a sequence of
positions of the robots and by using the lexicographic ordering of local views,
the robots agree on a total ordering among $\{ P_1, P_2, \ldots, P_m \}$, 
i.e., $P_i$ is smaller than $P_j$ if and only if the local view of 
some $p \in P_i$ is smaller than that of some $p' \in P_j$
in the lexicographic order. 

To define the {\em local view} of a robot,
we first introduce {\em amplitude, longitude} and {\em latitude}.
Let $P = \{p_1, p_2, \ldots, p_n\}$ be a configuration,
where $p_i$ is the position (in $Z_0$) of $r_i$.
Assume that $P$ is not contained in a plane and $b(P) \not\in P$, 
because otherwise the plane formation is trivially solvable
as we will show later.\footnote{
We note that when the robots are on a plane (especially, when the
robots are on a line), we cannot define the local view in the same way.} 
The {\em innermost empty ball} $I(P)$ is the ball centered at $b(P)$ 
and contains no point in $P$ in its interior and 
contains at least one point in $P$ on it. 
Since $b(P) \not\in P$, $I(P)$ is well-defined.
Intuitively, $r_i$ considers $I(P)$ as the earth, 
and the line containing $p_i$ and $b(P)$ as the earth's axis. 
Recall that $r_i$ can recognize its relative positions from the others,
since $Z_i(p_i) = (0,0,0)$ always holds.
The intersection of a line segment $\overline {p_i b(P)}$ and $I(P)$ 
is the ``north pole'' $NP_i$. 
Then it chooses a robot $r_{m_i}$ not on the earth's axis
as its {\em meridian robot}.
Indeed, there is a robot satisfying the condition
by the assumption that the robots are not on one plane.
The meridian robot should be chosen more carefully for our purpose 
as shown later. 
Let $MP_i$ be the intersection of a line segment 
$\overline{p_{m_i}b(P)}$ and $I(P)$. 
The large circle on $I(P)$ containing $NP_i$ 
and $MP_i$ defines the ``prime meridian''.
Specifically, the half arc starting from $NP_i$ and containing $MP_i$ 
is the prime meridian. 
Robot $r_i$ translates its local observation $Z_i(P)$
with geocentric longitude, latitude, and altitude. 
The position of a robot $r_j \in R$ is now 
represented by the 
{\em altitude} % $B9bEY(B 
$h_j$ in $[0,1]$, 
{\em longitude} % $B7PEY(B
$\theta_j$ in $[0, 2 \pi)$, and 
{\em latitude} % $B0^EY(B
$\phi_j$ in $[0, \pi]$.
Here the altitude of a point on $I(P)$ is 0,
and that on $B(P)$ is $1$.
The longitude of $MP_i$ is 0,
and the positive direction is the counter-clockwise direction. 
Since the robots are all right-handed
they can agree on the counter clockwise direction
(i.e., rotating positive $x$-axis to positive $y$-axis) 
on $I(P)$ by using $b(P)$. 
For example, the robots can agree on the clockwise direction 
by considering that the negative $z$-axis of 
their local coordinate systems point to $b(P)$. 
Finally, the latitudes of the  ``north pole'' $NP_i$,  
the ``equator,'' and the ``south pole'' are 0, $\pi/2$,
and $\pi$, respectively.

Now $p_j$ is represented by a triple $p_j^* = (h_j, \theta_j, \phi_j)$
(or more formally, $r_i$ transforms $Z_i(p_j)$ to $p_j^*$) 
for all $j = 1, 2, \ldots , n$,
where $\theta_i = \phi_i = 0$ by definition.
Observe that $p_j^*$ depends on the choice of the meridian robot $r_{m_i}$
and $p_j^* \not= p_{\ell}^*$ if and only if $p_j \not= p_{\ell}$.
See Figure~\ref{fig:geosky} as an example.

%%%%%%%%%%%%%%%%%%%%%%%%%%%%%%%%%%%%%%%%%%%%%%%%%%%%%%%%%%%%%%%%%%%%%%%%%%
\begin{figure}[t]
\centering 
\includegraphics[width=4cm]{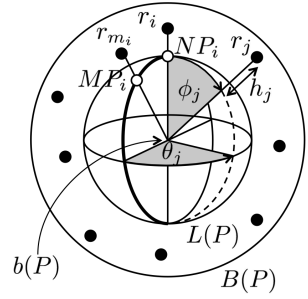}
\caption{Amplitude, longitude, and latitude calculated from
$r_i$'s local observation. 
The prime meridian for $r_i$ is drawn by bold arc. 
The position of $r_j$ is now represented by a triple
$p_j^* = (h_j, \theta_j, \phi_j)$.}
\label{fig:geosky}
\end{figure}
%%%%%%%%%%%%%%%%%%%%%%%%%%%%%%%%%%%%%%%%%%%%%%%%%%%%%%%%%%%%%%%%%%%%%%%%%%

We then use the lexicographic ordering $<$ among the positions $p_j^*$
to compare them:
For two positions $(h, \theta, \phi)$ and $(h', \theta', \phi')$, 
$(h, \theta, \phi) < (h', \theta', \phi')$ if and only if
(i) $h < h'$, (ii) $h = h'$ and $\theta < \theta'$, or 
(iii) $h = h'$, $\theta = \theta'$ and $\phi < \phi'$. 

Let $V_i^* = \langle p_i^*, p_{m_i}^*, p_{j_1}^*, p_{j_2}^*, \ldots, p_{j_{n-2}}^* \rangle$ 
be a sorted list of the positions $p_j^*$,
in which the positions $r_i$ and its meridian robot $r_{m_i}$
are placed as the first and the second elements 
and the positions $p_j^*$ of the other robots $r_j$ are placed
in the increasing order,
i.e., $p_{j_k}^* \leq p_{j_{k+1}}^*$ for all $k = 1, 2, \ldots , k-3$, and 
$\{ j_1, j_2, \ldots, j_{n-2} \} = \{ 1, 2, \ldots , n\}
\setminus \{ i, m_i \}$ where the ties are arbitrarily resolved. 

Let us return to the problem of how to choose the meridian robot $r_{m_i}$.
As explained, $V_i^*$ depends on the choice of $r_{m_i}$.
Robot $r_i$ computes the robot that minimizes $V_i^*$ in the
lexicographical order 
and chooses it as the meridian robot $r_{m_i}$, 
where a tie is resolved arbitrarily. 
We call this minimum $V_i^*$ (for $r_{m_i}$ chosen in this way)
the {\em local view} of $r_i$.
Regardless of the choices of meridian robot $r_{m_i}$ by robot $r_i$,
the next lemma holds.

\begin{lemma}
Let $P \in {\mathcal P}_n^3$ and $\{ P_1, P_2, \ldots, P_m \}$ be 
a configuration of robots represented as a set of points 
and its $\gamma(P)$-decomposition, respectively. 
Then we have the following two properties: 
\begin{enumerate}
\item 
For each $P_i$ ($i = 1, 2, \ldots , m$), 
all robots in $P_i$ have the same local view.
\item
Any two robots, one in $P_i$ and the other in $P_j$, 
have different local views, for all $i \not= j$.
\end{enumerate}
\end{lemma}

\begin{proof}
The first property is obvious by the definitions of 
$\gamma(P)$-decomposition and local view,
since for any $p, q \in P_i$ there is an element
$g \in \gamma(P)$ such that $q = g*p$.

As for the second property,
to derive a contradiction,
suppose that there are distinct integers $i$ and $j$,
such that robots $r_k \in P_i$ and  $r_{\ell} \in P_j$ have the same local view.
That is, $V_k^* = V_{\ell}^*$.
Let us consider a function $f$ that maps the $d$-th element of
$V_k^*$ to that of $V_{\ell}^*$.
More formally, letting the $d$-th element of $V_k^*$ (resp. $V_{\ell}^*$)
be $p_x^*$ (reps. $p_y^*$), 
$f$ maps $p_x$ to $p_y$.
Then $f$ is a congruent transformation that keeps $b(P)$ unchanged
by the definition of local view,
i.e., $f$ is a rotation in $\gamma(P)$,
which contradicts to the definition of $\gamma(P)$-decomposition.
\shortqed
\end{proof}

\begin{corollary}
\label{lemma:ordering}
Let $P \in {\mathcal P}_n^3$ and $\{ P_1, P_2, \ldots, P_m \}$ be 
a configuration of robots represented as a set of points and 
its $\gamma(P)$-decomposition, respectively.
Then the robots can agree on a total ordering among these subsets. 
\end{corollary}

\begin{proof}
By using the lexicographical ordering of the local views
of robots in each element of the $\gamma(P)$-decomposition of $P$.
\shortqed
\end{proof}

We now conclude Theorem~\ref{theorem:decomposition}
by Lemma~\ref{lemma:subsets} and Corollary~\ref{lemma:ordering}.
In the following, we assume that the $\gamma(P)$-decomposition of $P$,
$\{P_1, P_2, \ldots, P_m\}$ is ordered in this way. 
From the definition, $P_1$ is on $I(P)$, $P_m$ is on $B(P)$, and
$P_i$ is in the interior or on the ball that is centered at $b(P)$
and contains $P_{i+1}$ on it. 

We go on to the analysis of the structure of a transitive
set of points regarding a 3D rotation group. 
Recall that a transitive set of points is spherical. 
Any transitive set of points $P$ is specified
by a rotation group $G \in {\mathbb S}$
and a seed point $s$ as the orbit $Orb(s)$ 
of the group action of $G$ through $s$, 
so that $G = \gamma(P)$ holds.
Not necessarily $|G| = |Orb(s)|$ holds.
For any $p \in P$, we call
$\mu(p) = |\{ g \in G : g*s = p \}|$ the {\em folding} of $p$.
We of course count the identity element of $G$ for $\mu(p)$ 
and $\mu(p) \geq 1$ holds for all $p \in P$.\footnote{
In group theory, 
the folding of a point $P$ is simply the size of 
the stabilizers of $p$ defined by $G(p) = \{ g \in G: g*p = p \}$.
Although the lemma is known in group theory (see e.g., \cite{A88}),
we provide a proof for the convenience of readers.}

\begin{lemma}
\label{lemma:folding}
Let $P$ be the transitive set of points generated 
by a rotation group $G \in \{T, O, I\}$ and a seed point $s \in \Real^3$. 
If $p \in P$ is on a $k$-fold axis of $G$ for some $k$,
so are the other points $q \in P$ 
and $\mu(p) = \mu(q) = k$ holds.
Otherwise, if $p \in P$ is not on any axis of $G$, 
so are the other points $q \in P$ 
and $\mu(p) = \mu(q) = 1$ holds.
\end{lemma}

\begin{proof} 
We first show that $\mu(p) = \mu(q)$ for any $p,q \in P$.
To derive a contradiction,
we assume $\mu(p) > \mu(q)$ for some $p,q \in P$.
Let $g_1, g_2, \ldots, g_{m_p}$ 
(resp. $h_1, h_2, \ldots, h_{m_q}$)
be the set of rotations in $G$ such that $g_i * s = p$ 
(resp. $h_i * s = q$) holds for $i = 1, 2, \ldots , \mu(p)$
(resp. $i = 1, 2, \ldots , \mu(q)$).
Clearly $g_i \neq h_j$ for any $i$ and $j$.
Let $g \in G$ be a rotation satisfying $q = g*p$,
which definitely exists by definition.
Hence $q = (g \cdot g_i)*s$ for all $i = 1, 2, \ldots , \mu(p)$,
a contradiction,
since $g \cdot g_i \not= g \cdot g_j$ if $i \neq j$,
and $\mu(q) \geq \mu(p)$ holds.

Note that the seed point $s$ can be taken as $p$ in the above proof.
Suppose that $s$ is on a $k$-fold axis of $G$, 
then $\mu(s) = k$,
since the rotations in $G$ that move $s$ to itself are
the rotations around this $k$-fold axis.

Otherwise if $s$ is not on a rotation axis of $G$, 
only the identity element of $G$ can move $s$ to itself 
and hence $\mu(s) = 1$.
\shortqed
\end{proof}

\begin{lemma}
\label{lemma:cardinality}
When a set of points $P$ is transitive 
regarding $\gamma(P) \in \{T, O, I\}$, 
then we have $|P|\in \{4, 6, 8, 12, 20, 24, 30, 60\}$. 
\end{lemma}

\begin{proof}
By Lemma~\ref{lemma:folding}, 
we can compute the cardinality of any transitive set of points 
for each rotation group. 

The tetrahedral group $T$ consists of $2$-fold axes and 
$3$-fold axes, and its order is $12$. 
If we put a seed on a $2$-fold axis,
we obtain a $6$-set as $P$ forming a regular octahedron. 
If we put a seed on a $3$-fold axis,
we obtain a $4$-set as $P$ forming a regular tetrahedron. 
If we put a seed not on any axis, 
we obtain a $12$-set as $P$. 

By the same argument, we have the following results:
The order of the octahedral group $O$ is $24$ 
and the possible cardinalities of $P$ are $6, 8, 12$, and $24$. 
The order of the icosahedral group $I$ is $60$  
and the possible cardinalities of $P$ are $12, 20, 30$, and $60$. 
\shortqed
\end{proof}

By Lemmas~\ref{lemma:folding} and Lemma~\ref{lemma:cardinality}, 
folding of a point determines the positions of a 
transitive set of points in the arrangement of rotation axes 
and these polyhedra are shown in Table~\ref{table:vt-sets}. 
When the folding is $1$,
a seed point can be taken any point not on any rotation axis 
and depending on the seed point,
infinite number of different polyhedra are obtained.\footnote{
Table~\ref{table:vt-sets} does not contain all uniform polyhedra. 
There are uniform polyhedra consisting of $48$ vertices or $120$ vertices,
such as a rhombitruncated cuboctahedron with $48$ vertices
and a rhombitruncated icosidodecahedron with $120$ vertices. 
However, they require a mirror plane to induce transitivity 
and the robots with right-handed local coordinate systems
can partition them into two groups.
For example, a rhombitruncated cuboctahedron is decomposed into two
$24$-sets by its rotation group $O$. } 
We have the following property 
by the definition of $\gamma(P)$-decomposition of a set of points $P$.

\begin{property}
Let $P \in {\mathcal P}_n^3$ and $\{ P_1, P_2, \ldots, P_m \}$ be 
a set of points and its $\gamma(P)$-decomposition, respectively. 
Then if $\gamma(P)$ is a 3D rotation group, 
$P_i$ is one of the polyhedra shown in Table~\ref{table:vt-sets}
for $i = 1, 2, \ldots , m$.
\end{property}

\begin{table}[t]
\begin{center}
\caption{Transitive sets of points in 3D-space (i.e., polyhedra)
characterized by rotation group and folding.}
\label{table:vt-sets}
\begin{tabular}{c|c|c|c|l}
\hline 
Rotation group & Order & Folding & Cardinality & Polyhedra \\ 
\hline 
 & & 3 & 4 & Regular tetrahedron \\ 
$T$ & 12 & 2 & 6 & Regular octahedron \\ 
 & & 1 & 12 & Infinitely many polyhedra \\ 
\hline 
\multirow{4}{*}{$O$} & \multirow{4}{*}{24}& 4 & 6 & Regular octahedron \\ 
 & & 3 & 8 & Cube \\ 
 & & 2 & 12 & Cuboctahedron \\ 
 & & 1 & 24 & Infinitely many polyhedra \\ 
\hline 
\multirow{4}{*}{$I$} & \multirow{4}{*}{60} & 5 & 12 & Regular icosahedron \\ 
 &  &  3 & 20 & Regular dodecahedron \\ 
 &  & 2 & 30 & Icosidodecahedron \\ 
 &  & 1 & 60 & Infinitely many polyhedra \\
\hline
\end{tabular}
\end{center}
\end{table} 

% Yukiko: $B2<5-$O%-!<%W$7$F$*$-$^$9(B
%The notion of multiplicity is, in group theory, 
%the size of stabilizer of a point. 
%Let $P$ be a vertex-transitive set of points, i.e., 
%$P$ is a $\gamma(P)$-set. 
%Then $\gamma(P)_p = \{g \in \gamma(P)| g p = p \}$ is the
%stabilizer of $p$. 
%For the size of the $\gamma(P)$-orbit of $p_i \in P$, $|P_i|$, 
%we have the following equation~\cite{R94},  
%\begin{equation*}
%|P_i| = |\gamma(P)| / |\gamma(P)_{p_i}|
%\end{equation*}
%where the size of group $\gamma(P)$ is denoted by $|\gamma(P)|$. 

\section{Proofs of Theorem~\ref{theorem:main}}
\label{sec:PlF}

We show the proofs of Theorem~\ref{theorem:main}
in this section. 
In Subsection~\ref{subsec:nec}, 
we first show the necessity of Theorem~\ref{theorem:main} 
by showing that any algorithm for oblivious FSYNC robots cannot 
form a plane from an initial configuration if the initial configuration
does not satisfy the condition in Theorem~\ref{theorem:main}. 
Specifically, for any initial configuration $P$ that satisfies
$\gamma(P)$ is in $\{T, O, I\}$ 
and the size of each element of its $\gamma(P)$-decomposition 
is in $\{12, 24, 60\}$, 
we construct an arrangement of initial local coordinate systems 
that makes the robots keep the rotation axes of a 3D rotation group 
forever so that they never form a plane 
no matter which algorithm they obey. 
The orders of $T$, $O$, and $I$ are $12$, $24$, and $60$, respectively and
when an initial configuration does not satisfy the condition of
Theorem~\ref{theorem:main},
we can decompose the robots into transitive subsets so that 
the cardinality of each subset is ``full'' 
regarding a 3D rotation group (not necessarily $\gamma(P)$). 
Then we show that there exists an arrangement of 
local coordinate systems that is also transitive regarding the 
selected rotation group so that the robots continue symmetric 
movement forever. 
The impossibility proof holds for non-oblivious robots 
because starting from such a symmetric initial configuration $P$, 
the contents of memory at robots in the same element are kept identical 
and if the initial memory content of the robots are identical,
they cannot break the symmetry. 
Thus we obtain the necessity of Theorem~\ref{theorem:main}.

In Subsection~\ref{subsec:suff}, 
we show the sufficiency of Theorem~\ref{theorem:main} 
by presenting a plane formation algorithm 
for oblivious FSYNC robots. 
When $\gamma(P)$ of an initial configuration $P$ is a 2D rotation group,
the robots are on one plane  
or they can agree on the plane that is perpendicular to the single
rotation axis (or the principal axis). 
Actually, the robots can land on such a plane without making any 
multiplicity. 
On the other hand, 
when $\gamma(P)$ is a 3D rotation group,
the condition of Theorem~\ref{theorem:main} guarantees that 
there exists an element in the $\gamma(P)$-decomposition of $P$ 
that forms a regular tetrahedron, a regular octahedron, a cube, 
a regular dodecahedron, or an icosidodecahedron
(Table~\ref{table:vt-sets}). 
The proposed algorithm adopts the ``go-to-center'' strategy,
which is very similar to the ``go-to-midpoint'' algorithm in 
Subsection~\ref{subsec:idea}. 
Then we show that after the movement, 
the rotation group of the robots' positions is not a 3D rotation group
any more intuitively because the candidates of next positions form a 
transitive set of points, 
while the number of the robots is not sufficient 
to select a set of points with 3D rotation group from such set of
points. 
Because their rotation group is a 2D rotation group, 
the robots can form a plane. 
Clearly non-oblivious FSYNC robots can execute the proposed algorithm  
and we obtain the sufficiency of Theorem~\ref{theorem:main}.

\subsection{Necessity}
\label{subsec:nec} 

Provided $|P| \in \{ 12, 24, 60\}$,
we first show that when a set of points $P$ is 
a transitive set of points regarding a 3D rotation group, 
there is an arrangement of local coordinate system 
$Z_i$ for each robot $r_i \in R$ such that the execution from $P$ 
keeps a 3D rotation group forever 
no matter which algorithm the oblivious FSYNC robots obey.

\begin{lemma}
 \label{lemma:spherical-nec}
 Consider $n$ oblivious FSYNC robots with $n \in \{12, 24, 60\}$.
Then the plane formation problem is unsolvable
from an initial configuration $P$ 
if $P$ is a transitive set of points 
regarding a 3D rotation group. 
\end{lemma}

\begin{proof}
Let $P(0)$ be an initial configuration of $n \in \{12, 24, 60\}$ robots
that is transitive
regarding $\gamma(P(0)) \in \{T, O, I\}$. 

To derive a contradiction,
we assume that there is an algorithm $\psi$ that enables the 
robots to solve  
the plane formation problem for any choice of initial arrangement of 
local coordinate systems of robots.
We will show that there is an initial arrangement of 
local coordinate systems such that 
the robots move symmetrically and keep the axes of 
rotation group $G$ forever,
where $G$ is given as follows depending on $n$:
\begin{equation*} 
 G = \left\{\begin{array}{ll} 
 T & \ \ \mbox{if} \ n = 12, \\ 
 O & \ \ \mbox{if} \ n = 24, \\ 
 I & \ \ \mbox{if} \ n = 60. 
 \end{array} \right.
\end{equation*}

We first claim that there is always an embedding of $G$ to $\gamma(P(0))$.
The claim obviously holds when $G = \gamma(P(0))$.
Suppose $G \neq \gamma(P(0))$.
Then $n = 12$,
since otherwise (i.e., $n$ is either $24$ or $60$), 
$G =  \gamma(P(0))$ by Table~\ref{table:vt-sets}
and by the definition of $G$.
If $n = 12$, then $G = T$ by the definition of $G$. 
Since $\gamma(P(0)) \in \{ O, I \}$,
the claim holds.

We fix an arbitrary embedding of $G$ to $\gamma(P(0))$.
For any point $s \in P(0)$,
we next claim $P(0) = Orb(s) = \{ g*s : g \in G \}$
and $|P(0)|$ is the order of $G$, 
i.e., $\mu(s)=1$.
Obviously the claim holds when $G = \gamma(P(0))$ from the definition.
Suppose that $G \not= \gamma(P(0))$.
Then $n = 12$, $G = T$ and $\gamma(P(0)) \in \{ O, I \}$
by the argument above.
If $\gamma(P(0)) = O$,
all points in $P(0)$ are on $2$-fold axes of $O$
from Table~\ref{table:vt-sets},
but there is no embedding of $T$ to $O$ that makes 
the rotation axes of $T$ overlap 
$2$-fold axes of $O$. 
That is, $\mu(s)$ regarding $T$ is $1$. 
Otherwise if $\gamma(P(0)) = I$,
like the above case,
all points in $P(0)$ are on $5$-fold axes of $I$
from Table~\ref{table:vt-sets}, 
but there is no embedding of $T$ to $I$ that makes
the rotation axes of $T$ overlap $5$-fold axes of $I$. 
That is, $\mu(s)$ of $T$ is $1$. 

Now we define a local coordinate system $Z_i$ for each $r_i \in R$
by using $Z_1$, the local coordinate system of $r_1 \in R$,
so that any algorithm $\psi$ produces an execution 
${\cal E:} P(0), P(1), \ldots$ such that $G$ is a subgroup of 
$\gamma(P(t))$ for all $t = 0, 1, \ldots$.
We define $Z_1 = Z_0$  
and $Z_1$ is specified by $(0,0,0),(1,0,0),(0,1,0),(0,0,1)$. 
Let $P(t) = \{ p_1(t), p_2(t), \ldots , p_n(t) \}$,
where $p_i(t)$ is the position of $r_i$ at time $t \geq 0$.
For each $r_i \in R$, 
there is an element $g_i \in G$ such that $p_i(0) = g_i * p_1(0)$,
and this mapping between $r_i$ and $g_i$ is a bijection between $R$ and $G$,
i.e.,  $g_i \neq g_j$ if $i \not= j$, and $G = \{g_i|r_i \in R\}$ 
because $\mu(r_i)=1$. 
Thus $g_1$ is the identity element.
Local coordinate system $Z_i$ is specified by the positions of
its origin $(0,0,0)$, $(1,0,0)$, $(0,1,0)$ and $(0,0,1)$ in $Z_0$.
That is, we can specify $Z_i$ by a quadruple
$(o_i, x_i, y_i, z_i) \in (\Real^3)^4$.
Define $Z_i$ as the coordinate system specified 
by a quadruple $(g_i*(0,0,0),g_i*(1,0,0),g_i*(0,1,0),g_i*(0,0,1))$,
for $i = 1, 2, 3, \ldots , n$.
\footnote{Recall that $Z_i$ here means $Z_i$ at time 0.}

Then $Z_i(P(0)) = Z_1(P(0))$ for $i = 1,2, \ldots , n$ 
and $\psi$ outputs the same value $\psi (Z_i(P(0))) = d$
in every robot $r_i$ as its next position. 
Let $d_i$ be this output at $r_i$ observed in $Z_0$. 
Then we have $d_i = g_i*d_1$.
That is, $P(1) = \{ d_1, d_2, \ldots , d_n \}$ is the orbit of $G$
through $d_1$
and obviously $G$ is a subgroup of $\gamma(P(1))$.
By an easy induction, 
we can show that $\gamma(P(t)) \succeq G$ 
is a 3D rotation group for $t = 0,1, \ldots$.

We finally address multiplicity during any execution of $\psi$. 
Algorithm $\psi$ may move some robots to one point at some time $t$.
Because $\gamma(P(t)) \succeq G$,
all robots gather at one point. 
However, since $\psi$ further needs to move the robots 
to distinct positions by the definition of the plane formation problem,
$\psi(Z_i(P(t))) \not= 0$ must hold,
that is, $\psi$ outputs a point that is different from the current position 
(i.e., the origin of $Z_i$) as the next position 
and these destinations form a transitive set of points regarding 
$G$ or its supergroup in $\gamma(P(t+1))$. 
Thus the robots never form a plane. 
\shortqed 
\end{proof}

Lemma~\ref{lemma:spherical-nec} considers 
an arbitrary transitive initial configurations regarding
a 3D rotation group. 
We next extend it to handle general initial configurations,
which may not be transitive.
Let  $\{ P_1, P_2, \ldots, P_m \}$ be the $\gamma(P)$-decomposition of
an initial configuration $P$.
Intuitively, we wish to specify $Z_j$ for $p_j \in P_i$ in the same way 
as the proof of Lemma~\ref{lemma:spherical-nec} for each $P_i$
($i = 1, 2, \cdots, m$).
We however need to take into account the cases in which $|P_i| \not= |P_j|$ 
and $G$ for $P_i$ is different from the one for $P_j$.
For example, consider a configuration $P$ consisting of a regular icosahedron 
($12$ points) and a truncated icosahedron ($60$ points),
where $\gamma(P) = I$. 
Then the $I$-decomposition of $P$ consists of the regular icosahedron $P_1$
and the truncated icosahedron $P_2$,
and $G$ for $P_1$ is $T$, while it is $I$ for $P_2$. 
In this case,
we make use of the $T$-decomposition (instead of the $I$-decomposition)
of $P$ and apply Lemma~\ref{lemma:spherical-nec} to each element
of the $T$-decomposition of $P$. 
Then we show that any execution keeps the rotation axes of $T$ forever.

\begin{theorem}
\label{theorem:nec}
Let $P$ and $\{ P_1, P_2, \ldots, P_m \}$ be an initial configuration
and the $\gamma(P)$-decomposition of $P)$, respectively.
Then the plane formation problem is unsolvable from $P$ 
for oblivious FSYNC robots,
if $\gamma(P)$ is a 3D rotation group 
and $|P_i| \in \{12, 24, 60\}$ for $i = 1, 2, \ldots , m$.
\end{theorem}

\begin{proof}
Let $P_1, P_2, \ldots, P_m$ be the $\gamma(P(0))$-decomposition of 
an initial configuration $P(0)$. 
We define the rotation group $G$ by:
\begin{equation*} 
 G = \left\{\begin{array}{ll} 
 T & \ \ \mbox{if} \ \min_{i=1, 2, \ldots, m}\{|P_i|\} = 12, \\ 
 O & \ \ \mbox{if} \ \min_{i=1, 2, \ldots, m}\{|P_i|\} = 24, \\ 
 I & \ \ \mbox{if} \ \min_{i=1, 2, \ldots, m}\{|P_i|\} = 60. 
 \end{array} \right.
\end{equation*} 
We show that there exists an arrangement of local coordinate systems of 
 robots that makes the robots keep the rotation axes of $G$ forever
 regardless of the algorithm they obey. 

By Table~\ref{table:vt-sets}, 
$G = \gamma(P(0))$ or $G$ is a subgroup of $\gamma(P(0))$ 
and there is an embedding of $G$ to $\gamma(P(0))$. 
We fix an arbitrary embedding of $G$ to $\gamma(P(0))$,
and consider the $G$-decomposition of $P(0)$ 
which is defined in the same way as the $\gamma(P(0))$-decomposition.
Formally, consider the orbit space 
$\{ Orb(p) : p \in P(0) \} = \{ Q_1, Q_2, \ldots , Q_k \}$
regarding $G$. 

For example,
let $P$ be a cuboctahedron embedded in a truncated cube as illustrated
in Figure~\ref{fig:TinO-1}.
Then $\gamma(P) = O$.
The $\gamma(P)$-decomposition is $\{P_1, P_2 \}$,
where the cardinalities of the elements are 12 and 24.
By definition, $G = T$.
The $G$-decomposition of $P$ is $\{Q_1, Q_2, Q_3\}$,
that is obtained with seed points $s_1, s_2$, and $s_3$, and 
the orbit $Q_i = Orb(s_i)$ regarding $T$ through $s_i$
for $i = 1, 2, 3$. 
(See Figure~\ref{fig:TinO-2}.)
 
%%%%%%%%%%%%%%%%%%%%%%%%%%%%%%%%%%%%%%%%%%%%%%%%%%%%%%%%%%%%%%%%%%%%%%%%%%
\begin{figure}[t]
\centering 
\subfigure[]{\includegraphics[width=3.2cm]{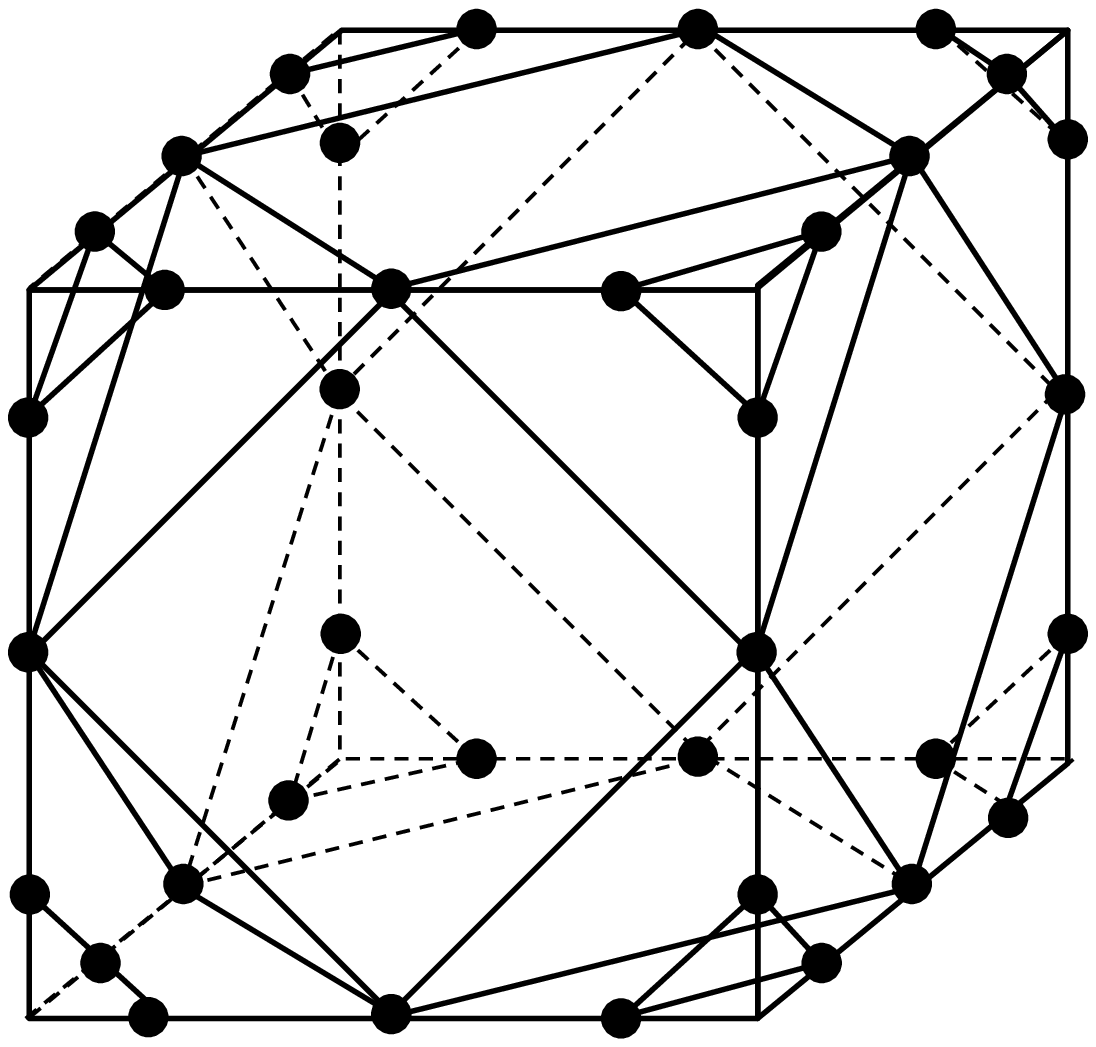}\label{fig:TinO-1}}
\hspace{5mm}
\subfigure[]{\includegraphics[width=3.2cm]{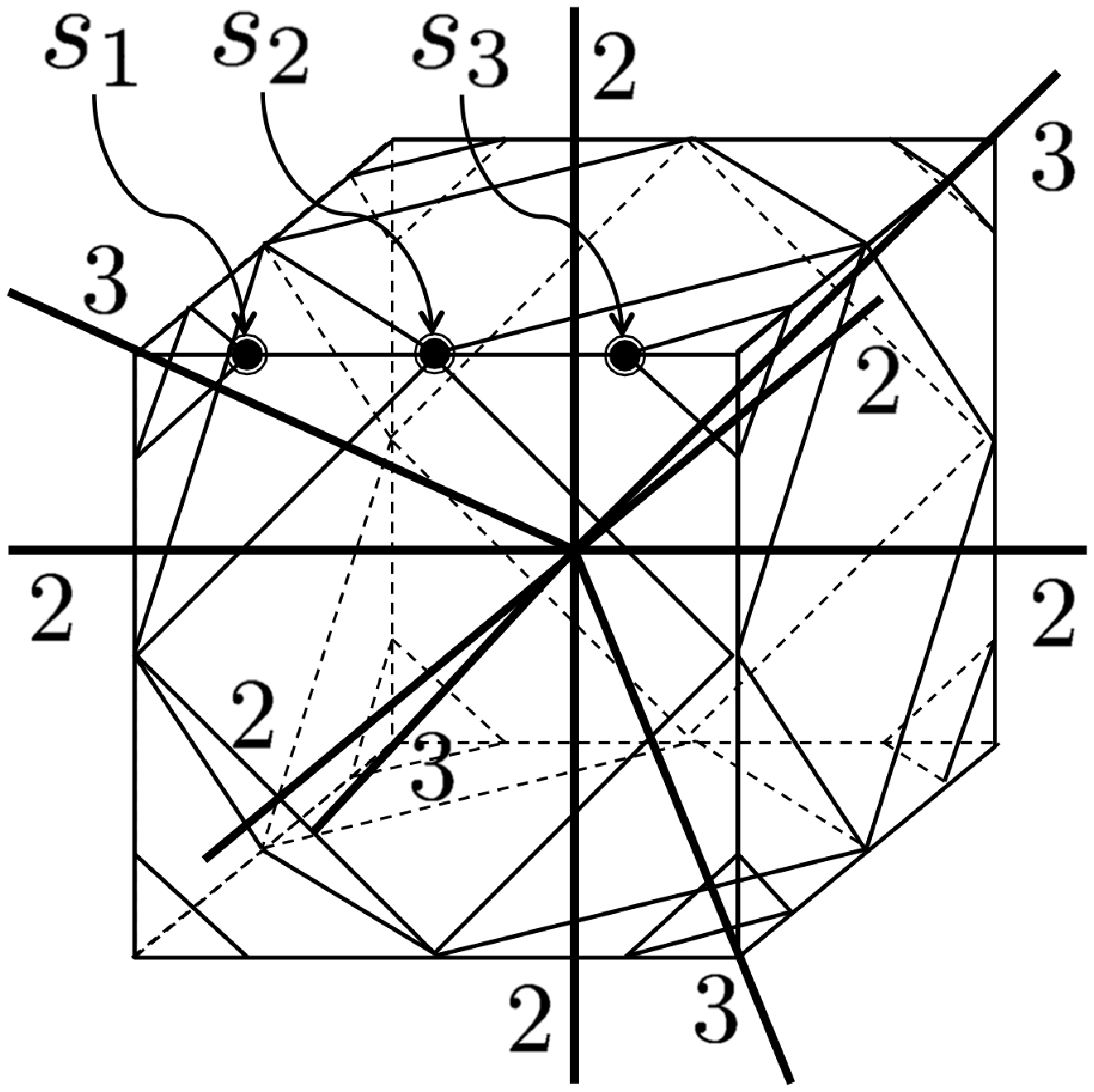}\label{fig:TinO-2}}
\hspace{5mm}
\\ 
\subfigure[]{\includegraphics[width=3.2cm]{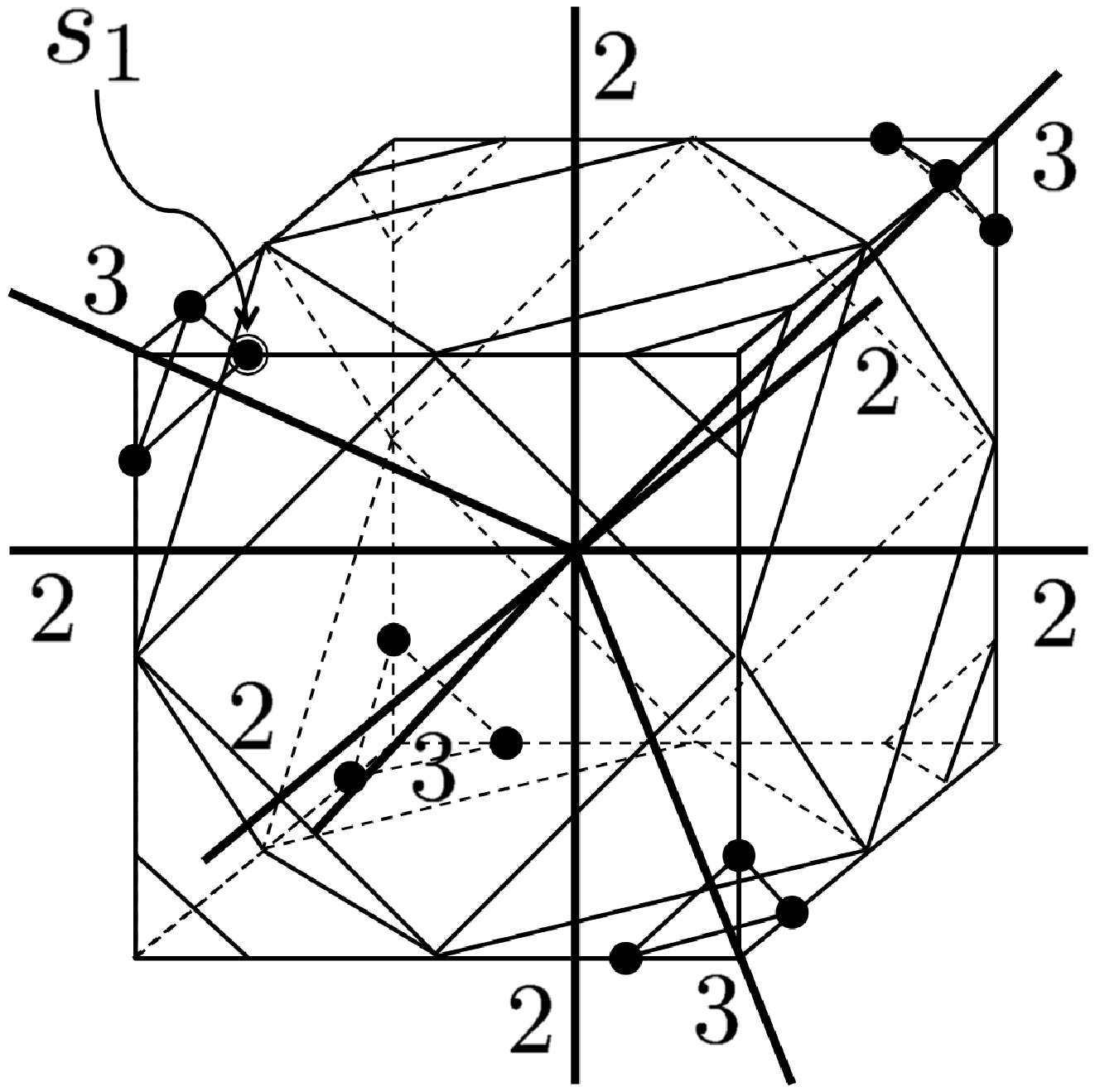}\label{fig:TinO-3}}
\hspace{5mm}
\subfigure[]{\includegraphics[width=3.2cm]{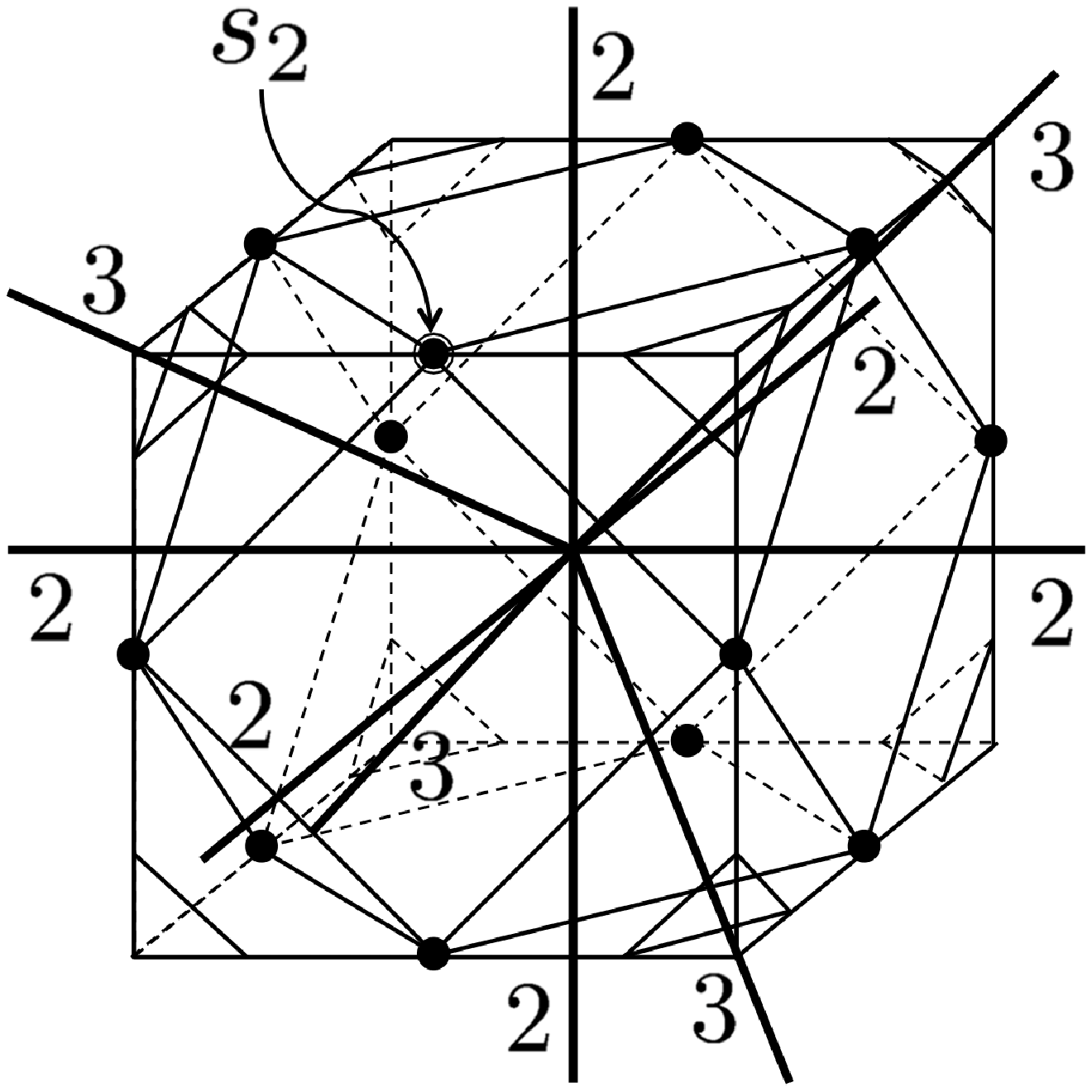}\label{fig:TinO-4}}
\hspace{5mm}
\subfigure[]{\includegraphics[width=3.2cm]{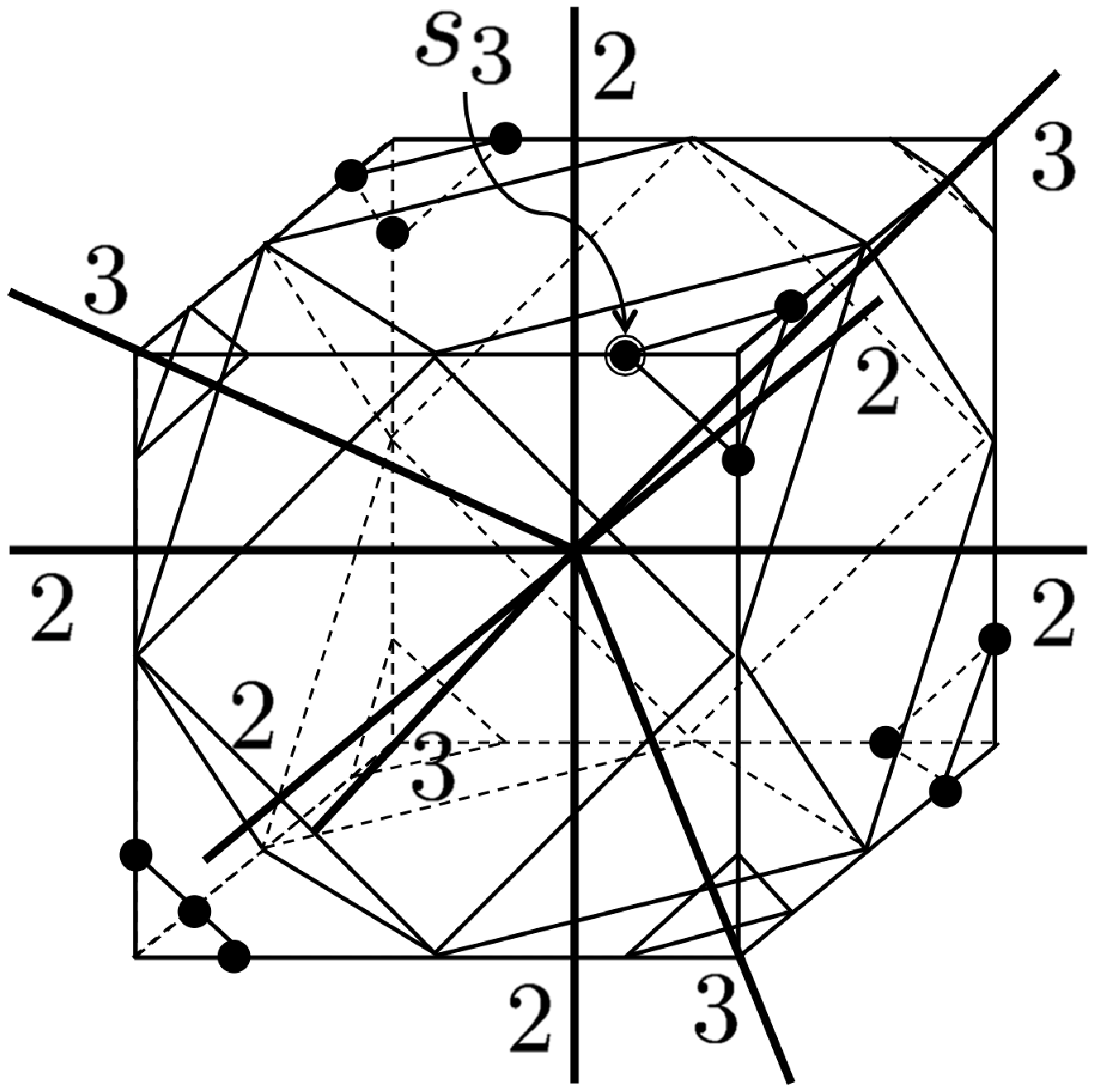}\label{fig:TinO-5}}
\hspace{5mm}
\caption{(a) $P$ consists of a cuboctahedron and a truncated cube 
so that its $\gamma(P) = O$.
(b) Three seed points $s_1, s_2, s_3 \in P$.
(c) $Orb(s_1)$ of $T$ through $s_1$, whose cardinality is 12.
(d) $Orb(s_2)$ of $T$ through $s_2$, whose cardinality is 12.
(e) $Orb(s_3)$ of $T$ through $s_3$, whose cardinality is 12.}
\label{fig:TinO}
\end{figure}
%%%%%%%%%%%%%%%%%%%%%%%%%%%%%%%%%%%%%%%%%%%%%%%%%%%%%%%%%%%%%%%%%%%%%%%%%%

We first show that for each $Q_i$ ($i = 1, 2, \cdots, k$),
$|Q_i| = |G|$, thus for any $q \in Q_i$, $\mu(q)=1$ regarding $G$. 
Let $C$ be the set of sizes of the elements of the
$\gamma(P(0))$-decomposition of $P(0)$, that is, 
$C = \{|P_i|: i=1, 2, \ldots, m\}$.
Observe that $\{24, 60\} \not\subseteq C$,
since while $60 \in C$ implies $\gamma(P(0)) = I$, 
there is no transitive set of points $S$
with $|S| = 24$ regarding $I$ by Lemma~\ref{lemma:cardinality}.
Hence $1 \leq |C| \leq 2$.
Depending on $|C|$,
we have the following three cases. 

\noindent{\bf Case A: $|C| = 1$.~}
The case $G = \gamma(P(0))$ is trivial.
When $G \neq \gamma(P(0))$, 
we must consider the following two cases. 

{\bf Case A1:~} When $G = T$ and $\gamma(P(0)) = O$. 
Then $|G| = |T| = 12$.
Let $p_j \in P_i$ be any point for $i = 1, 2, \ldots , m$.
By definition $P_i = Orb(p_j)$ regarding $\gamma(P(0)) = O$
for some $p_j \in P_i$.
Observe that under an arbitrary 
embedding of $T$ to $\gamma(P(0))=O$,
$p_j$ is not on any rotation axis of $T$,
since otherwise, $\mu(p_j)$ regarding $O$ is $3$ or $4$,
and $|P_i|$ is $8$ or $6$. 
Consequently, there is no point in $P$ that is 
on a rotation axis of any embedding of $T$ to $\gamma(P(0))$.
Thus we have $|Q_i| = |T| = 12$ for $i = 1, 2, \ldots , k$.

{\bf Case A2:~} When $G = T$ and $\gamma(P(0)) = I$. 
Then $|G| = |T| = 12$.
The proof is exactly the same as (A1),
except that, in this case,
we observe that there is no point in $P$ that is 
on a rotation axis of any embedding of $T$ to $\gamma(P(0))$, 
since otherwise $|P_i|$ is $30$ or $20$.
Thus we have $|Q_i| = |T| = 12$ for $i = 1, 2, \ldots , k$.

\noindent{\bf Case B: $|C| = 2$.~}
Then $G \neq \gamma(P)$ 
and we have the following two cases. 

{\bf Case B1:~} When $C = \{12, 24\}$. 
Then $G = T$ and $\gamma(P(0)) = O$. 
Like Case A1, 
under an arbitrarily fixed embedding of $T$ to $\gamma(P(0))=O$,
any $p \in P(0)$ is not on a rotation axis of $T$ 
and $|Q_i| = |T| = 12$ for $i = 1, 2, \ldots , k$.

{\bf Case B2:} When $C = \{12, 60 \}$. 
Then $G = T$ and $\gamma(P(0)) = I$. 
Like Case A2,
any $p \in P(0)$ is not on a rotation axis of any embedding of $T$ to
$\gamma(P(0))$ 
and $|Q_i| = |T| = 12$ for $i = 1, 2, \ldots , k$.

 Thus we conclude that $|Q_i| = |G|$ and for $q \in Q_i$,
 $\mu(q)=1$ for $i = 1,2, \ldots , k$.

To derive a contradiction,
we assume that there is an algorithm $\psi$ that makes the robots 
form a plane from $P(0)$. 
The scenario to derive a contradiction is exactly the same as the
proof of Lemma~\ref{lemma:spherical-nec}. 
For each element $Q_i$ of the $G$-decomposition of $P(0)$
($i = 1, 2, \cdots, k$),
we pick up an arbitrary local coordinate system of a robot in $Q_i$ 
and by applying the elements of $G$ to the local coordinate system,
we obtain symmetric local coordinate systems of robots in $Q_i$
because $|Q_i| = |G|$.
In the same way as Lemma~\ref{lemma:spherical-nec},
each $Q_i$ keeps $G$ forever irrespective of the algorithm that the robots
forming $Q_i$ execute. 
Hence the rotation group of the robots contains $G$ as a subgroup
forever and the robots never form a plane since $G \in \{T, O, I\}$. 
\shortqed 
\end{proof}

Finally, we obtain the impossibility result for 
non-oblivious robots by Theorem~\ref{theorem:nec}, 
since starting from an initial configuration $P$, 
that satisfies the condition of Theorem~\ref{theorem:nec}, 
the robots in the same element of the $G$-decomposition 
keep the identical memory contents forever, i.e., 
in each configuration, 
they obtain the identical local observation and the 
identical output of computation. 
Thus from an initial configuration where the local memory is empty at
each robot,
they follow Theorem~\ref{theorem:nec}. 

\begin{theorem}
\label{theorem:nec-power}
Let $P$ and $\{ P_1, P_2, \ldots, P_m \}$ be an initial configuration
and the $\gamma(P)$-decomposition of $P$, respectively.
Then the plane formation problem is unsolvable from $P$ 
for non-oblivious FSYNC robots,
if $\gamma(P)$ is a 3D rotation group 
and $|P_i| \in \{12, 24, 60\}$ for $i = 1, 2, \ldots , m$.
\end{theorem}

As a concluding remark of this subsection,
we examine some initial configurations. 
By Theorem~\ref{theorem:nec}, irrespective of obliviousness, 
FSYNC robots cannot form a plane from 
an initial configuration where they form a
regular icosahedron because there exists an execution where
the robots keep $T$ forever. 
This corresponds an example in Section~\ref{subsec:idea}. 
Other unsolvable initial configurations with the minimum number of robots
are 
the initial configurations of $12$ robots forming
a vertex-transitive polyhedron regarding $T$,
e.g., a regular icosahedron, a truncated tetrahedron, 
a cuboctahedron, or 
infinitely many polyhedra generated by a seed point
which is not on any rotation axis of $T$. 
The FSYNC robots cannot form a plane from initial configurations
where they form a semi-regular polyhedron 
except the icosidodecahedron consisting of $30$ robots.

\subsection{Sufficiency}
\label{subsec:suff}

We show a pattern formation algorithm for oblivious FSYNC
robots for an arbitrary initial configuration
that satisfies the condition of Theorem~\ref{theorem:main}.
We first show the following theorem. 

\begin{theorem}
\label{theorem:suf}
Let $P$ and $\{ P_1, P_2, \ldots, P_m \}$ be an initial configuration
and the $\gamma(P)$-decomposition of $P$, respectively.
Then oblivious FSYNC robots can form a plane from $P$ if either
(i) $\gamma(P)$ is a 2D rotation group, or 
(ii) $\gamma(P)$ is a 3D rotation group
and there is a subset $P_i$ such that $|P_i| \not\in \{12, 24, 60\}$. 
\end{theorem}

Since non-oblivious FSYNC robot can execute any algorithm 
for oblivious FSYNC robots with ignoring its memory contents,
we have the following theorem. 

\begin{theorem}
\label{theorem:suf-power}
Let $P$ and $\{ P_1, P_2, \ldots, P_m \}$ be an initial configuration
and the $\gamma(P)$-decomposition of $P$, respectively.
Then non-oblivious FSYNC robots can form a plane from $P$ if either
(i) $\gamma(P)$ is a 2D rotation group, or 
(ii) $\gamma(P)$ is a 3D rotation group
and there is a subset $P_i$ such that $|P_i| \not\in \{12, 24, 60\}$. 
\end{theorem}

The proposed plane formation algorithm consists of
the symmetry breaking phase and the landing phase. 
A very rough idea behind the plane formation algorithm is the following:
If $\gamma(P)$ is a 2D rotation group, 
since there is a single rotation axis or a principal axis,
which the robots can recognize, 
the robots can agree on the plane perpendicular 
to this axis and containing $b(P)$ 
and proceed to a landing phase to land on distinct points on the plane.
When $\gamma(P)$ is $C_1$, the target plane is defined by
$P_1$ (i.e., single robot), its meridian robot, and $b(P)$. 

Suppose otherwise that $\gamma(P)$ is a 3D rotation group.
Then there is at least one element $P_i$ in the
$\gamma(P)$-decomposition $\{P_1, P_2, \ldots, P_m\}$ of $P$ 
such that $|P_i| \not\in \{12, 24, 60 \}$.
That is, $|P_i| < |\gamma(P_i)|$ holds 
and all robots in $P_i$ are on some rotation axes of $\gamma(P_i)$ from
Lemma~\ref{lemma:folding}. 
The symmetry breaking phase moves the robots in $P_i$ 
so that no robot is on the rotation axes of $\gamma(P_i)$.
This move cannot maintain $\gamma(P_i)$, 
since otherwise if $\gamma(P_i)$ is maintained in the new configuration,
the folding of any point would be 1 regarding $\gamma(P_i)$,
a contradiction. 
Such $P_i$ forms a regular tetrahedron, 
a regular octahedron, a cube,
a regular dodecahedron, or an icosidodecahedron 
by Table~\ref{table:vt-sets}.\footnote{As we will mention later, 
we assume $b(P) \not\in P$ for the simplicity of the algorithm.}
The symmetry breaking phase breaks the symmetry 
of these (semi-)regular polyhedra 
and as a result a configuration with a 2D rotation group yields.

The proposed algorithm solves the plane formation problem
in at most three cycles. 
The first cycle completes some preparations for the symmetry breaking
algorithm. 
The second cycle realizes the symmetry breaking phase by
translating the current configuration with a 3D rotation
group into another configuration with a 2D rotation group. 
The robots execute a symmetry breaking algorithm
similar to the ``go-to-midpoint'' algorithm in
Subsection~\ref{subsec:idea}. 
Then they agree on the plane that is perpendicular to the single
rotation axis (or the principal axis) and contains the
center of the smallest enclosing ball of themselves. 
The third cycle completes the landing phase. 
The landing algorithm we use in the third phase is conceptually easy
because the robots are FSYNC,
but contains some technical subtleties to land the robots 
to distinct positions on the plane.

The proposed algorithm consists of three algorithms
Algorithms~\ref{alg:sparse}, \ref{alg:break}, and \ref{alg:landing}, 
each of which consumes a single Look-Compute-Move cycle.
To formally describe these algorithms,
we define three conditions $T_1, T_2$, and $T_3$ on configuration $P$.
\begin{eqnarray*}
T_1(P) &=& (\gamma(P)~\text{is a 3D rotation group}) \Rightarrow 
(|P_1| \not\in \{12, 24, 60\}) \\ 
T_2(P) &=& (\gamma(P)~\text{is a 2D rotation group}) \\ 
T_3(P) &=& (\text{there is a plane $F$ such that $P \subset F$})
\end{eqnarray*}

Robot system $R$ solves the plane formation problem
if it reaches a configuration $P$ satisfying $T_3(P)$. 
The proposed algorithm guarantees that such $P$ satisfies $|P|=n$. 
The preparation algorithm (Algorithm~\ref{alg:sparse}) 
is executed in configuration $P$, 
if and only if $|P| = n$ and $\neg T_1(P)$ hold 
and a configuration $P'$ satisfying $T_1(P')$ and $|P'| = n$ yields.
The symmetry breaking algorithm (Algorithm~\ref{alg:break}) 
is executed in configuration $P'$,
if and only if $|P'| = n$ and $(T_1(P') \wedge \neg T_2(P'))$ holds 
and a configuration $P''$ satisfying $T_2(P'')$ and $|P''| = n$ yields.
Finally the landing algorithm (Algorithm~\ref{alg:landing}) 
is executed in configuration $P''$,
if and only if $|P''| = n$ and $(T_2(P'') \wedge \neg T_3(P''))$ hold 
and a configuration $P'''$ satisfying $T_3(P''')$ and $|P'''| = n$ yields.

It is worth emphasizing that since $T_j(P)$ for $j = 1, 2, 3$ does 
not depend on the local coordinate system $Z_i$ of a robot $r_i$. 
Since  $\neg T_1(P)$ implies $\neg T_2(P) \wedge \neg T_3(P)$
and $\neg T_2(P)$ implies $\neg T_3(P)$, 
(1) exactly one of the three algorithms is executed by the robots at any
configuration $P$ unless $T_3(P)$ holds and
(2) none of them is executable at any configuration $P$ if $T_3(P)$ holds;
the plane formation algorithm then terminates.

We formally define the set of terminal configurations of the proposed
algorithm. 
A configuration $P$ is a terminal configuration if it satisfies
$T_3(P)$.
For any execution $P(0), P(1), P(2), \cdots$ of the proposed algorithm,
if $P(t)$ is the first configuration that satisfies $T_3(P(t))$,
then the robots do not move thereafter from the definition of $T_1$,
$T_2$, and $T_3$. Thus the robots can easily agree on the termination
of the proposed algorithm. 

Note that if an initial configuration $P$ satisfies $T_3(P)$,
then the execution immediately terminates,
solving the plane formation problem trivially.

Although we defined an algorithm as a function $\psi$ from the set 
of configurations to a point in Subsection~\ref{SSmodel},
we mainly use English to describe it in what follows,
since an English description is usually more readable than the
mathematically defined function.

Recall that $P = \{ p_1, p_2, \ldots , p_n \}$ is a configuration,
where $p_i$ is the position of robot $r_i$ in $Z_0$.
Robot $r_i$ observes it in $Z_i$,
i.e., $r_i$ gets $Z_i(P) = \{ Z_i(p_1), Z_i(p_2), \ldots , Z_i(p_n) \}$
as its local observation.
However, $r_i$ can recognize its relative position in $P$,
since $Z_i(p_i) = (0,0,0)$ always holds.
For example, $r_i$ can decide if it is located at the center of $B(P)$.
In the following, 
we frequently use a robot $r_i$ and its position $p_i$ interchangeably,
that is, ``robot $p$'' means the robot located at a point $p$ 
and ``the robots in $Q \subseteq P$'' means
those located in a set of points $Q$.

For the simplicity of the algorithm, 
we assume that initial configuration $P$ satisfies 
$b(P) \not\in P$  since 
trivially the robots can translate any configuration $P$ such that 
$b(P) \in P$ to another configuration $P'$ such that $b(P') \not\in P'$ 
in one cycle 
by the robot on $b(P)$ moving to some point on the sphere centered at $b(P)$ 
and with radius $I(P)/2$.
From the resulting configuration $\gamma(P')$, 
the robots can form a plane as shown in the following
since $\gamma(P')$ is cyclic. 

\subsubsection{Algorithm for preparation}

The purpose of the preparation phase is to make the robots forming
one of the five (semi-)regular polyhedra to shrink toward the center of
the smallest enclosing ball of themselves,
so that the symmetry breaking algorithm is executed by these robots
with keeping the smallest enclosing ball. 
In a configuration $P$ that does not satisfy $T_1(P)$, 
the robots execute Algorithm~\ref{alg:sparse} and let $P'$ be a
resulting configuration. 
Because $P$ satisfies the condition in Theorem~\ref{theorem:suf}, 
if $\gamma(P)$ is a 3D rotation group, 
there is an element $P_i$ with $|P_i| \not\in \{12, 24, 60\}$,
where $\{ P_1, P_2, \ldots, P_m \}$ is 
the $\gamma(P)$-decomposition of $P$. 
Recall that $\{ P_1, P_2, \ldots, P_m \}$ is 
sorted so that $P_1$ is on $I(P)$. 
Algorithm~\ref{alg:sparse} selects the smallest index $s$ 
such that $|P_s| \not\in \{12, 24, 60\}$ and 
shrinks $P_s$ by making each robot $p_i \in P_s$ 
to move to a point $d$ on line segment $\overline{p_i b(P)}$,
where $dist(d, b(P)) = rad(I(P))/2$. 
Thus the robots form new innermost ball in the resulting
configuration $P'$ and $P'$ satisfies $T_1(P')$.
We note that there is no robot on the track of robots in $P_s$
because $s$ is the minimum index and $P'$ contains no multiplicity. 

This preparation phase guarantees that the symmetry breaking in
the second phase occurs on $I(P)$ and 
keeps the center of the smallest enclosing circle of 
the robots unchanged 
when there is some $P_j$ ($j \neq s$). 

\begin{algorithm}[t]
\caption{Preparation algorithm for robot $r_i$}
\label{alg:sparse} 
\begin{tabbing}
xxx \= xxx \= xxx \= xxx \= xxx \= xxx \= xxx \= xxx \= xxx \= xxx
\kill 
{\bf Notation} \\ 
\> $P$: Current configuration observed in $Z_i$. \\ 
\> $\{P_1, P_2, \ldots, P_m\}$: $\gamma(P)$-decomposition of $P$. \\ 
\> $dist(p,q)$: Distance between two points $p$ and $q$ in $Z_i$. \\ 
{\bf Precondition} \\ 
\> $\neg T_1(P)$ \\ 
\\ 
{\bf Algorithm} \\ 
\> Let $P_s$ be the element of the $\gamma(P)$-decomposition of
 $P$ \\ 
\>  that has the smallest index among the elements with
 $|P_s| \in \{12, 24, 60\}$. \\ 
\> {\bf If} $p_i \in P_s$ {\bf then} \\
\> \> Move to the interior of $I(P)$ 
to a point $d$ on line segment $\overline{p_i b(P)}$,  \\ 
\> \> where $dist(d, b(P)) = rad(I(P))/2$. \\ 
\> {\bf Endif} 
\end{tabbing}
\end{algorithm}

\begin{lemma}
\label{lemma:sparse} 
Let $P$ be a configuration that satisfies $\neg T_1(P)$.
 Then the robots execute Algorithm~\ref{alg:sparse} in $P$ 
and suppose that a configuration $P'$ yields as a result.
Then $T_1(P')$ holds.
\end{lemma}

\begin{proof}
Since $P$ satisfies $\neg T_1(P)$,
all robots execute Algorithm~\ref{alg:sparse} in $P$ and   
$\gamma(P) \in \{T,O,I\}$ and $|P_1| \not\in \{12, 24, 60\}$ hold.
There exists at least an element $P_s$ in the $\gamma(P)$-decomposition
 of $P$ that satisfies $|P_s| \not\in \{12, 24, 60\}$.
 Thus $P_s$ is uniquely determined and the robots can agree $P_s$ from
 Theorem~\ref{theorem:decomposition}.

 Because $\gamma(P) \in \{T, O, I\}$,
 the rigid movement of FSYNC robots does not change the center of the 
smallest enclosing ball. 
The movement of robots in $P_s$ keep rotation axes of $\gamma(P)$, 
i.e., $\gamma(P') = \gamma(P)$, 
and for $\gamma(P')$-decomposition $\{ P_1', P_2', \ldots, P_m' \}$ of $P'$, 
the robots in $P_s$ now form $P_1'$, 
i.e., $|P_1'| \not\in \{12,24, 60\}$. 
Thus, $T_1(P')$ holds.
Because there is no robot in the interior of $I(P(0))$ in $P(0)$, 
$P'$ contains no multiplicity. 
 \shortqed
\end{proof}

\subsubsection{Algorithm for symmetry breaking}

The purpose of the symmetry breaking phase is to translate configuration
$P$ that satisfies $T_1(P)$ and $\neg T_2(P)$ to a configuration $P'$
whose rotation group $\gamma(P')$ is a 2D rotation group. 
In configuration $P$ that satisfies $(T_1(P) \wedge \neg T_2(P))$, 
the robots execute Algorithm~\ref{alg:break} and let $P'$ be a resulting
configuration. 
Let $\{P_1, P_2, \ldots, P_m\}$ be the $\gamma(P)$-decomposition of $P$. 
Because $T_2(P)$ does not hold, $\gamma(P)$ is a 3D rotation group.
Because $T_1(P)$ holds, $|P_1| \not\in \{12, 24, 60\}$, 
i.e., $P_1$ is either a regular tetrahedron, 
a regular octahedron, a cube, a regular dodecahedron
or an icosidodecahedron by Table~\ref{table:vt-sets}.
Algorithm~\ref{alg:break} sends the robots in $P_1$ 
to points that are not on any rotation axis of $\gamma(P_1)$. 
Specifically, Algorithm~\ref{alg:break} makes the robots in $P_1$ select an
adjacent face of the polyhedron that $P_1$ forms
and approach the center, but stops the robots $\epsilon$ before the
center.
The exceptional case is when $P_1$ form a regular icosidodecahedron
and each robot of $P_1$ selects an adjacent regular pentagon
face.
We will show that the rotation group of any
resulting configuration $P'$ is a 2D rotation group
and the robots succeed in breaking their symmetry. 

The distance $\epsilon$ is selected so that
after the movement, the robots gather around some 
vertices of the dual polyhedron of $P_1$.
(When $P_1$ forms the icosidodecahedron, the robots gather
around the vertices of a regular icosahedron.) 
For the simplicity of the correctness proof, we use this property. 
Observe that unless we select $\epsilon$ properly,
we do not obtain such polyhedra.
For example, consider the case where $P_1$ forms a unit cube.
In this case, robots in $P_1$ move to their destinations
by moving on the face of a selected face.
If $\epsilon = 1 - (1/\sqrt{2})$,
the set of destination points form a rhombicuboctahedron,
which is, in some sense, in between the cube and its dual
regular octahedron.
(See Figure~\ref{fig:rhombicuboctahedron}.) 
To avoid such configuration, we set
$\epsilon = \ell/100$ where $\ell$ is the length of
the edge of the uniform polyhedron that $P_1$ forms. 
Clearly, the robots in $P_1$ can agree on $\epsilon$ irrespective
of local coordinate systems. 

%%%%%%%%%%%%%%%%%%%%%%%%%%%%%%%%%%%%%%%%%%%%%%%%%%%%%%%%%%%%%%%%%%%%%%%%%%
\begin{figure}[t]
\centering 
\includegraphics[width=3cm]{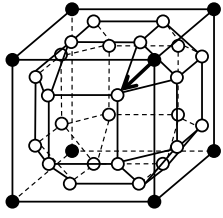}
 \caption{An example of execution of Algorithm~\ref{alg:break}. 
 Starting from a configuration $P$ where the robots form a a unit cube
 (black circles) with $\epsilon = 1 - (1/\sqrt{2})$,
 the set of all destinations (white circles) form a rhombicuboctahedron.} 
\label{fig:rhombicuboctahedron}
\end{figure}
%%%%%%%%%%%%%%%%%%%%%%%%%%%%%%%%%%%%%%%%%%%%%%%%%%%%%%%%%%%%%%%%%%%%%%%%%%

\begin{algorithm}[t]
\caption{Symmetry breaking algorithm for robot $r_i$}
\label{alg:break} 
\begin{tabbing}
xxx \= xxx \= xxx \= xxx \= xxx \= xxx \= xxx \= xxx \= xxx \= xxx
\kill 
{\bf Notation} \\ 
\> $P$: Current configuration observed in $Z_i$.\\ 
\> $\{P_1, P_2, \ldots, P_m\}$: $\gamma(P)$-decomposition of $P$. \\
\> $\ell$: the length of an edge of the polyhedron that $P_1$ forms. \\ 
\> $\epsilon = \ell /100$. \\
{\bf Precondition} \\ 
\> $T_1(P) \wedge \neg T_2(P)$ \\
\\ 
{\bf Algorithm}  \\ 
\> {\bf If} $p_i \in P_1$ {\bf then} \\
\> \> {\bf If} $P_1$ is an icosidodecahedron {\bf then}\\ 
\> \> \> Select an adjacent regular pentagon face of $P_1$. \\ 
\> \> \> Destination $d$ is the point $\epsilon$ before the center 
 of the face on the line \\
\> \> \> from $p_i$ to the center. \\ 
\> \> {\bf Else} \\ 
\> \> \> // $P_1$ is a regular tetrahedron, a regular octahedron, a cube, 
 or a \\
\> \> \> // regular dodecahedron. \\ 
\> \> \> Select an adjacent face of $P_1$. \\ 
\> \> \> Destination $d$ is the point $\epsilon$ before the center 
 of the face on the line \\
\> \> \> from $p_i$ to the center. \\ 
\> \> {\bf Endif} \\ 
\> \> Move to $d$. \\  
\> {\bf Endif} 
\end{tabbing}
\end{algorithm}

\begin{lemma}
\label{lemma:break}
Let $P$ be a configuration that satisfies $T_1(P) \wedge \neg T_2(P)$. 
 Then the robots execute Algorithm~\ref{alg:break} in $P$ 
and suppose that a configuration $P'$ yields as a result.
Then $T_2(P')$ holds.
\end{lemma}

\begin{proof}
Let $\{ P_1, P_2, \ldots, P_m \}$ be the $\gamma(P)$-decomposition of $P$.
Since $T_2(P)$ does not hold, $\gamma(P) \in \{ T, O, I \}$. 
Since $T_1(P)$ holds,
$|P_1| \not\in \{ 12, 24, 60 \}$. 
Thus, as mentioned, $P_1$ is either a regular tetrahedron, 
a regular octahedron, a cube, a regular dodecahedron
or an icosidodecahedron by Table~\ref{table:vt-sets}.
The robots execute Algorithm~\ref{alg:break} in $P$. 

In Algorithm~\ref{alg:break}, only the robots in $P_1$ move.
% \footnote{We emphasize again that a robot $r_i$ can decide 
% whether or not $p_i \in P_1$.}
Each robot $p$ in $P_1$ selects an adjacent
face $F$ of the polyhedron that $P_1$ forms 
% adjacent $B$J$N$+!$(B incident $B$J$N$+!)(B
and moves to $d$ which is at distance $\epsilon$ 
from the center $c(F)$ of $F$ on
line segment $\overline{p c(F)}$, 
with a restriction that $p$ selects a regular pentagon face  
if $P_1$ is an icosidodecahedron. 
Let $k$ be the number of points in $P_1$ incident on a face $F$, 
i.e., $F$ is a regular $k$-gon. 
Then these $k$ robots will form a small regular $k$-gon $U_F$
with the center being $c(F)$ and the distance from the center being $\epsilon$, 
if they all select $F$.
That is, letting $D$ be the set of points consisting
of the candidates for $d$ (for all $p \in P_1$),
$D$ consists of a set of regular $k$-gons $U_F$ congruent each other.

Let $\cal F$ be the set of faces of $P_1$ 
that can be selected by a robot in $P_1$.
(Thus $\cal F$ is a set of regular pentagons 
if $P_1$ is an icosidodecahedron.)
The centers $c(F)$ for $F \in {\cal F}$ 
form a regular polyhedron 
$P_1^d$ that is similar to the dual of $P_1$,
i.e., $P_1^d$ is $c({{\cal F}}) = \{ c(F): F \in {\cal F} \}$ 
(except for the case of icosidodecahedron).
 The convex hull of $D$ is obtained from the dual polyhedron $P_1^d$
 by moving each face of $P_1^d$ away from the center 
 with keeping the center. (See Figure~\ref{fig:expansion}.)
 Then the obtained new polyhedra consists of the moved faces of $P_1^d$
 and new faces formed by the separated vertices and the separated
 edges of $P_1^d$.
Figure~\ref{fig:fulldest} illustrates, for each $P_1$,
the set $D$ by small circles and 
$P_1^d$ as a large polyhedron containing all circles. 
Since the duals of the regular tetrahedron, the regular octahedron, 
the cube, and the regular dodecahedron are 
the regular tetrahedron, the cube, the regular octahedron, and 
the regular icosahedron, respectively, 
we call those convex hulls of $D$ 
$\epsilon$-expanded tetrahedron,
$\epsilon$-expanded cube,
$\epsilon$-expanded octahedron, and
$\epsilon$-expanded icosahedron.\footnote{
 The operation is also known as {\em cantellation}: 
 the convex hull of $D$ is obtained from the dual polyhedron $P_1^d$ 
 by truncating the vertices and beveling the edges. 
}
When $P_1$ is an icosidodecahedron,
although $P_1^d$ is a regular icosahedron,
$D$ is called an $\epsilon$-truncated icosahedron because it is
obtained by just truncating the vertices of $P_1^d$. 

%%%%%%%%%%%%%%%%%%%%%%%%%%%%%%%%%%%%%%%%%%%%%%%%%%%%%%%%%%%%%%%%%%%%%%%%%%
\begin{figure}[t]
\centering 
\includegraphics[width=12cm]{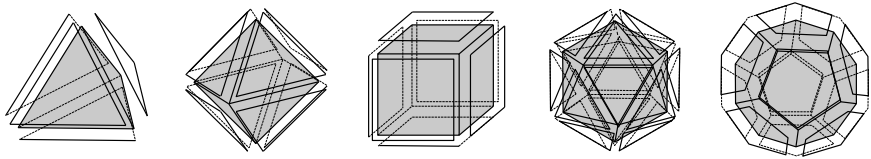}
\caption{Expansion of dual polyhedra.}
\label{fig:expansion}
\end{figure}
%%%%%%%%%%%%%%%%%%%%%%%%%%%%%%%%%%%%%%%%%%%%%%%%%%%%%%%%%%%%%%%%%%%%%%%%%%

Specifically, Figure~\ref{fig:expanded-tetra} illustrates 
an $\epsilon$-expanded tetrahedron,
which corresponds to the convex hull of $D$
when $P_1$ is a regular tetrahedron.
Figure~\ref{fig:expanded-cube} illustrates
an $\epsilon$-expanded cube,
which corresponds to the convex hull of $D$
when $P_1$ is a regular octahedron.
Figure~\ref{fig:expanded-octa} illustrates
an $\epsilon$-expanded octahedron,
which corresponds to the convex hull of $D$
when $P_1$ is a cube.
Figure~\ref{fig:expanded-icosa} illustrates
an $\epsilon$-expanded icosahedron,
which corresponds to the convex hull of $D$
when $P_1$ is a regular dodecahedron.
Finally, Figure~\ref{fig:ticosa} illustrates
an $\epsilon$-truncated icosahedron,
which corresponds to the convex hull of $D$
when $P_1$ is an icosidodecahedron.

%%%%%%%%%%%%%%%%%%%%%%%%%%%%%%%%%%%%%%%%%%%%%%%%%%%%%%%%%%%%%%%%%%%%%%%%%%
\begin{figure}[t]
\centering 
\subfigure[$\epsilon$-expanded tetrahedron]
{\includegraphics[width=3.5cm]{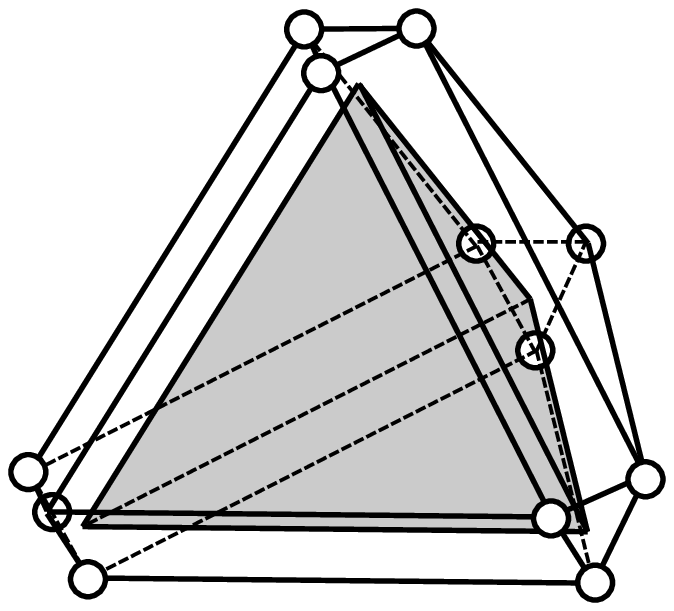}\label{fig:expanded-tetra}}
\hspace{3mm}
\subfigure[$\epsilon$-expanded cube]
{\includegraphics[width=3.5cm]{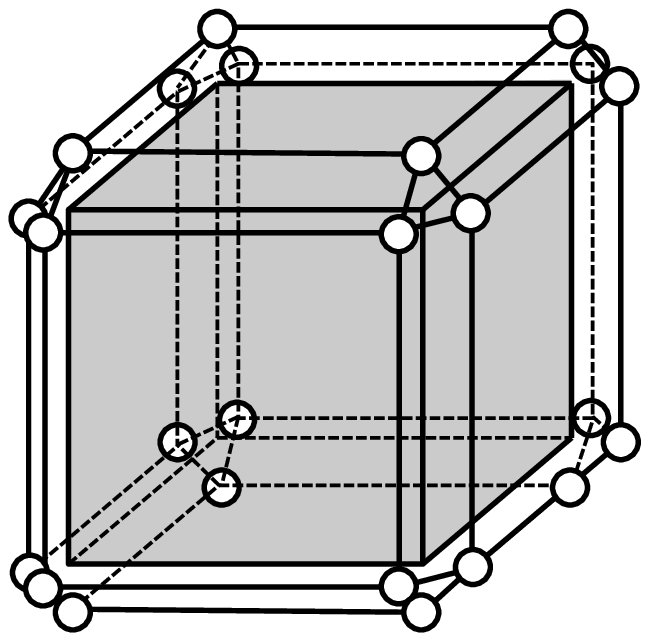}\label{fig:expanded-cube}}
\subfigure[$\epsilon$-expanded octahedron]
{\includegraphics[width=3.5cm]{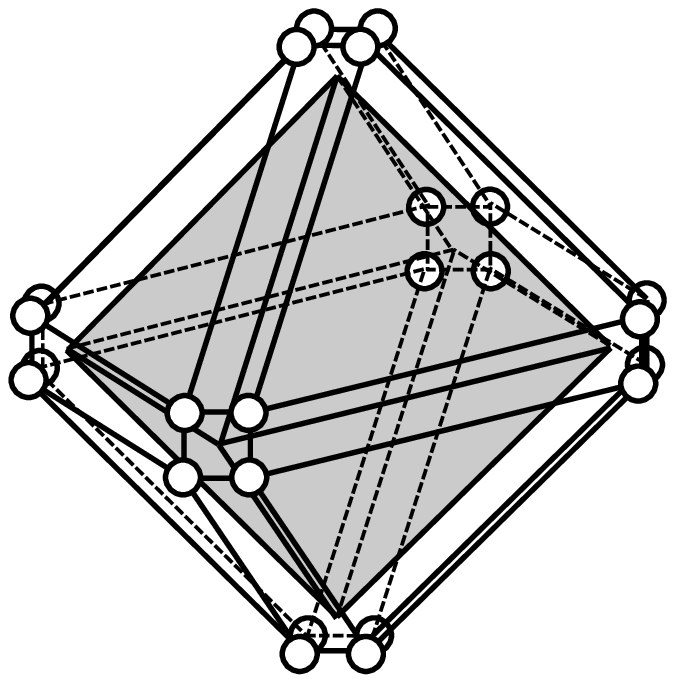}\label{fig:expanded-octa}}
\hspace{3mm}
\\ 
\subfigure[$\epsilon$-expanded icosahedron]
{\includegraphics[width=4cm]{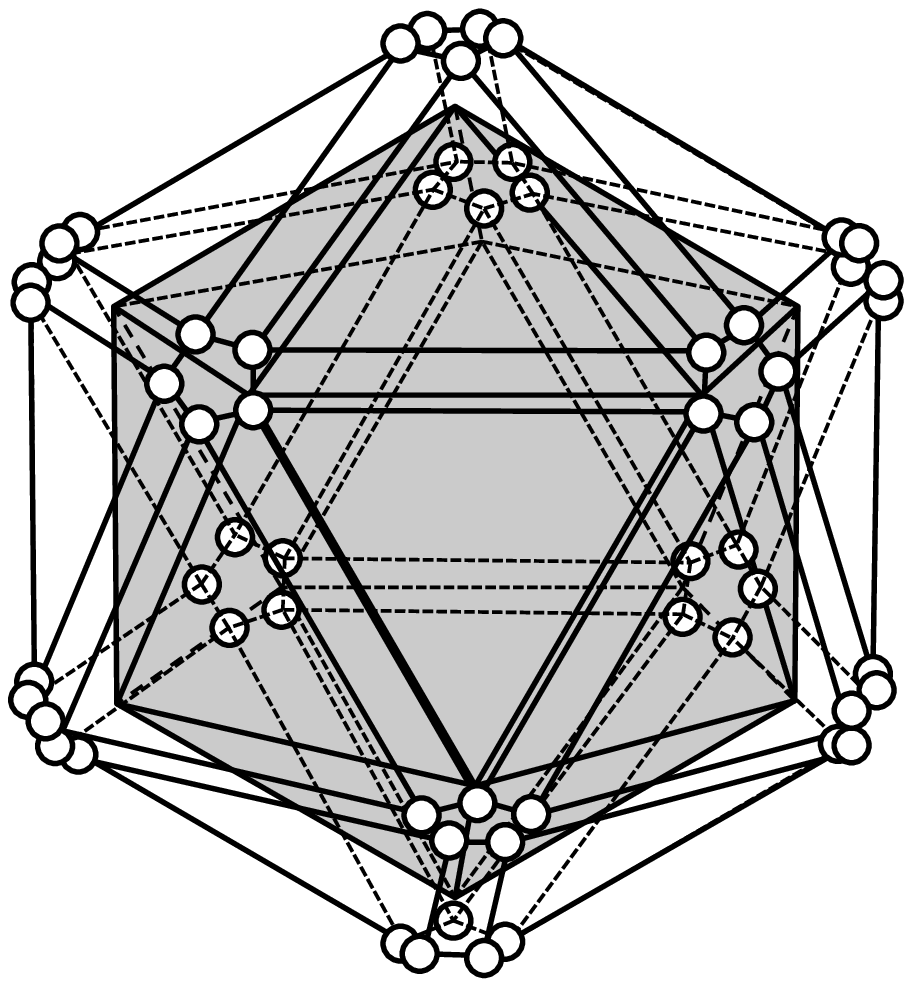}\label{fig:expanded-icosa}}
\hspace{3mm}
\subfigure[$\epsilon$-truncated icosahedron]
{\includegraphics[width=4cm]{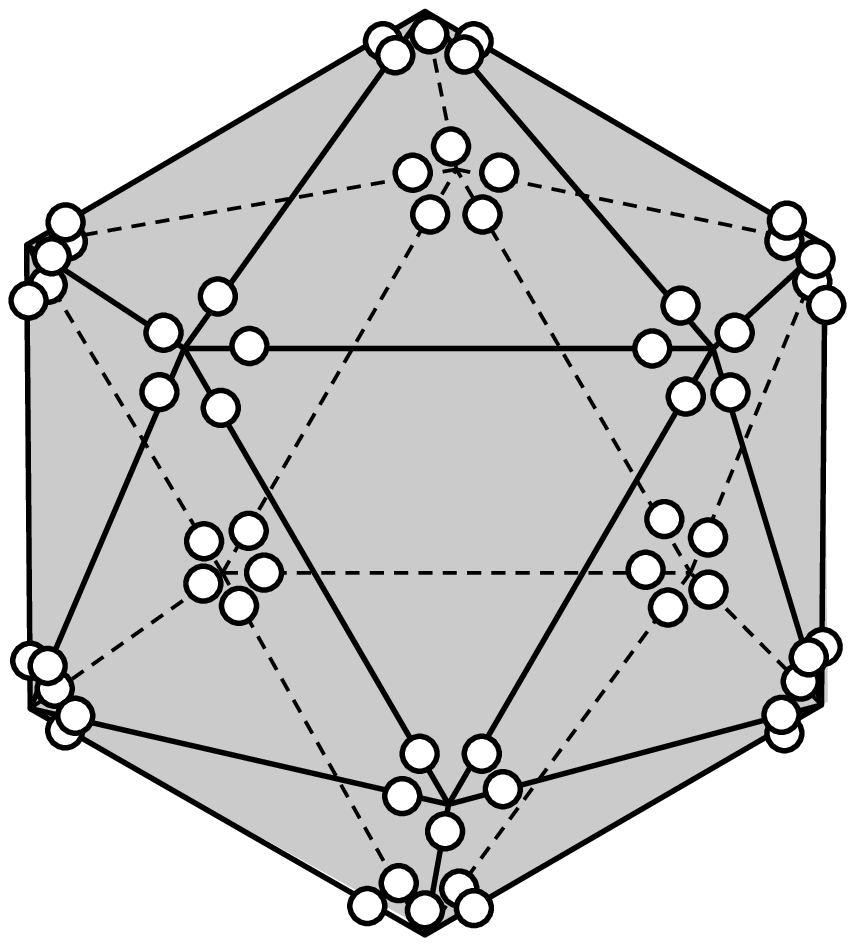}\label{fig:ticosa}}
\caption{Candidate set $D$ corresponding to $P_1$.}
\label{fig:fulldest}
\end{figure}
%%%%%%%%%%%%%%%%%%%%%%%%%%%%%%%%%%%%%%%%%%%%%%%%%%%%%%%%%%%%%%%%%%%%%%%%%%

Let $S \subset D$ be any set selected by robots in $P_1$.
Thus $|S| = |P_1|$ holds.
Then it is sufficient to show that $\gamma(S)$ is a 2D rotation group.
To derive a contradiction,
suppose that there is an $S$ such that $\gamma(S)$ is a 3D rotation group.
We first claim $b(S) = b(D)$ and $B(S) = B(D)$. 
At least two points in $S$ are on the sphere of $B(S)$ 
since $|S| = |P_1| > 2$ 
and $S$ is contained in the sphere of $B(D)$ 
as a subset by definition.
If $B(S) \not= B(D)$,
the intersection of the spheres of $B(S)$ and $B(D)$ is a circle $C$ 
and indeed $S \subseteq C$, 
% since $S$ is contained in the sphere of $B(D)$,
which implies that $\gamma(S)$ is a 2D rotation group.
Thus $B(S) = B(D)$ and $b(S) = b(D)$ hold.
For each of the polyhedra that $S$ can form, 
we now show by contradiction that $\gamma(S)$ is a 2D rotation group  
partly in a brute force manner. 

\noindent{\bf Case A: $P$ is a regular tetrahedron.~} 
The set of destinations $D$ forms an $\epsilon$-expanded 
tetrahedron. See Figure~\ref{fig:expanded-tetra}.
If $\gamma(S)$ is a 3D rotation group, 
$S$ is a regular tetrahedron,
since $|S| = |P_1| = 4$ and the size of any transitive set of
points regarding a 3D rotation group is larger than $4$
except the regular tetrahedron. 

By definition $c({\cal F})$ forms a regular tetrahedron 
and the points of $D$ are $\epsilon$ apart from them. 
Because $b(S) = b(D)$, if $S$ forms a regular tetrahedron,
at most one vertex is selected from $U_F$ for each 
$F \in {\cal F}$. 
Clearly, no such $4$-set forms a regular tetrahedron. 
Thus $\gamma(S)$ is a 2D rotation group for any $4$-set $S \subset D$.
 
\noindent{\bf Case B: $P$ is a regular octahedron.~} 
The set of destinations $D$ forms an $\epsilon$-expanded cube. 
See Figure~\ref{fig:expanded-cube}.
If $\gamma(S)$ is a 3D rotation group, 
because $|S| = 6$, 
$S$ is a regular octahedron,
since otherwise $S$ is a union of a regular tetrahedron and a $2$-set 
and $\gamma(S)$ is a 2D rotation group. 

By definition $c({\cal F})$ forms a regular cube 
and the points of $D$ are $\epsilon$ apart from them. 
Because $b(S) = b(D)$, if $S$ forms a regular octahedron,
at most one vertex is selected from $U_F$ for each
$F \in {\cal F}$. 
Obviously $S$ cannot be a regular octahedron,
 because all vertices of $D$ are around vertices of a cube. 
Thus $\gamma(S)$ is a 2D rotation group for any $6$-set $S \subset D$.

\noindent{\bf Case C: $P$ is a cube.~} 
The set of destinations $D$ forms an $\epsilon$-expanded octahedron.
See Figure~\ref{fig:expanded-octa}.
If $\gamma(S)$ is a 3D rotation group,
because $|S| = 8$, 
$S$ contains either a regular tetrahedron,  
a regular octahedron or a cube as a subset. 

By definition $c({\cal F})$ forms a regular octahedron 
and the points of $D$ are $\epsilon$ apart from them. 
Because $b(S) = b(D)$, if $S$ forms a cube or a regular tetrahedron,
at most one vertex is selected from $U_F$ for each
$F \in {\cal F}$. 
Obviously $S$ is neither a cube nor a regular tetrahedron
because all vertices of $D$ are around vertices of a cube. 
Additionally, like (B), $S$ cannot contain a regular octahedron. 
Thus $\gamma(S)$ is a 2D rotation group for any $8$-set $S \subset D$.
 
\noindent{\bf Case D: $P$ is a regular dodecahedron.~} 
The set of destinations $D$ forms an $\epsilon$-expanded icosahedron. 
See Figure~\ref{fig:expanded-icosa}.
If $\gamma(S)$ is a 3D rotation group,
because $|S| = 20$, 
$S$ contains either a regular tetrahedron,
a regular octahedron, or a cube as a subset.
From Table~\ref{table:vt-sets}, 
$S$ may contain a $12$-set that is transitive regarding $T$,
but in this case there remains $8$ vertices that form one of these
three regular polyhedra. 

By definition $c({\cal F})$ forms a regular icosahedron 
and the points of $D$ are $\epsilon$ apart from them. 
Because $b(S) = b(D)$, if $S$ forms a regular tetrahedron,
a cube, or a regular octahedron, 
at most one vertex is selected from $U_F$ for each
$F \in {\cal F}$. 
Obviously $S$ is not a cube, nor a regular tetrahedron, nor a
regular octahedron, 
because all vertices of $D$ are around vertices of a regular 
 icosahedron.\footnote{Remember that there are five embeddings of
 a cube to a regular dodecahedron. However, we cannot embed a cube
 into its dual regular icosahedron.} 
Thus $\gamma(S)$ is a 2D rotation group for any $20$-set $S \subset D$.

\noindent{\bf Case E: $P$ is an icosidodecahedron.} 
The set of destinations $D$ forms an $\epsilon$-expanded icosahedron.
See Figure\ref{fig:ticosa}.
If $\gamma(S)$ is a 3D rotation group, 
because $|S| = 30$, 
$S$ contains a regular tetrahedron, a regular octahedron, 
a cube, or a regular dodecahedron as a subset.
From Table~\ref{table:vt-sets},
$S$ may contain a transitive set of points
whose size is $12$, $20$, or $24$.
When $S$ contains a transitive $12$-set,
the remaining $18$ points are divided into
(i) $12$-set and $6$-set, (ii) $8$-set, $6$-set, and $4$-set,
or (iii) three $6$-sets. 
When $S$ contains a transitive $20$-set,
its rotation group is $I$ and there is no transitive set of
points with less than $10$ pints regarding $I$.
When $S$ contains a transitive $24$-set,
its rotation group is $O$ and 
the remaining $6$ points form a regular octahedron. 
Thus we check these four regular polyhedra.
 
By definition $c({\cal F})$ forms a regular icosahedron 
and the points of $D$ are $\epsilon$ apart from them. 
Because $b(S) = b(D)$, if $S$ forms a regular tetrahedron,
a cube, a regular octahedron, or a regular dodecahedron, 
at most one vertex is selected from $U_F$ for each
$F \in {\cal F}$. 
Obviously $S$ is not any one of these uniform polyhedra, 
because all vertices of $D$ are around vertices of a regular 
icosahedron.
Thus $\gamma(S)$ is a 2D rotation group for any $30$-set $S \subset D$.

We conclude that $\gamma(S)$ is a 2D rotation group
for any possible $|P_1|$-subset $S$ of $D$. 
If $\gamma(P)$-decomposition of $P$ is a singleton $\{P_1\}$,
the lemma holds. 
Otherwise, the robots forming $S$ is on $I(P')$ in a resulting 
configuration $P'$ and 
$\gamma(P')$ acts on $S$. 
Thus $\gamma(P')$ is a subgroup of $\gamma(S)$ and the lemma holds. 
\shortqed 
\end{proof}

We conclude this section with the following lemma for any resulting configuration of
Algorithm~\ref{alg:break}. 

\begin{lemma}
\label{lemma:breakline}
Let $P$ be a configuration that satisfies $T_1(P) \wedge \neg T_2(P)$. 
Then the robots execute Algorithm~\ref{alg:break} in $P$ 
and suppose that a configuration $P'$ yields as a result. 
Then $P'$ is not on a line. 
\end{lemma}
\begin{proof}
 Let $\{P_1, P_2, \ldots, P_m\}$ be the $\gamma(P)$-decomposition of $P$
 and let $D$ be the set of candidate destinations of the robots in $P_1$.
 From the definition, $D$ is spherical and
 at most two points of $D$ are on a line.
 Because $|P_1| \geq 4$, we have the lemma. 
 \shortqed
\end{proof}

\subsubsection{Algorithm for landing}
\label{subsec:landing}

The purpose of the landing phase is to make the robots agree on a
plane and land on the plane without making any multiplicity. 
In a configuration $P$ that satisfies 
$T_1(P) \wedge T_2(P) \wedge \neg T_3(P)$, 
the robots execute Algorithm~\ref{alg:landing} and let $P'$ be a
resulting configuration. 
Because $T_2(P)$ holds,
$\gamma(P)$ is a 2D rotation group   
and the robots can agree on a plane $F$ perpendicular to 
the single rotation axis or the principal axis of $\gamma(P)$
and containing $b(P)$. 
When $\gamma(P)$ is $C_1$, the robots agree on a plane formed by
$P_1$, say ${p_1}$, its meridian robot, and $b(P)$. 
Clearly, when $\gamma(P)$ is $C_1$, $p_1$ has just one meridian robot
(otherwise, $\gamma(P) \succ C_1$) and these three points 
are not on one line. Thus $F$ is uniquely defined. 
Function \texttt{SelectPlane} in Algorithm~\ref{func:select} 
actually returns this plane $F$ irrespective of local coordinate systems. 

Then the robots carefully determine distinct points on $F$ as their landing
points by Function \texttt{SelectDestination} in
Algorithm~\ref{func:dest}.
The robots on $F$ do not move this landing phase. 
Let $\{ P_1, P_2, \ldots, P_m \}$ be the $\gamma(P)$-decomposition of $P$. 
Each robot computes the expected next positions of all the robots 
so that it avoids collision with other robots. 
The computation of landing points starts with $P_1$. 
For each point $p_i \in P_1$, 
let $f_i$ be the foot of the perpendicular line from $p_i$ to $F$
and $p_i$ adopts it as its landing point.
We denote the set of these landing points by
$F_1 = \{f_i :  p_i \in P_1\}$. 
Unfortunately, at most two robots in $P_1$ have the same landing point.
To resolve this collision, 
we make use of the following trick:
Let $f_i = f_j$ for two robots $p_i, p_j \in P_1$.
Then $p_i$ and $p_j$ are in the opposite side regarding $F$. 
Now $p_i$ ($p_j$, respectively) assumes that it rotates its local coordinate system
so that the direction of negative $z$-axis coincides with
the direction of $f_i$ ($f_j$, respectively). 
Then their clockwise directions, i.e.,
the rotation from positive $x$-axis to
positive $y$-axis on $F$ are opposite 
because their local coordinate systems are right-handed.
(See Figure~\ref{fig:achiral2}.) 
Function \texttt{SelectDestination} changes the landing points of
$p_i$ and $p_j$ by using this property. 
Let $C(f_i)$ be the circle centered at $f_i$ and contains 
no point in $(F \cap P) \cup (F_1 \setminus \{f_i\})$ in its interior 
and at least one point in 
$(F \cap P) \cup (F_1 \setminus \{f_i\})$ on its circumference. 
Then, let $C'(f_i)$ be the circle centered at $f_i$ with 
radius $rad(C(f_i))/4$. 
Clearly, such quarter circles for $p_i \in P_1$ have no intersection 
unless they have the common foot. 
Then \texttt{SelectDestination} outputs 
distinct landing points for $p_i$ and $p_j$ 
from $C'(f_i) = C'(f_j)$. 
If $f_i = f_j \neq b(P)$, \texttt{SelectDestination} 
outputs their destinations 
by rotating the intersection of $C'(f_i) = C'(f_j)$ and the line segment 
$\overline{f_i b(P)} = \overline{f_j b(P)}$ 
clockwise by $\pi/2$ with the center being $f_i = f_j$.
(See Figure~\ref{fig:distinct1}.)
Thus \texttt{SelectDestination} at $p_i$ and $p_j$ output
different landing pints. 
The obtained landing points are marked so that 
they will not be selected in the succeeding computation 
for $P_2, P_3, \ldots, P_m$. 

If $f_i = f_j = b(P)$, we cannot use the above technique. 
For $p_i$, \texttt{SelectDestination} 
selects a vertex $q_i'$ of a $|\gamma(P)|$-gon $Q(P)$ on $F$, 
that is defied by rotation axes of $\gamma(P)$ as we will define later 
and consider the intersection of $C'(f_i)$ and line segment 
$\overline{b(P) q_i'}$ as $q_i$. 
Then, \texttt{SelectDestination} 
rotates $q_i$ clockwise by $(2\pi)/(4|\gamma(P)|)$ with the center being
$b(P)$. 
A new landing point for $q_j$ is obtained in the same way, 
but even when the same vertex of $Q(P)$ is selected, 
the clockwise rotations guarantee that 
the landing points of $p_i$ and $p_j$ are distinct. 
(See Figure~\ref{fig:distinct2}.) 

We formally define $Q(P)$ as follows: 
If $\gamma(P)$ is dihedral, 
the vertices of $Q(P)$ are the intersections of 
the $2$-fold axes and the large circle formed by $B(P)$ and
$F$.\footnote{Because $F$ contains $b(P)$, the intersection of 
$B(P)$ and $F$ is a large circle of $B(P)$.} 
Otherwise, $\gamma(P)$ is cyclic and let $P_{\ell}$ be the 
subset of $\gamma(P)$-decomposition of $P$ with the largest index 
such that $P_{\ell}$ form a regular $\gamma(P)$-gon. 
Such $P_i$ exists from the definition and actually, 
$P_{\ell}$ form a plane parallel to $F$. 
Then $Q(P)$ is the $\gamma(P)$-gon obtained by 
projecting $P_{\ell}$ on $F$ and expanding it with keeping the center 
so that it touches the large circle formed by $B(P)$ and $F$. 
We note that when $\gamma(P)$ is cyclic, 
no two robots have the same foot for each subset, 
but \texttt{SelectDestination} uses $Q(P)$ for robot $p_i \in P_k$  
to avoid a point that is already marked as a landing point of 
some robot in $P_1, P_2, \ldots, P_{k-1}$. 

The above collision resolution procedure has a small flaw: 
It does not work correctly 
when $\gamma(P) = C_1$ and $f_i =b(P) \in F \cap P$. 
In this case, \texttt{SelectDestination} computes a landing point of 
a robot at a time with avoiding the expected landing points and 
it selects an arbitrary point on $C'(f_i)$ as the landing point of
$p_i$. 

Finally, to avoid further collisions 
all points on $C'(f_i)$ are considered to be ``expected landing points'' 
when $f_i = b(P)$,
because \texttt{SelectDestination} invoked at $r_k$ and  
\texttt{SelectDestination} invoked at $r_{k'}$ ($r_k, r_{k'} \in R$) 
may not output the same landing point for $p_i$. 
Then \texttt{SelectDestination} proceeds $P_2$. 
The landing points of $P_2$ avoid all (expected) landing points of $P_1$  
and collision among $P_2$ in the same way.
During the computation of \texttt{SelectDestination},
if the landing point of $r_i$ is not $f_i$, we say
$r_i$'s landing point is {\em perturbed}. 
By computing expected landing points of all robots, 
\texttt{SelectDestination} invoked at 
each robot outputs its landing point on $F$. 
Finally, robots move to their landing points directly and 
in any resulting configuration $P'$,
$T_3(P')$ holds.

%%%%%%%%%%%%%%%%%%%%%%%%%%%%%%%%%%%%%%%%%%%%%%%%%%%%%%%%%%%%%%%%%%%%%%%%%%
\begin{figure}[t]
\centering 
\includegraphics[width=5cm]{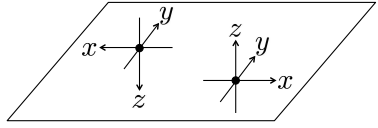}
\caption{When two right-handed robots have opposite $z$ axes, 
they do not agree on the clockwise direction. }
\label{fig:achiral2}
\end{figure}
%%%%%%%%%%%%%%%%%%%%%%%%%%%%%%%%%%%%%%%%%%%%%%%%%%%%%%%%%%%%%%%%%%%%%%%%%%

%%%%%%%%%%%%%%%%%%%%%%%%%%%%%%%%%%%%%%%%%%%%%%%%%%%%%%%%%%%%%%%%%%%%%%%%%%
\begin{figure}[t]
\centering 
\subfigure[]
{\includegraphics[width=5cm]{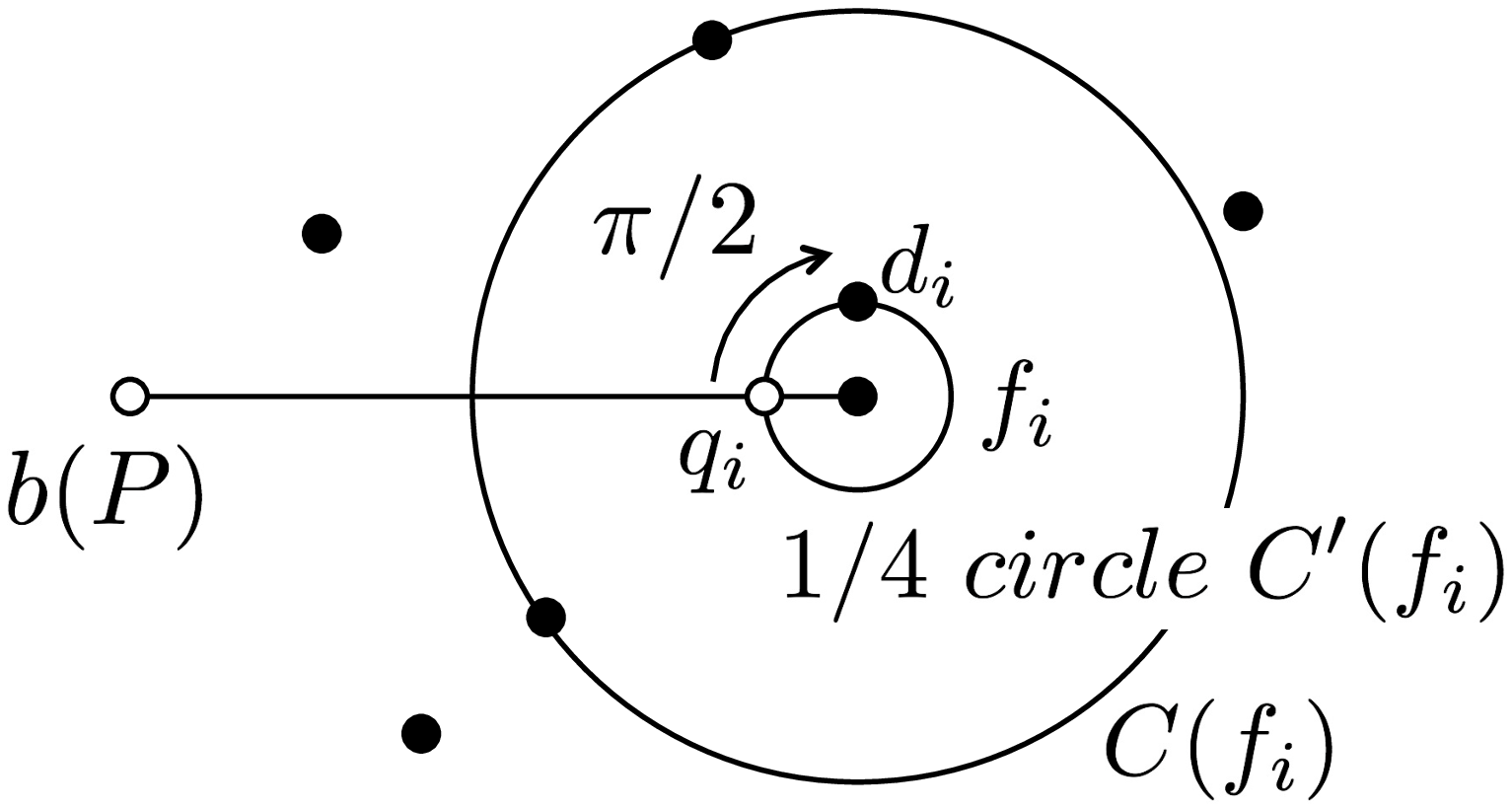}\label{fig:distinct1}}
\hspace{3mm}
\subfigure[]
{\includegraphics[width=5cm]{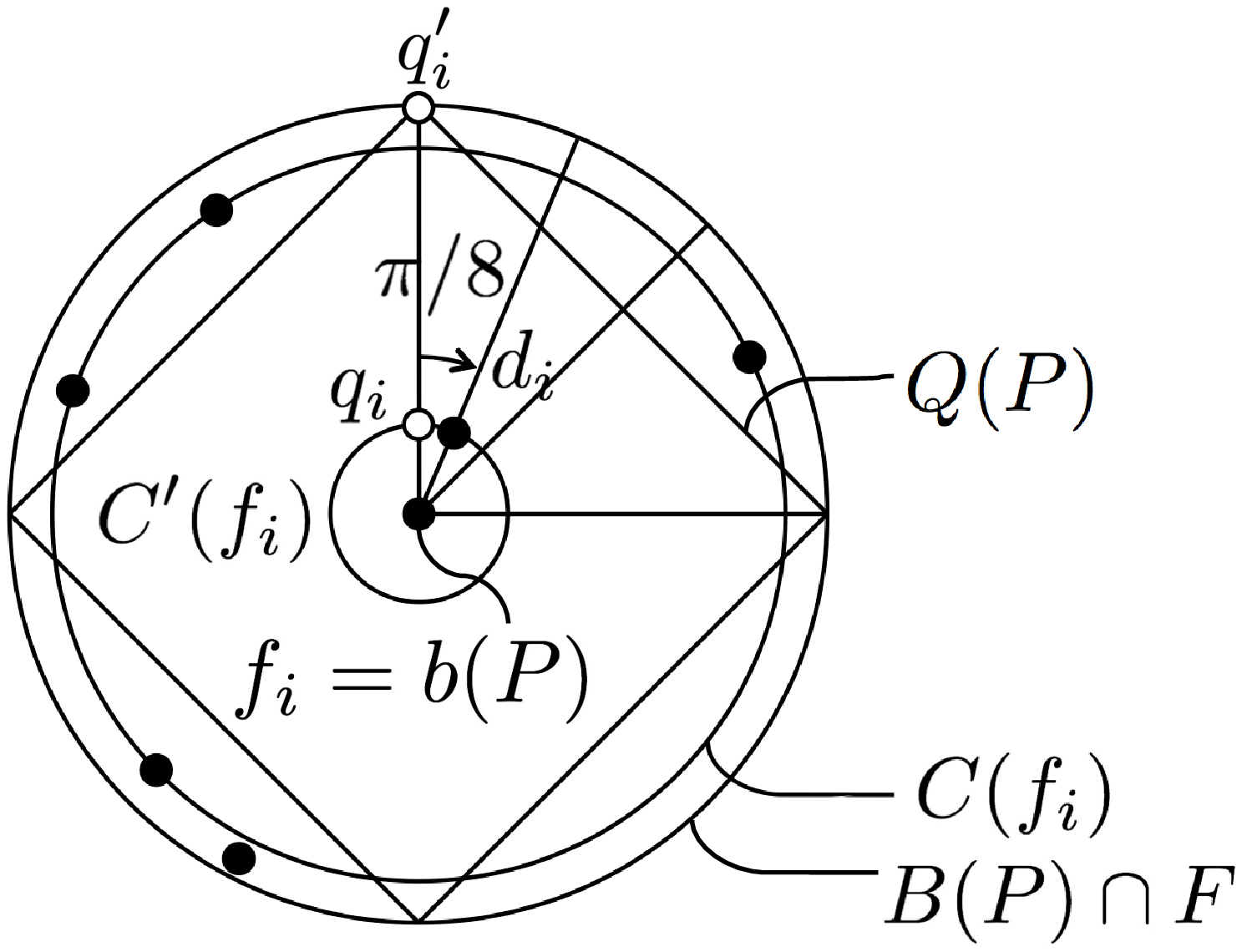}\label{fig:distinct2}}
\caption{Trick to avoid a collision of landing points on plane $F$. 
(a) When the original landing point (which may collide with another robot)
is $f_i \neq b(P)$, $r_i$ chooses a point $d_i$,
which is obtained from $q_i$ by rotating it clockwise
by angle $\pi/2$ with $f_i$ being the rotation center. 
(b) When the original landing point 
is $f_i = b(P)$, $r_i$ chooses a point $d_i$,
which is obtained from $q_i$ by rotating it clockwise
with $f_i$ being the rotation center by using $Q(P)$. 
The figure shows when $Q(P)$ is a square. 
}
\label{fig:perturb}
\end{figure}
%%%%%%%%%%%%%%%%%%%%%%%%%%%%%%%%%%%%%%%%%%%%%%%%%%%%%%%%%%%%%%%%%%%%%%%%%%

\begin{algorithm}[t]
\caption{Landing algorithm for robot $r_i$}
\label{alg:landing} 
\begin{tabbing}
xxx \= xxx \= xxx \= xxx \= xxx \= xxx \= xxx \= xxx \= xxx \= xxx
\kill 

{\bf Notation} \\ 
\> $P$: Current configuration observed in $Z_i$.\\ 
\> $\{P_1, P_2, \ldots, P_m\}$: $\gamma(P)$-decomposition of $P$. \\ 
{\bf Precondition} \\ 
\> $T_1(P) \wedge T_2(P) \wedge \neg T_3(P)$ \\
\\ 
{\bf Algorithm} \\ 
\> Consider that the local coordinate system is turned
 so that negative $z$-axis \\
\> points to $b(P)$. \\ 
\> $F = \texttt{SelectPlane}(P)$. \\ 
\> $d = \texttt{SelectDestination}(P, F)$.\\ 
\> Move to $d$. 
\end{tabbing}
\end{algorithm}

\begin{algorithm}[h]
\caption{Function $\texttt{SelectPlane}(P)$}
\label{func:select} 
\begin{tabbing}
xxx \= xxx \= xxx \= xxx \= xxx \= xxx \= xxx \= xxx \= xxx \= xxx
\kill 
{\bf Notation} \\ 
\> $P$: Current configuration observed in $Z_i$.\\ 
\> $\{P_1, P_2, \ldots, P_m\}$: $\gamma(P)$-decomposition of $P$. \\ 
\\
{\bf Function} \\ 
$\texttt{SelectPlane}(P)$ \\ 
\> {\bf If} $\gamma(P) = C_1$ {\bf then} \\ 
\> \> Let $P_1 = \{r^*\}$. \\ 
\> \> Let $F$ be the plane containing $r^*$, $r^*$'s meridian 
robot and $b(P)$. \\ 
\> {\bf Endif} \\ 
\> {\bf If} $\gamma(P) = C_k$ ($k \geq 2$) {\bf then} \\ 
\> \> Let $F$ be the plane perpendicular to the single rotation axis 
 and \\
\> \> containing $b(P)$. \\ 
\> {\bf Endif} \\ 
\> {\bf If} $\gamma(P) = D_k$ ($k >2$) {\bf then} \\ 
\> \> Let $F$ be the plane perpendicular to the principal axis and \\
\> \> containing $b(P)$. \\ 
\> {\bf Endif} \\ 
\> {\bf Return} $F$.  
\end{tabbing}
\end{algorithm}

\begin{algorithm}[]
\caption{Function $\texttt{SelectDestination}(P, F)$}
\label{func:dest} 
\begin{tabbing}
xxx \= xxx \= xxx \= xxx \= xxx \= xxx \= xxx \= xxx \= xxx \= xxx
\kill 
{\bf Notation} \\ 
\> $P$: Current configuration observed in $Z_i$.\\ 
\> $\{P_1, P_2, \ldots, P_m\}$: $\gamma(P)$-decomposition of $P$. \\ 
\> $Q(P)$: the regular $\gamma(P)$-gon on $F$ fixed by $P$. \\ 
\\ 
{\bf Function} \\ 
$\texttt{SelectDestination}(P, F)$ \\ 
\> Set landing points $D=P \cap F$. \\ 
\> {\bf For} $k=1$ {\bf to} $m$ \\  
\> \> For each $p_i \in P_k$, let $f_i$ be the foot of the perpendicular line. \\
\> \> $F_k = \{f_i : p_i \in P_k\}$. \\ 
\> \> {\bf If} there exists $p_j \in P_k$ that satisfies 
$(f_j \in D)$ or $(\exists p_{j'} \in P_k: f_j = f_{j'})$ {\bf then} \\ 
\> \> \> For each $p_i \in P_k$, let $C(f_i)$ be the circle centered
 at $f_i$, \\ 
\> \> \> containing no point in $D \cup F_k \setminus \{f_i\}$ in its interior \\ 
\> \> \> and at least one point in 
$D \cup F_k \setminus \{f_i\}$ on its circumference. \\ 
\> \> \> (If $C(f_i)$ is not fixed (i.e., $D = \{b(P)\}$, let
 $C(f_i)$ be the circle \\ 
\> \> \>  centered at $f_i$ with radius $rad(B(P))$.)) \\ 
\> \> \> Let $r = min_{p_i \in P_k} rad(C(f))$. \\ 
\> \> \> For each $p_i \in P_k$, let $C'(f_j)$ be the circle centered at $f_i$ 
with radius $r/4$. \\ 
\> \> \> {\bf If} $f_j \neq b(P)$ {\bf then} \\ 
\> \> \> \> Let $q_j$ be the intersection of $C'(f_j)$ 
and the line segment $\overline{f_j b(P)}$. \\ 
\> \> \> \> Assume $r_j$'s negative $z$ axis points to $F$. \\ 
\> \> \> \> Let $d_j$ be the point on $C'(f_j)$ obtained by turning $q_j$ \\
\> \> \> \> around $f_j$ by $\pi/2$ clockwise. \\ 
\> \> \> \> $D'_j = \{d_j\}$. \\ 
\> \> \> {\bf Else} // $f_j = b(P)$. \\ 
\> \> \> \> {\bf If} $\gamma(P)=C_1$ {\bf then} // $f_j \in D$. \\ 
\> \> \> \> \> Select an arbitrary point on $C'(f_j)$ as $d_j$.\\ 
\> \> \> \> \> $D_j' = C'(f_j)$. \\ 
\> \> \> \> {\bf Else} \\ 
\> \> \> \> \> Select an arbitrary vertex $q_j'$ from $Q(P)$. \\ 
\> \> \> \> \> Let $q_j$ be the intersection of $C'(f_j)$ 
and the line segment $\overline{q_j' b(P)}$. \\ 
\> \> \> \> \> Let $d_j$ be the point on $C'(f_j)$ obtained by turning 
$q_j$ around $f_j$ \\ 
\> \> \> \> \> by $2\pi/4|\gamma(P)|$ clockwise. \\ 
\> \> \> \> \> $D_j' = C'(f_j)$. \\ 
\> \> \> \> {\bf Endif} \\ 
\> \> \> {\bf Endif} \\ 
\> \> {\bf Else} $d_j = f_j$ and $D_j' = \{d_j\}$. \\ 
\> \> {\bf Endif} \\ 
\> \> {\bf If} $r_j = r_i$ {\bf then} $d = d_j$ {\bf Endif}\\ 
\> \> $D = D \cup \bigcup_{ p_j \in P_k} D'_j$. \\ 
\> {\bf Endfor} \\  
\> {\bf Return} $d$. 
\end{tabbing}
\end{algorithm}

Lemma~\ref{lemma:landing} shows that the robots
occupy distinct positions on $F$ in any resulting configuration $P'$ and
Lemma~\ref{lemma:noline} shows that the robots do not form a line
in $P'$. 
For the simplicity of the proof for Lemma~\ref{lemma:noline},
we incorporate the following small improvement to
\texttt{SelectDestination}: 
If $\gamma(P)$ is $C_k$ ($k \geq 2$) ($D_{\ell}$ ($\ell \geq 2$),
respectively),
each element $P_i$ of the $\gamma(P)$-decomposition of $P$,
$\{P_1, P_2, \ldots, P_m\}$ forms one regular $k$-gon
(or two regular $\ell$-gons, respectively) on $F$.
Consider the case where $\gamma(P) = C_k$ and
we have $|P_1| = |P_2| = 1$ and $|P_3| = |P_4| = k$, i.e.,
$P$ consists of two pyramids. 
Then, the destinations of $P_2$ is perturbed because $b(P)$ is the
destinations of $P_1$. 
Suppose that this perturbed landing point of $P_2$
is on the foot of the perpendicular
line from some robot of $P_3$. (See Figure~\ref{fig:coll-1}.) 
To keep the regular $k$-gon of $P_3$,
we make \texttt{SelectDestination} perturb the destinations of
all robots of $P_3$. 
\texttt{SelectDestination} makes these $k$-robots,
say $r_1, r_2, \ldots, r_k$ agree on the radius of $C(f_i)$
and choose a new destination from $C(f_i)$.
(See Figure~\ref{fig:coll-2}.) 
We also have such a situation when $\gamma(P) = D_{\ell}$ 
($\ell \geq 2$). In this case, we consider the elements of
the $\gamma(P)$-decomposition of $P$ with size $|D_k| = 2k$,
thus each of the elements consists of two regular $k$-gons one is 
``above'' $F$ and the other is ``under'' $F$. 
Then \texttt{SelectDestination} keeps these two regular $k$-gons
in the same way when at least one of the $2k$ robots have a collision.

%%%%%%%%%%%%%%%%%%%%%%%%%%%%%%%%%%%%%%%%%%%%%%%%%%%%%%%%%%%%%%%%%%%%%%%%%%
\begin{figure}[t]
\centering 
\subfigure[]
{\includegraphics[width=4.5cm]{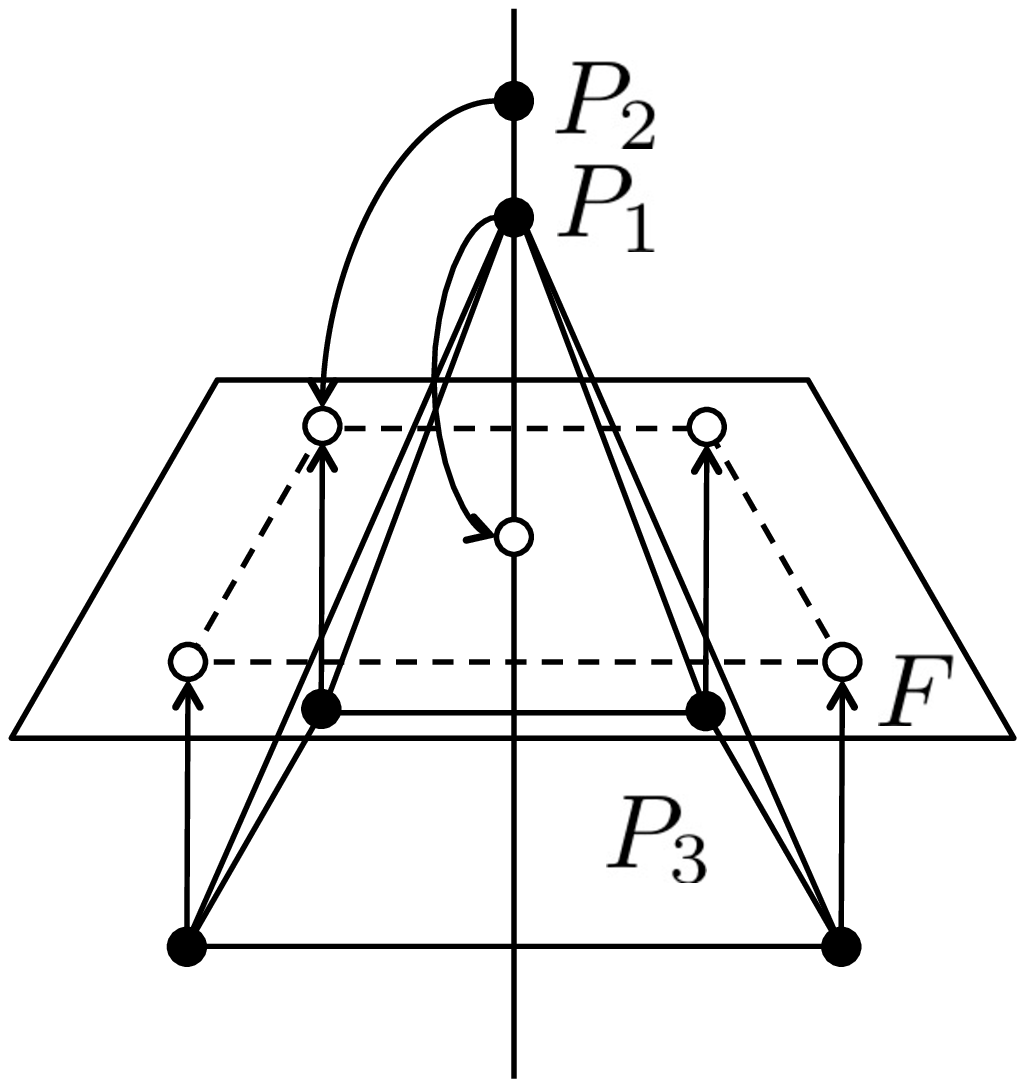}\label{fig:coll-1}}
\hspace{3mm}
\subfigure[]
{\includegraphics[width=4.5cm]{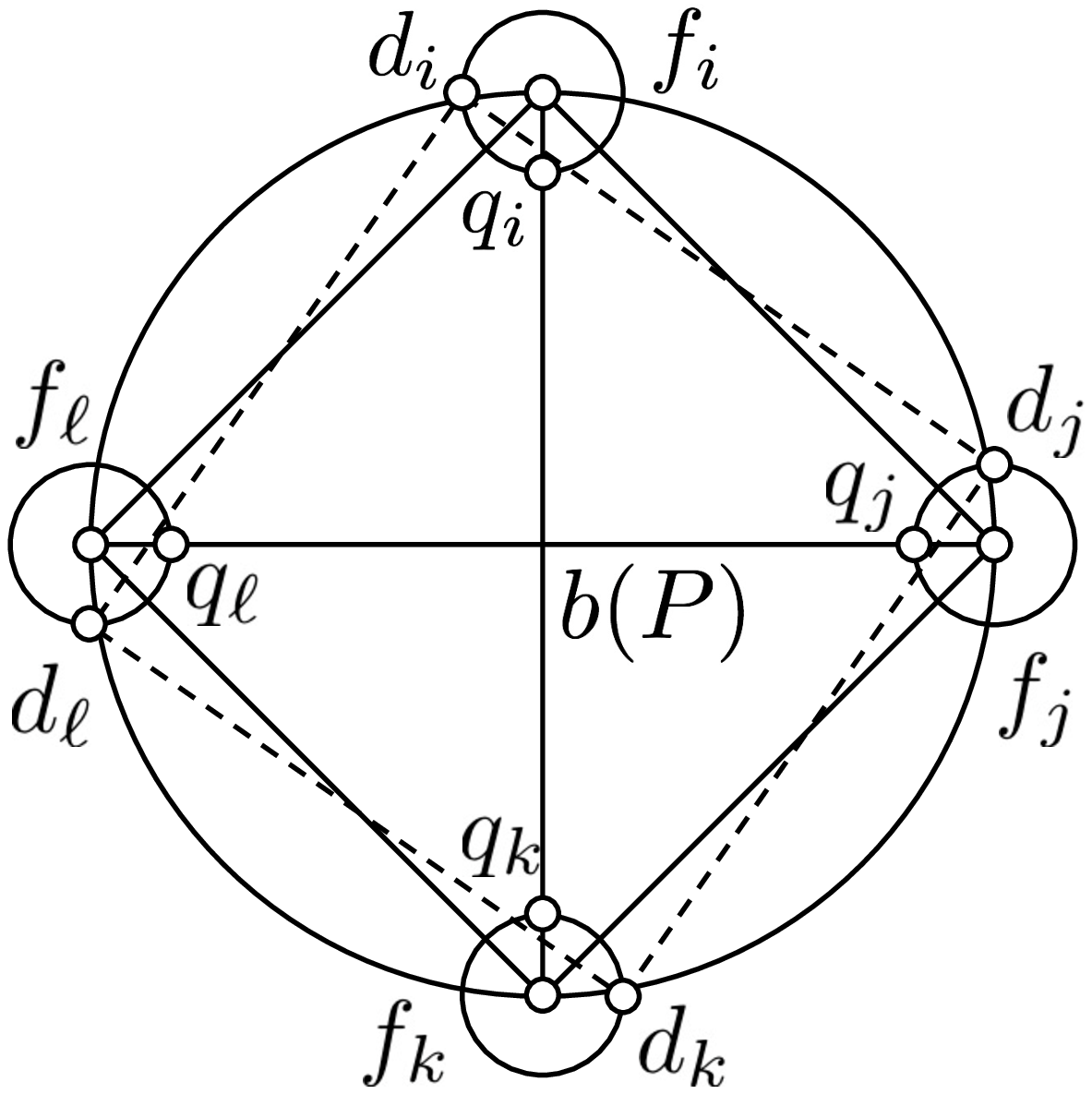}\label{fig:coll-2}}
 \caption{Avoiding a collision with keeping the regular polygon
 formed by an element of the $\gamma(P)$-decomposition of $P$. 
}
\label{fig:coll}
\end{figure}
%%%%%%%%%%%%%%%%%%%%%%%%%%%%%%%%%%%%%%%%%%%%%%%%%%%%%%%%%%%%%%%%%%%%%%%%%%

\begin{lemma}
\label{lemma:landing}
 Let $P$ be a configuration that satisfies
 $T_1(P) \wedge T_2(P) \wedge \neg T_3(P)$. 
Then the robots execute Algorithm~\ref{alg:landing} in $P$ and 
suppose that a configuration $P'$ yields as a result.
Then $T_3(P')$ holds.
\end{lemma}

\begin{proof}
Let $\{ P_1, P_2, \ldots, P_m \}$ 
be the $\gamma(P)$-decomposition of $P$. 
Since $T_2(P)$ holds, $\gamma(P)$ is a 2D rotation group  
and the robots can agree on a common plane $F$
as we have discussed several times.
Indeed, $\texttt{SelectPlane}(P)$ returns $F$,
as one can easily observe.
More clearly, let $\texttt{SelectPlane}(Z_i(P)) = F_i$ for any $r_i \in R$.
Then there is a common plane $F$ such that 
 $Z_i(F) = F_i$ for $i = 1, 2, \ldots , n$ because
 $F_i$ does not depend on the local coordinate system. 

What remains is to show that $\texttt{SelectDestination}(Z_i(P), Z_i(F))$
outputs distinct positions for the robots. 
Let $p_j$ be the position of $r_j \in R$. 
The robots can agree on the foot $f_j$ on $F$ for $r_j$. 
If there are robots with the same foot, 
$\texttt{SelectDestination}$ resolves the collision. 
We consider the execution of $\texttt{SelectPlane}(Z_i(P))$ at $r_i$, 
and show by induction that no robot other than $r_i$ 
selects the destination of $d_i$ computed 
by $\texttt{SelectDestination}(Z_i(P), Z_i(F))$ at $R_i$
as its destination. 

First, $\texttt{SelectDestination}(Z_i(P), Z_i(F))$ at $r_i$ initializes 
the set of landing points $D = P \cap F$.  
As for $P_1$, if $f_j = f_{j'}$ for $p_j, p_{j'} \in P_1$, 
or $f_j \in D$, 
$\texttt{SelectDestination}$ 
computes distinct destinations $d_j$ (and $d_{j'}$) from 
$C'(f_j)(=C'(f_{j'}))$ 
based on the trick shown in Figure~\ref{fig:distinct1} and 
\ref{fig:distinct2} and it appends these 
(expected) destinations to $D$. 

Actually, for $r_i, r_{i'} \in R$, 
the destination of $p_j$ output 
by $\texttt{SelectDestination}(Z_i(P), Z_i(F))$ at $r_i$ 
and that by  $\texttt{SelectDestination}(Z_{i'}(P), Z_{i'}(F))$ at $r_{i'}$ 
are not always identical; for example, if $f_j = f_{j'}=b(P)$, 
the destinations at $r_j$ and $r_{j'}$ may be different. 
However, in such a case, $r_i$ (and $r_i'$ also) appends $C'(f_j)$ 
to $D$ as expected destinations. 
Hence, if $r_i \in P_1$, its destination $d_i$ computed at $r_i$ 
is always in $D$ at each robot $r_{i'} \in R$. 

After the computation of $P_k$, 
$\texttt{SelectDestination}(Z_i(P), Z_i(F))$ computes the destinations of 
$P_{k+1}$ ($k < m$). In this phase, $\texttt{SelectDestination}(Z_i(P), Z_i(F))$ 
resolves collisions among the foot of $r_j \in P_{k+1}$ with 
avoiding the points in $D$, and appends new (expected) destinations to $D$. 
Hence, if $r_i \in P_1 \cup \cdots \cup P_k$, 
$d_i \in D$ is not selected as a destination of some robot in 
$P_{k+1}$. 
\shortqed
\end{proof}

We conclude Theorem~\ref{theorem:suf} by 
Lemmas~\ref{lemma:sparse}, Lemma~\ref{lemma:break}, and 
Lemma~\ref{lemma:landing}. 
From Theorem~\ref{theorem:suf} and
Theorem~\ref{theorem:suf-power}, 
we show that 
the robots can form a plane if an initial configuration satisfies 
the condition of Theorem~\ref{theorem:main} 
irrespective of the availability of memory. 

We finally show that the robots do not form a line
in any terminal configuration of the proposed algorithm as long as 
they form neither a plane nor a line in an initial configuration. 
This property is not required by the plane formation problem,
however it is useful when we combine other algorithms for robots on
2D-space 
to the proposed plane formation algorithm for more
complex tasks. 

\begin{lemma}
 \label{lemma:noline}
 Let $P(0)$ be an initial configuration and
 $P(t)$ be the terminal configuration of the proposed algorithm.
 Then the robots are not on a line in $P(t)$.
 \end{lemma}
 \begin{proof}
  Consider an arbitrary
  execution of the proposed algorithm $P(0), P(1), \cdots$.  
  Assume that the robots are on a common line in the terminal
  configuration 
  $P(t)$ ($t > 0$) that appears in the execution for the first time. 

  Because $P(t)$ is the first terminal configuration,
  the robots are not on a plane in $P(t-1)$ and they execute
  Algorithm~\ref{alg:landing} in $P(t-1)$
  by Lemma~\ref{lemma:breakline}. 
  We consider the rotation group of $P(t-1)$. 
  Let $F$ be the plane output by \texttt{SelectPlane} in $P(t-1)$. 
  When $\gamma(P(t-1))$ is $C_k$ for $k \geq 3$,
  there exists at least one element of the
  $\gamma(P(t-1))$-decomposition of $P(t-1)$ that forms a 
  regular $k$-gon. 
  From the definition of \texttt{SelectDestination},
  the destinations of these $k$ robots form a regular $k$-gon on $F$
  and the robots are not on a common line in $P(t)$. 
  We have the same case when $\gamma(P(t-1))$ is $D_{\ell}$ for 
  $\ell \geq 3$ because there exists at least one element of
  $\gamma(P(t-1))$-decomposition of $P(t-1)$ that form a
  regular $2\ell$-gon on $F$ or two regular $\ell$-gons each of
  which is parallel to $F$ (but not on $F$). 

  Hence the remaining cases are when $\gamma(P(t-1))$ is $C_1$,
  $C_2$, or $D_2$. 
  
  \noindent{\bf Case A: $\gamma(P(t-1)) = C_1$.}
  Assume that in $P(t)$, the robots are on a common line, say
  $\ell$ on $F$. 
  If no robot has its destination perturbed,
  then the positions of the robots are on the plane
  that is perpendicular to $F$ and whose intersection
  with $F$ is $\ell$. Hence the plane formation has
  completed in $P(t-1)$, which is a contradiction. 

  Otherwise, at least one robot has its destination perturbed.
  Let $r_i$ be this robot. We denote the foot of the perpendicular line
  to $F$ from $p_i(t-1)$ by $f_i$.
  Hence there exists another robot whose destination is $f_i$.
  Let $d_i$ be the perturbed destination of $r_i$.
  The destinations of the robots are on the line $\ell'$
  that contains $f_i$ and $d_i$.
  This means that all robots are on the plane 
  that is perpendicular to $F$ and whose intersection
  with $F$ is $\ell'$ in $P(t-1)$.
  Assume otherwise that there exists a robot $r_j$ that is not on
  this plane in $P(t-1)$.
  Hence the foot $f_j$ of the perpendicular line to $F$ from $r_j$ 
  is not on $\ell'$ and there exists at least one robot
  (not necessarily $r_j$) whose destination is $f_j$,
  which is a contradiction. 
  Hence the plane formation has
  completed in $P(t-1)$, which is a contradiction. 
  
  \noindent{\bf Case B: $\gamma(P(t-1)) = C_2$.}
  Hence the $\gamma(P(t-1))$ decomposition of $P(t-1)$ contains 
  at least two $2$-sets, otherwise all robots are on the
  plane containing the principal axis and the line formed by the single
  $2$-set. 
  Because the robots are not on a plane in $P(t-1)$,
  the $\gamma(P(t-1))$-decomposition of $P(t-1)$ contains at least
  two $2$-sets, say $P_i$ and $P_j$, that are not on a common plane.
  Thus their destinations on $F$ are distinct and
  they form a rhombus and the robots are not on one line in $P(t)$. 
  
  \noindent{\bf Case C: $\gamma(P(t-1)) = D_2$.}
  When a robot is on $F$ in $P(t-1)$,
  then the robot does not move during the transition from $P(t-1)$ to
  $P(t)$.
  We consider the robots that are not on $F$.
  Specifically, we consider each element of the
  $\gamma(P(t-1))$-decomposition of $P(t-1)$ that is not on $F$.
  We have the following three cases.

  \noindent{\bf Case C(i): There exists at least one element that
  forms a sphenoid.}
  In this case, the destinations of the four robots forming the sphenoid
  form a rectangle on $F$, and the robots are not on a common line in
  $P(t)$. 

  \noindent{\bf Case C(ii): There exists at least one element that
  forms a rectangle.}
  In this case, the four robots forming the element
  forms a rectangle that is on a plane containing the principal axis of
  $\gamma(P(t-1))$ and one of the secondary axis.
  Hence, their destinations on $F$ are perturbed and form a rectangle
  on $F$. Thus the robots are not on a common line in $P(t)$. 
  
  \noindent{\bf Case C(iii): There exists at least one element that
  is on the principal axis.}
  In this case, the $\gamma(P(t-1))$-decomposition of $P(t-1)$ contains
  at least
  (a) one element forming a sphenoid,
  (b) one element forming a rectangle, or 
  (b) two elements on the secondary axis (two line elements). 
  In the first two cases, the robots are not on a common line in $P(t)$ 
  from Case C(i) and C(ii).
  In the last case, these two elements form a rhombus on $F$ and
  the robots are not on a common line in $P(t)$. 
  
  Consequently we have the lemma.
  \shortqed
\end{proof}

%========================================================================
\section{Concluding remark}
\label{sec:concl}

In this paper, 
we have investigated the plane formation problem 
for anonymous oblivious FSYNC robots 
in 3D-space. 
To analyze it,
we have defined the rotation group of a set of points 
in 3D-space in terms of its rotation group 
and we present a necessary and sufficient condition for 
the FSYNC robots to solve the plane formation problem.
We show a plane formation algorithm for oblivious FSYNC robots and
proved its correctness. 
We finally address the configuration space of the proposed algorithm.
The proposed algorithm is executed in a configuration where the robots
are not on a common plane.
During any execution, the robots do not reside on a common plane except
a terminal configuration.
This property is useful when the robots execute some existing algorithm
for 2D-space after the proposed plane formation algorithm,
because the configuration space of the plane formation algorithm and
that of the algorithm for 2D-plane are disjoint.
The progress of the composite algorithm is automatically guaranteed.

Another important result is related to the chirality of robots. 
Our first motivation is to apply existing algorithms for robots in
2D-space when the robots are put in 3D-space. 
However as we have shown in Section~\ref{subsec:landing},
when the robots with right-handed
$x$-$y$-$z$ local coordinate systems are put on a plane, 
they may not agree on the clockwise direction on the plane. 
It highlights the importance of distributed algorithms
without assuming chirality for robots on 2D-space. 

Since real systems work in a three dimensional space,
many natural problems would arise from practical applications.
The following is a partial list of open problems arising 
from the theory side:
\begin{enumerate}
\item
Understanding of the impact of chirality in the setting of this paper.
\item 
Understanding of the impact of visibility in the setting of this paper. 
\item
The general pattern formation problem for three dimensional space.
\item
Extensions to SSYNC and ASYNC robots.
\item
Extensions to arbitrary $d$ dimensional space. 
\end{enumerate}

\newpage
\appendix

\section{Property of rotation groups}

\label{app:rotation-groups}

\noindent{\bf Property~\ref{property:d2-principal}.~} 
{\it
Let $P \in {\cal P}_n^3$ be a set of points. 
If $D_2$ acts on $P$ and we cannot distinguish the principal axis of 
(an arbitrary embedding of) $D_2$, then $\gamma(P) \succ D_2$. 
}
\begin{proof}
Without loss of generality, 
we can assume that $x$-$y$-$z$ axes of the global coordinate system $Z_0$
 are the $2$-fold axes of 
 $D_2$.\footnote{There exists a translation consisting of
 rotation and translation that overlaps the $2$-fold axis of
 $\gamma(P)$ to the three axes.}
We define the octant according to $Z_0$ as shown in 
 Figure~\ref{fig:rot3-0} and Table~\ref{table:octant}.

\begin{table}[h]
\begin{center}
\caption{Definition of octant} 
\label{table:octant} 
\begin{tabular}[t]{|c|c|c|c|}
Number & $x$ & $y$ & $z$\\
\hline 
1 & + & + & + \\
2 & - & + & + \\
3 & - & - & + \\
4 & + & - & + \\
5 & + & + & - \\
6 & - & + & - \\
7 & - & - & - \\
8 & + & - & - \\
\end{tabular}
\end{center}
\end{table}

We consider the positions of points of $P$ in the first octant,  
which defines the positions of points of $P$ in the 
third, sixth, and the eighth octant by the rotations of $D_2$. 
The discussion also holds symmetrically in the second octant, 
that determines the positions of points in the 
fourth, fifth, and seventh octant. 

 We focus on a point $p \in P$ and
 depending on the position of $p$, we have the 
following five cases. 
\begin{itemize}
\item $p$ is on the $x$-axis (thus, the discussion follows for 
$y$-axis and $z$-axis, respectively). 
\item $p$ is on the $x$-$y$ plane (thus, the discussion follows for 
$y$-$z$ plane and $z$-$x$ plane, respectively). 
\item $p$ is on the line $x=y=z$. 
\item other cases. 
\end{itemize}
We will show that in any of the four cases, 
if we cannot recognize the principal axis, then 
we can rotate $P$ around the four $3$-fold axis $x=y=z$, 
$-x=y=z$, $-x=-y=z$, and $x=-y=z$.

\noindent{\bf Case A:~} When $p \in P$ is on the $x$-axis. 
Because $\gamma(P) = D_2$, we have a corresponding point on 
the negative $x$-axis (Figure~\ref{fig:rot3-1}). 
This allows us to recognize the $x$-axis from the $y$-axis and 
$z$-axis, hence $P$ should have corresponding points on 
$y$-axis and $z$-axis. 
In this case, we can rotate the corresponding six points 
around the four $3$-fold axes. 

\noindent{\bf Case B:~} When $p \in P$ is on the $x$-$y$ plane. 
First consider the case where a point $p \in P$ is on the line $x=y$. 
Because $\gamma(P) = D_2$, we have four corresponding points on 
the $x$-$y$ plane that forms a square (Figure~\ref{fig:rot3-2}). 
This allows us to recognize the $z$-axis from the other two 
axes, hence 
$y$-$z$ plane and $z$-$x$ plane also have the corresponding squares. 
Hence, the twelve points form a cuboctahedron, and 
we can rotate them around the four $3$-fold axes. 

When $p$ is not on the line $x=y$, 
because $\gamma(P) = D_2$, we have four corresponding points on 
the $x$-$y$ plane that forms a rectangle (Figure~\ref{fig:rot3-3}). 
This allows us to recognize the principal axis. 
In the same way as the above case, 
there are two rectangles on the $y$-$z$ plane and $z$-$x$ plane. 
The obtained polyhedron consists of $12$ vertices 
and we can rotate it around the four $3$-fold axes. 

\noindent{\bf Case C:~} When $p \in P$ is on the line $x=y=z$. 

Because $\gamma(P) = D_2$, we have four corresponding points in 
the third, sixth, and the eighth octant, 
that forms a regular tetrahedron (Figure~\ref{fig:rot3-4}). 
In this case, we can rotate the corresponding four points 
around the four $3$-fold axes. 

\noindent{\bf Case D:~} Other cases. 

For a point $p \in P$ in the first octant, 
because $\gamma(P) = D_2$, we have corresponding four points in 
the third, sixth, and the eighth octant, 
that forms a sphenoid (Figure~\ref{fig:rot3-5}). 
This allows us to recognize the $z$-axis from the others, 
hence $y$-axis and $x$-axis also have the corresponding sphenoids. 
The obtained polyhedron consists of $12$ vertices 
and we can rotate it around the four $3$-fold axes. 

Consequently when $D_2$ acts on $P$ 
but we cannot recognize the principal axis, 
we can rotate $P$ around the four $3$-fold axes. 
Thus $\gamma(P) \succeq T$. 

%%%%%%%%%%%%%%%%%%%%%%%%%%%%%%%%%%%%%%%%%%%%%%%%%%%%%%%%%%%%%%%%%%%%%%%%%%
\begin{figure}[t]
\centering 
\subfigure[] 
{\includegraphics[width=4cm]{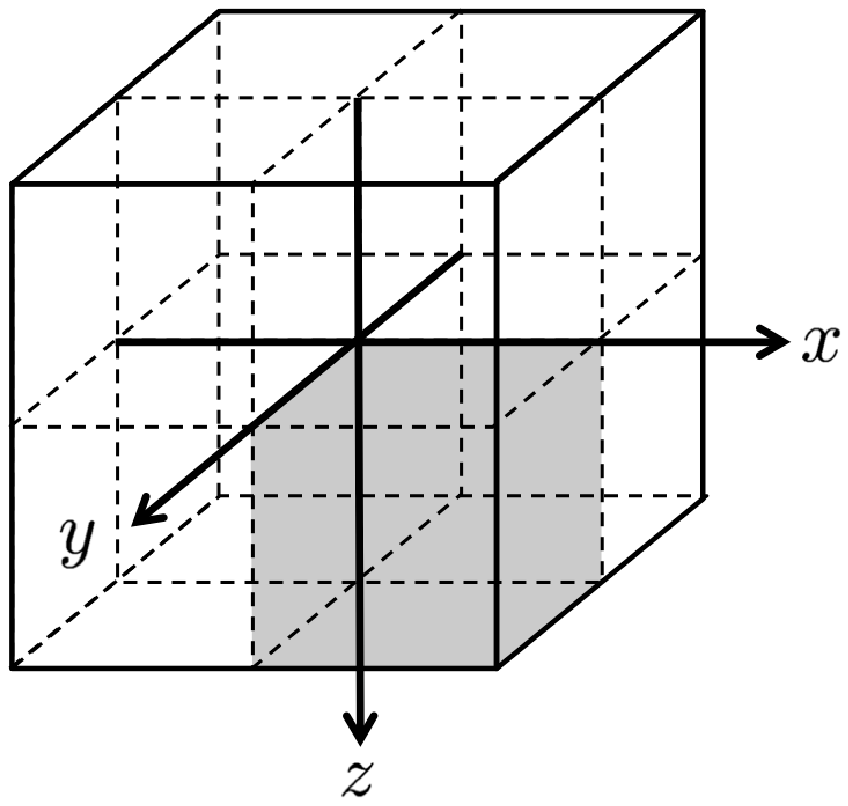}\label{fig:rot3-0}}
\hspace{3mm}
\subfigure[] 
{\includegraphics[width=4cm]{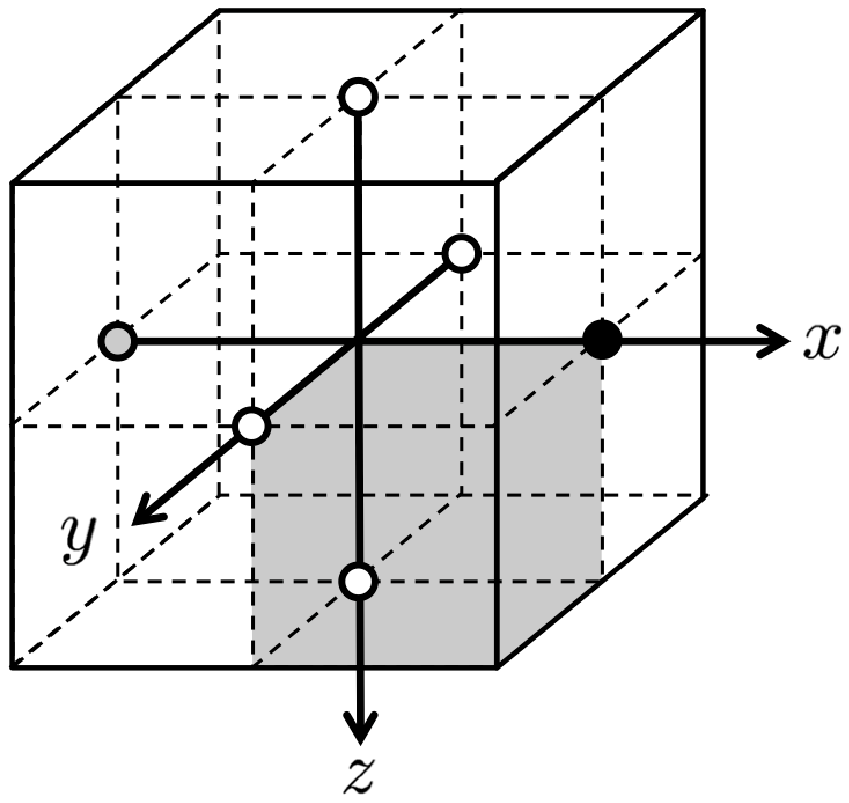}\label{fig:rot3-1}}
\hspace{3mm}
\subfigure[] 
{\includegraphics[width=4cm]{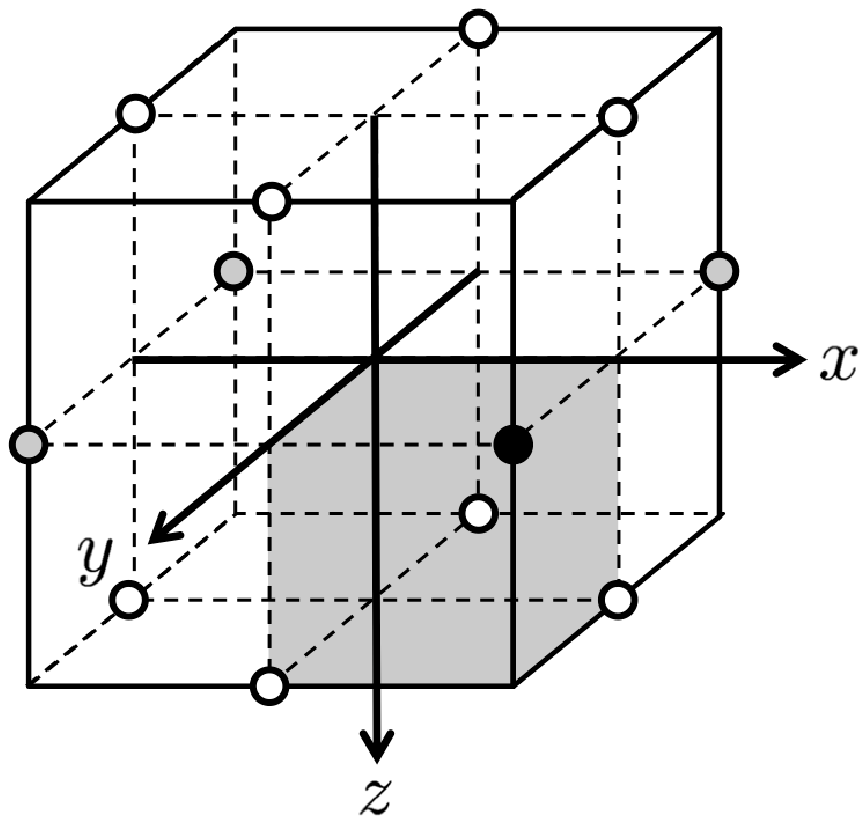}\label{fig:rot3-2}}
\hspace{3mm}
\subfigure[] 
{\includegraphics[width=4cm]{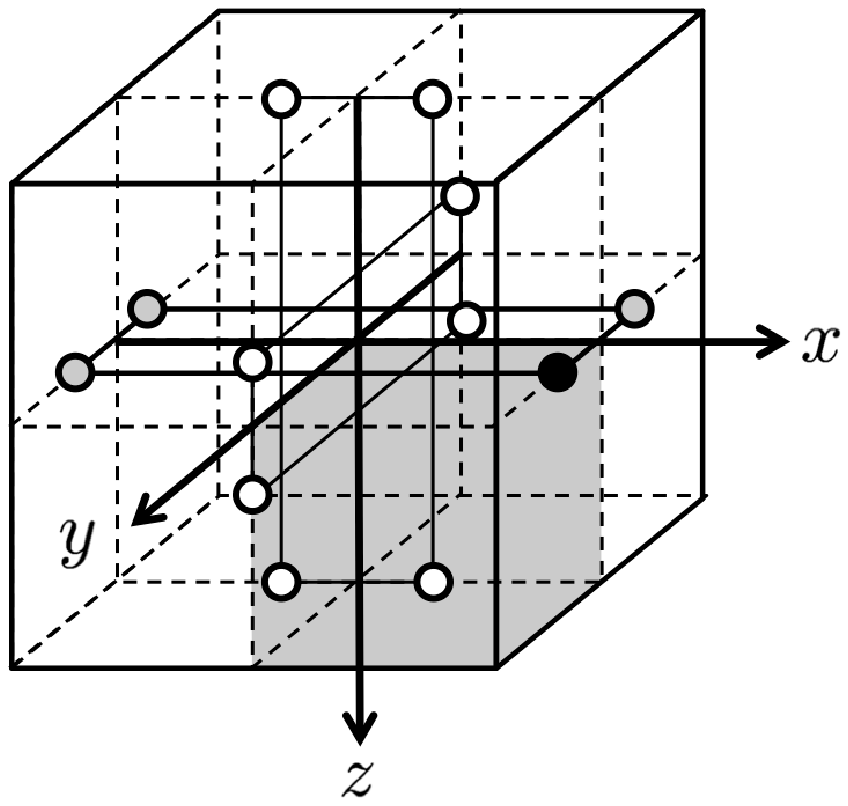}\label{fig:rot3-3}}
\hspace{3mm}
\subfigure[] 
{\includegraphics[width=4cm]{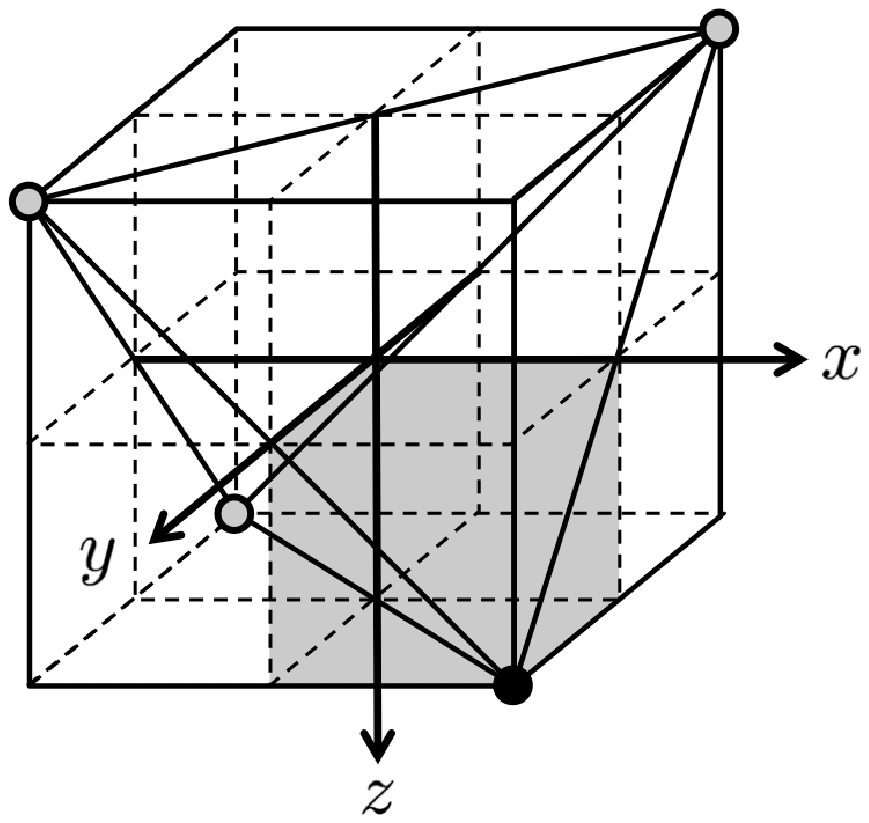}\label{fig:rot3-4}}
\hspace{3mm}
\subfigure[] 
{\includegraphics[width=4cm]{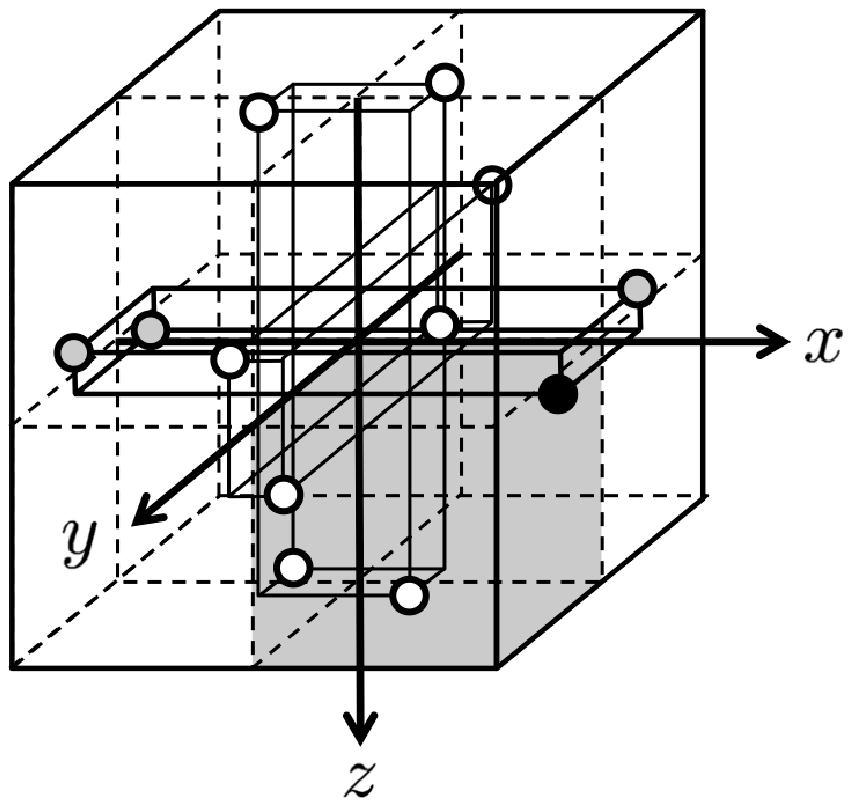}\label{fig:rot3-5}}
\caption{Position of a point of $P$ in the first octant, 
and the corresponding points generated by the $D_2$. 
The first octant is shown in the gray box in (a). 
The black circle is a point of $P$, and the gray circles are 
the points generated by $D_2$. 
The white circles are generated so that none of the three 
rotation axes is recognized.  }
\label{fig:rot3}
\end{figure}
%%%%%%%%%%%%%%%%%%%%%%%%%%%%%%%%%%%%%%%%%%%%%%%%%%%%%%%%%%%%%%%%%%%%%%%%%%

\shortqed 
\end{proof}

Clearly, Property~\ref{property:d2-principal} holds for
the robots since the above discussion 
does not depend on the local coordinate systems.


\begin{thebibliography}{99}

\bibitem{AP06}
N.~Agmon and D.~Peleg,
Fault-tolerant gathering algorithms for autonomous mobile
robots,
{\em SIAM J. Comput.}, 36, 1, pp.56--82. 
 
\bibitem{AOSY99}
H.~Ando, Y.~Oasa, I.~Suzuki, and M.~Yamashita,
Distributed memoryless point convergence algorithm for mobile
robots with limited visibility,
{\em IEEE Trans. Robotics and Automation}, 15, 5, pp.818--828,
1999. 

\bibitem{A88}
M.A.~Armstrong, Groups and symmetry,
Springer-Verlag New York Inc., 1988. 
	
\bibitem{BGT10}
Z.~Bouzid, M.~Gradinariu Potop-Butucaru, and S.~Tixeuil,
Optimal Byzantine-resilient convergence in uni-dimensional
robot networks,
{\em Theor. Comput. Sci.,} 411, pp.3154--3168, 2010. 

\bibitem{CFPS12} 
M.~Cieliebak, P.~Flocchini, G.~Prencipe, and N.~Santoro, 
Distributed computing by mobile robots: gathering, 
{\em SIAM J. Comput.,} 41, 4, pp.829--879, 2012. 
	
\bibitem{CP05}
R.~Cohen and D.~Peleg,
Convergence properties of the gravitational algorithm in
asynchronous robot systems,
{\em SIAM J. Comput.}, 34, pp.1516--1528, 2005. 
	
\bibitem{CP08}
R.~Cohen and D.~Peleg,
Convergence of autonomous mobile robots with inaccurate sensors
and movements,
{\em SIAM J. Comput.}, 38, 1, pp.276--302, 2008.

\bibitem{C73} 
H.S.M.~Coxeter,
Regular polytopes, 
Dover Publications, 1973. 

\bibitem{C97} 
P.~Cromwell, 
Polyhedra, 
University Press, 1997. 

\bibitem{DFPSY16}
S.~Das, P.~Flocchini, G.~Prencipe, N.~Santoro, and
M.~Yamashita,
Autonomous mobile robots with lights,
{\em Theor. Comput. Sci,} 609, pp.171--184, 2016. 

\bibitem{DFSY15}
S.~Das, P.~Flocchini, N.~Santoro, and M.~Yamashita,
Forming sequence of geometric patterns with oblivious mobile
robots,
{\em Distrib. Comput.}, 28, pp.131--145, 2015. 
	
\bibitem{D74}
E.W.~Dijkstra,
Self-stabilizing systems in spite of distributed control, 
{\em Comm. ACM}, 17, 11, pp.643--644, 1974.

\bibitem{EP09}
A.~Efrima and D.~Peleg,
Distributed algorithms for partitioning a swarm of autonomous
mobile robots,
{\em Theor. Comput. Sci.,} 410, pp.1355--1368, 2009. 
	
\bibitem{FPS12} 
P.~Flocchini, G.~Prencipe, and N.~Santoro,
Distributed Computing by Oblivious Mobile Robots,
Morgan \& Claypool, 2012.
	
\bibitem{FPSV14}
P.~Flocchini, G.~Prencipe, N.~Santoro, and 
G.~Viglietta, 
Distributed computing by mobile robots: 
Solving the uniform circle formation problem, 
{\em In Proc. of OPODIS'14}, pp.217-232, 2014. 
 
\bibitem{FPSW05}
P.~Flocchini, G.~Prencipe, N.~Santoro, and P.~Widmayer,
Gathering of asynchronous robots with limited visibility,
{\em Theor. Comput. Sci.}, 337, pp.147--168, 2005. 

\bibitem{FPSW08} 
P.~Flocchini, G.~Prencipe, N.~Santoro, and P.~Widmayer, 
Arbitrary pattern formation by asynchronous, anonymous, oblivious 
robots, 
{\em Theor. Comput. Sci.,} 407, pp.412--447, 2008.

\bibitem{FYOKY15} 
N.~Fujinaga, Y.~Yamauchi, H. Ono, S.~Kijima, and M.~Yamashita, 
Pattern formation by oblivious asynchronous mobile robots,
{\em SIAM J. Comput.}, 44, 3, pp.740--785, 2015. 

\bibitem{IKY14}
T.~Izumi, S.~Kamei, and Y.~Yamauchi,
Approximation algorithms for the set cover formation by
oblivious mobile robots,
{\em In Proc. of OPODIS 2014}, pp.233--247, 2014. 
	
\bibitem{ISKI12}
T.~Izumi, S.~Souissi, Y.~Katayama, N.~Inuzuka,
X.~D\'efago, K.~Wada, and M.~Yamashita,
The gathering problem for two oblivious robots with
unreliable compasses,
{\em SIAM J. Comput.}, 41, 1, pp.26--46, 2012. 

\bibitem{P05}
D.~Peleg,
Distributed coordination algorithms for mobile robot swarms:
New directions and challenges,
{\em In Proc. of IWDC 2005}, pp.1--12, 2005. 
	
\bibitem{SDY09}
S.~Souissi, X.~D\'efago, and M.~Yamashita,
Using eventually consistent compasses to gather memory-less
mobile robots with limited visibility,
{\em ACM Trans. Autonomous Adaptive Systems}, 4, 1, 9, 2009.
	
\bibitem{SY99} 
I.~Suzuki and M.~Yamashita, 
Distributed anonymous mobile robots: Formation of geometric patterns, 
{\em SIAM J. Comput.}, 28, 4, pp.1347--1363, 1999. 

\bibitem{YK96}
M.~Yamashita and T.~Kameda,
Computing on anonymous networks: Part I-Characterizing the solvable cases,
{\em IEEE Trans. Parallel Distrib. Syst}, 
7, 1, pp.69-89, 1996.

\bibitem{YS10} 
M.~Yamashita and I.~Suzuki, 
Characterizing geometric patterns formable by oblivious anonymous mobile 
robots, 
{\em Theor. Comput. Sci.}, 411, pp.2433--2453, 2010. 

\bibitem{YY13}
Y.~Yamauchi and M.~Yamashita, 
Pattern formation by mobile robots with limited visibility,
{\em In Proc. of SIROCCO 2013}, pp.201--212, 2013. 
	
\bibitem{YY14}
Y.~Yamauchi and M.~Yamashita,
Randomized pattern formation algorithm for asynchronous
oblivious mobile robots,
{\em In Proc. of DISC 2014}, pp.137--151, 2014. 

\end{thebibliography}
\end{document}